\documentclass[11pt]{article}
\usepackage{float}
\usepackage{makeidx}
\usepackage{amssymb}
\usepackage{amsmath}
\usepackage{amsfonts}
\usepackage{lscape}
\usepackage[onehalfspacing]{setspace}
\usepackage{geometry}
\usepackage{chicago}
\usepackage{chbibref}
\usepackage[toc,page,title,titletoc,header]{appendix}
\usepackage{color}
\usepackage{epsfig}
\usepackage{hyperref}

\newtheorem{theorem}{Theorem}

\newtheorem{assumption}{Assumption}

\newtheorem{lemma}{Lemma}

\newtheorem{remark}{Remark}

\newenvironment{proof}[1][Proof]{\noindent \textbf{#1.} }{\  \rule{0.5em}{0.5em}}
\input{tcilatex}
\begin{document}

\title{Pooled Bewley Estimator of Long Run Relationships in Dynamic
Heterogenous Panels\thanks{%
The views expressed in this paper are those of the authors and do not
necessarily reflect those of the Federal Reserve Bank of Dallas or the
Federal Reserve System. Monte Carlo simulations in this paper were computed
using computational resources provided by the Big-Tex High Performance
Computing Group at the Federal Reserve Bank of Dallas. We are grateful for
comments by two referees and the editor.}}
\author{Alexander Chudik\thanks{%
Email: alexander.chudik@gmail.com.} \\
Federal Reserve Bank of Dallas \and M. Hashem Pesaran\thanks{%
Email: mhp1@cam.ac.uk.} \\
University of Southern California, USA and Trinity College, Cambridge, UK
\and Ron P. Smith\thanks{%
Email: r.smith@bbk.ac.uk.} \\
Birkbeck, University of London, United Kingdom}
\date{October 30, 2023}
\maketitle

\begin{abstract}
Using a transformation of the autoregressive distributed lag model due to
Bewley, a novel pooled Bewley (PB) estimator of long-run coefficients for
dynamic panels with heterogeneous short-run dynamics is proposed. The PB
estimator is directly comparable to the widely used Pooled Mean Group (PMG)
estimator, and is shown to be consistent and asymptotically normal. Monte
Carlo simulations show good small sample performance of PB compared to the
existing estimators in the literature, namely PMG, panel dynamic OLS
(PDOLS), and panel fully-modified OLS (FMOLS). Application of two
bias-correction methods and a bootstrapping of critical values to conduct
inference robust to cross-sectional dependence of errors are also
considered. The utility of the PB estimator is illustrated in an empirical
application to the aggregate consumption function.\smallskip


\bigskip \bigskip
\bigskip 
\end{abstract}

\thispagestyle{empty}\pagebreak

\pagenumbering{arabic}%

\section{Introduction}

\doublespacing%

Estimation of cointegrating relationships in panels with heterogeneous
short-run dynamics is important for empirical research in open economy
macroeconomics as well as in other fields in economics. Existing
single-equation panel estimators in the literature are panel Fully Modified
OLS (FMOLS) by 
\citeANP{Pedroni1996} (\citeyearNP{Pedroni1996}, \citeyearNP{Pedroni2001}, \citeyearNP{Pedroni2001ReStat})%
, panel Dynamic OLS (PDOLS) by 
\citeN{MarkSul2003}%
, and the likelihood based Pooled Mean Group (PMG) estimator by 
\citeN{PesaranShinSmith1999}%
. Multi-equation (system) approach by 
\citeN{Breitung2005}%
, and the related system PMG approach by 
\citeN{ChudikPesaranSmith2022}
are another contributions in the literature on estimating cointegrating
vectors in a panel context. In this paper, we propose a pooled Bewley (PB)
estimator of long-run relationships, relying on the Bewley transform of an
autoregressive distributed lag (ARDL) model (%
\citeANP{Bewley1979}, \citeyearNP{Bewley1979}%
). See also 
\citeANP{WickensBreusch1988} (\citeyearNP{WickensBreusch1988})
for a discussion of the Bewley transform.

Our setting is the same as that of 
\citeN{PesaranShinSmith1999}%
. Under this setting, any short-run feedbacks between the outcome variable ($%
y$) and regressors ($x$) are allowed, but the direction of the long run
causality is assumed to go from $x$ to $y$. Hence, as with the PMG, the PB
estimator allows for heterogeneity in short-run feedbacks, but restricts the
direction of long run causality. The PB estimator is computed analytically,
and does not rely on numerical maximization of the likelihood function that
underlies the PMG estimation. We derive the asymptotic distribution of the
PB estimator when the cross-section dimension ($n$) and the time dimension ($%
T$) diverge to infinity jointly such that $n=\Theta \left( T^{\theta
}\right) $, for $0<\theta <2$, where we use the notation $\Theta \left(
.\right) $ to denote the same order of magnitude asymptotically, namely if $%
\left\{ f_{s}\right\} _{s=1}^{\infty }$ and $\left\{ g_{s}\right\}
_{s=1}^{\infty }$ are both positive sequences of real numbers, then $%
f_{s}=\ominus \left( g_{s}\right) $ if there exists $S_{0}\geq 1$ and
positive finite constants $C_{0}$ and $C_{1}$, such that $\inf_{s\geq
S_{0}}\left( f_{s}/g_{s}\right) \geq C_{0},$ and $\sup_{s\geq S_{0}}\left(
f_{s}/g_{s}\right) \leq C_{1}$. Our asymptotic analysis is an advance over
the theoretical results currently available for PMG, PDOLS and FMOLS
estimators where it is assumed that $n$ is small relative to $T$ (which
corresponds to the case where $\theta $ is close to zero).

How well individual estimators work in samples of interest in practice where 
$n$ and/or $T$ are often less than 50 is a different matter, which we shed
light on using Monte Carlo experiments. Monte Carlo evidence shows PB
estimator can be superior to PMG, PDOLS and FMOLS, in terms of its overall
precision as measured by the Root Mean Square Error (RMSE), and in terms of
accuracy of inference as measured by size distortions. These experiments
reveal PB is a useful addition to the literature.

Monte Carlo evidence also shows that the time dimension is very important
for the performance of these estimators, and that all of the four estimators
under consideration suffer from the same two drawbacks: small sample bias
and size distortion. Although the size distortions are found to be less
serious for the PB estimator in our experiments, all four estimators exhibit
notable over-rejections in sample sizes relevant in practice. In addition,
all four estimators (perhaps unsurprisingly) suffer from bias in finite
samples, albeit a rather small one. Both drawbacks diminish as $T$ is
increased relative to $n$.

To conduct reliable inference regardless of the cross-sectional dependence
of errors, we make use of the sieve wild bootstrap procedure. To accommodate
cross-sectional dependence, we resample the cross-section vectors of
residuals, an idea that was originally proposed by 
\citeN{MaddalaWu1999}%
. We found the sieve wild bootstrap procedure to be remarkably effective for
all four estimators, regardless of cross-sectional dependence of errors, and
we therefore recommend using it in empirical research.

Regarding the small sample bias, we consider the application of two
bias-correction methods taken from the literature, relying either on
split-panel jackknife (%
\citeANP{DhaeneJochmansy2015}, \citeyearNP{DhaeneJochmansy2015}%
) or sieve wild bootstrap approaches. In contrast to split-panel approaches
in panels without stochastic trends, as, for instance, considered by of 
\citeN{DhaeneJochmansy2015}
or 
\citeN{ChudikPesaranYang2018}%
, in this paper we need to combine the full sample and half-panel subsamples
using different weighting due to the fact that the rate of convergence of
the estimators of long run coefficients is faster, at the rate of $T\sqrt{n}$%
, as compared to the standard rate of $\sqrt{nT}$. We find that both of
these approaches can be helpful in reducing the bias (for all four
estimators). However, given that the bias is small to begin with, the value
of bias correction methods is limited.

The relevance of choosing a particular estimation approach is illustrated in
the context of a consumption function application for OECD economies taken
from 
\citeN{PesaranShinSmith1999}%
. This application shows that quite a different conclusion would be reached
when using PB estimator, which does not reject the zero long-run coefficient
on inflation, in line with the long-run neutrality of monetary policy,
whereas the PMG estimator results in a highly statistically significant
negative long-run coefficient. Estimates of the long-run coefficient on real
income are less diverse across estimators, but the inference on whether a
unit long-run coefficient on real income (as suggested by balanced growth
path models in the literature) can be rejected or not depends on the choice
of a particular estimator.

The remainder of this paper is organized as follows. Section \ref{Model}
presents the model and assumptions, introduces the PB estimator, and
provides asymptotic results. Application of bias correction methods and
bootstrapping critical values are also discussed in Section \ref{Model}.
Section \ref{MC} presents Monte Carlo evidence. Section \ref{EA} revisits
the aggregate consumption function empirical application in 
\citeN{PesaranShinSmith1999}%
. Section \ref{CON} concludes. Mathematical derivations and proofs are
provided in Appendix A. Details on the implementation of individual
estimators and bootstrapping, and additional Monte Carlo results are
provided in Appendix B.

\section{Pooled Bewley estimator of long-run relationships\label{Model}}

We adopt the same setting as in 
\citeN{PesaranShinSmith1999}%
, and consider the following illustrative model%
\begin{eqnarray}
\Delta y_{it} &=&c_{i}-\alpha _{i}\left( y_{i,t-1}-\beta x_{i,t-1}\right)
+u_{y,it}\text{,}  \label{y} \\
\Delta x_{it} &=&u_{x,it}\text{,}  \label{x}
\end{eqnarray}%
for $i=1,2,...,n$, and $t=1,2,...,T$. For expositional clarity and
notational simplicity, we focus on a single regressor and one lag, but it is
understood that our analysis is applicable to multiple lags of $\Delta 
\mathbf{z}_{it}=\left( \Delta y_{it},\Delta x_{it}\right) ^{\prime }$
entering both equations (\ref{y})-(\ref{x}), and the approach is also
applicable to multiple $x_{it}$'s with a single long run relationship. We
consider the following assumptions:

\begin{assumption}
\textbf{(Coefficients) }\label{As1} $\sup_{i}\left\vert 1-\alpha
_{i}\right\vert <1$.
\end{assumption}

\begin{assumption}
\textbf{(shocks)}\label{As2} $u_{x,it}\sim IID\left( 0,\sigma
_{xi}^{2}\right) $, and $u_{y,it}$ is given by%
\begin{equation}
u_{y,it}=\delta _{i}u_{x,it}+v_{it}\text{,}  \label{uy}
\end{equation}%
for all $i$ and $t$, where $v_{it}\sim IID\left( 0,\sigma _{vi}^{2}\right) $%
, and $u_{x,it}$ is independently distributed of $v_{i^{\prime }t^{\prime }}$
for all $i$,$i^{\prime },t$, and $t^{\prime }$. In addition, $%
\sup_{i,t}E\left\vert v_{it}\right\vert ^{16}<K$ and $\sup_{i,t}E\left\vert
u_{x,it}\right\vert ^{8}<K$, and $\lim_{n\rightarrow \infty
}n^{-1}\sum_{i=1}^{n}\sigma _{xi}^{2}=\sigma _{x}^{2}>0$ and $\lim
{}_{n\rightarrow \infty }n^{-1}\sum_{i=1}^{n}\sigma _{xi}^{2}\sigma
_{vi}^{2}/\left( 6\alpha _{i}^{2}\right) =\omega _{v}^{2}>0$ exist.
\end{assumption}

\begin{assumption}
\textbf{(Initial values and deterministic terms)}\label{As3} The initial
values, $\mathbf{z}_{i,0}=\left( y_{i,0},x_{i,0}\right) ^{\prime }$, follow
the process%
\begin{equation*}
\mathbf{z}_{i0}=\mathbf{\mu }_{i}+\mathbf{C}_{i}^{\ast }\left( L\right) 
\mathbf{u}_{0}\text{,}
\end{equation*}%
for all $i$ and $t$, and $c_{i}=$ $\alpha _{i}\mu _{i,1}-\alpha _{i}\beta
\mu _{i,2}$ for all $i$, where $\mathbf{u}_{0}=\left(
u_{y,i,0},u_{x,i,0}\right) ^{\prime }$, $\mathbf{\mu }_{i}=\left( \mu
_{i,1},\mu _{i,2}\right) ^{\prime }$, $\left\Vert \mathbf{\mu }%
_{i}\right\Vert <K$, and $\mathbf{C}_{i}^{\ast }\left( L\right) $ is defined
in Section \ref{A1} in Appendix A.
\end{assumption}

\begin{remark}
\label{Ralpha}Assumption \ref{As1} requires that $\alpha _{i}\neq 0$ for 
\textit{all} $i$. In contrast PMG allows $\alpha _{i}=0$ for some (but not
all) units. Although the estimator works for $\alpha _{i}<2,$ empirically, $%
\alpha _{i}<1$ is likely to be the relevant case. \ 
\end{remark}

\begin{remark}
Assumption \ref{As2} allows for $u_{x,it}$ to be correlated with $u_{y,it}$.
Cross-section dependence of $\mathbf{u}_{it}$ is ruled out. Assumption \ref%
{As3} (together with the remaining assumptions) ensure that $\Delta \mathbf{z%
}_{it}$ and $\left( y_{it}-\beta x_{it}\right) $ are covariance stationary.
\end{remark}

\begin{remark}
In comparing our assumptions with the rest of the literature, it should be
noted that the rest of the literature does not consider the case of joint
convergence $n,T\rightarrow _{j}\infty $, but only the case where $n$ is
fixed as $T\rightarrow \infty $. Joint asymptotics typically requires more
stringent assumptions on the errors and restrictions on the relative
expansion rates of $n$ and $T$.
\end{remark}

Substituting first (\ref{uy}) for $u_{y,it}$ in (\ref{y}), and then
substituting $u_{x,it}=\Delta x_{it}$, we obtain the following ARDL
representation for $y_{it}$%
\begin{equation}
\Delta y_{it}=c_{i}-\alpha _{i}\left( y_{i,t-1}-\beta x_{i,t-1}\right)
+\delta _{i}\Delta x_{it}+v_{it}\text{.}  \label{ardl}
\end{equation}%
The pooled Bewley estimator takes advantage of the Bewley transform (%
\citeANP{Bewley1979}, \citeyearNP{Bewley1979}%
). Subtracting $\left( 1-\alpha _{i}\right) y_{it}$ from both sides of (\ref%
{ardl}) and re-arranging, we have%
\begin{equation*}
\alpha _{i}y_{it}=c_{i}-\left( 1-\alpha _{i}\right) \Delta y_{it}+\alpha
_{i}\beta x_{it}+\delta _{i}\Delta x_{it}+v_{it},
\end{equation*}%
or (noting that $\alpha _{i}>0$ for all $i$ and multiplying the equation
above by $\alpha _{i}^{-1}$)%
\begin{equation}
y_{it}=\alpha _{i}^{-1}c_{i}+\beta x_{it}+\mathbf{\psi }_{i}^{\prime }\Delta 
\mathbf{z}_{it}+\alpha _{i}^{-1}v_{it},  \label{b}
\end{equation}%
where $\Delta \mathbf{z}_{it}=\left( \Delta y_{it},\Delta x_{it}\right)
^{\prime }$, and $\mathbf{\psi }_{i}=\left( -\frac{1-\alpha _{i}}{\alpha _{i}%
},\frac{\delta _{i}}{\alpha _{i}}\right) ^{\prime }$. Further, stacking (\ref%
{b}) for $t=1,2,...,T$, we have%
\begin{equation}
\mathbf{y}_{i}=\alpha _{i}^{-1}c_{i}\mathbf{\tau }_{T}+\mathbf{x}_{i}\beta
+\Delta \mathbf{Z}_{i}\mathbf{\psi }_{i}+\alpha _{i}^{-1}\mathbf{v}_{i}\text{%
,}  \label{i}
\end{equation}%
where $\mathbf{y}_{i}=\left( y_{i1},y_{i2},...,y_{iT}\right) ^{\prime }$, $%
\mathbf{x}_{i}=\left( x_{i1},x_{i2},...,x_{iT}\right) ^{\prime }$, $\Delta 
\mathbf{Z}_{i}=\left( \Delta \mathbf{z}_{i1}^{\prime },\Delta \mathbf{z}%
_{i2}^{\prime },...,\Delta \mathbf{z}_{iT}^{\prime }\right) ^{\prime }$, $%
\mathbf{v}_{i}=\left( v_{i,1},v_{i,2},...,v_{i,T}\right) ^{\prime }$, and $%
\mathbf{\tau }_{T}$ is $T\times 1$ vector of ones. Define projection matrix $%
\mathbf{M}_{\tau }=\mathbf{I}_{T}-T^{-1}\mathbf{\tau }_{T}\mathbf{\tau }%
_{T}^{\prime }$. This projection matrix subtracts the period average. Let $%
\mathbf{\tilde{y}}_{i}=\left( \tilde{y}_{i1},\tilde{y}_{i2},...,\tilde{y}%
_{iT}\right) ^{\prime }=\mathbf{M}_{\tau }\mathbf{y}_{i}$, and similarly $%
\mathbf{\tilde{x}}_{i}=\left( \tilde{x}_{i1},\tilde{x}_{i2},...,\tilde{x}%
_{iT}\right) ^{\prime }=\mathbf{M}_{\tau }\mathbf{x}_{i}$, $\Delta \mathbf{%
\tilde{Z}}_{i}=\mathbf{M}_{\tau }\Delta \mathbf{Z}_{i}$, and $\mathbf{\tilde{%
v}}_{i}=\mathbf{M}_{\tau }\mathbf{v}_{i}$. Multiplying (\ref{i}) by $\mathbf{%
M}_{\tau }$, we have%
\begin{equation*}
\mathbf{\tilde{y}}_{i}=\mathbf{\tilde{x}}_{i}\beta +\Delta \mathbf{\tilde{Z}}%
_{i}\mathbf{\psi }_{i}+\alpha _{i}^{-1}\mathbf{\tilde{v}}_{i}\text{.}
\end{equation*}%
Now consider the matrix of instruments 
\begin{equation}
\mathbf{\tilde{H}}_{i}=\left( \mathbf{\tilde{y}}_{i,-1},\mathbf{\tilde{x}}%
_{i},\mathbf{\tilde{x}}_{i,-1}\right) =\mathbf{M}_{\tau }\mathbf{H}_{i}\text{%
, }\mathbf{H}_{i}=\left( \mathbf{y}_{i,-1},\mathbf{x}_{i},\mathbf{x}%
_{i,-1}\right) \text{,}  \label{z}
\end{equation}%
where $\mathbf{y}_{i,-1}=\left( y_{i,1},y_{i,1},...,y_{i,T-1}\right)
^{\prime }$ is the data vector on the first lag of $y_{it}$, similarly $%
\mathbf{x}_{i,-1}=\left( x_{i,1},x_{i,1},...,x_{i,T-1}\right) ^{\prime }$.
The PB estimator of $\beta $ is given by%
\begin{equation}
\hat{\beta}=\left( \sum_{i=1}^{n}\mathbf{\tilde{x}}_{i}^{\prime }\mathbf{M}%
_{i}\mathbf{\tilde{x}}_{i}\right) ^{-1}\left( \sum_{i=1}^{n}\mathbf{\tilde{x}%
}_{i}^{\prime }\mathbf{M}_{i}\mathbf{\tilde{y}}_{i}\right) \text{,}
\label{pmg}
\end{equation}%
where 
\begin{equation}
\mathbf{M}_{i}=\mathbf{P}_{i}-\mathbf{P}_{i}\Delta \mathbf{\tilde{Z}}%
_{i}\left( \Delta \mathbf{\tilde{Z}}_{i}^{\prime }\mathbf{P}_{i}\Delta 
\mathbf{\tilde{Z}}_{i}\right) ^{-1}\Delta \mathbf{\tilde{Z}}_{i}^{\prime }%
\mathbf{P}_{i}\text{,}  \label{mi}
\end{equation}%
and%
\begin{equation}
\mathbf{P}_{i}=\mathbf{\tilde{H}}_{i}\left( \mathbf{\tilde{H}}_{i}{}^{\prime
}\mathbf{\tilde{H}}_{i}\right) ^{-1}\mathbf{\tilde{H}}_{i}^{\prime }\text{,}
\label{pi}
\end{equation}%
is the projection matrix associated with $\mathbf{\tilde{H}}_{i}$.

In addition to Assumptions \ref{As1}-\ref{As3}, we also require the
following high-level conditions to hold in the derivations of the asymptotic
distribution of the PB estimator under the joint asymptotics $n,T\rightarrow
\infty $.

\begin{assumption}
\label{As4} Let $\mathbf{M}_{\tau }=\mathbf{I}_{T}-T^{-1}\mathbf{\tau }_{T}%
\mathbf{\tau }_{T}^{\prime },$ $\mathbf{\tilde{y}}_{i}=\mathbf{M}_{\tau }%
\mathbf{y}_{i}$, $\mathbf{\tilde{x}}_{i}=\mathbf{M}_{\tau }\mathbf{x}_{i}$, 
\newline
$\Delta \mathbf{\tilde{Z}}_{i}=\mathbf{M}_{\tau }\Delta \mathbf{Z}_{i}=%
\mathbf{M}_{\tau }\left( \Delta \mathbf{z}_{i1}^{\prime },\Delta \mathbf{z}%
_{i2}^{\prime },...,\Delta \mathbf{z}_{iT}^{\prime }\right) ^{\prime }$,
where $\Delta \mathbf{z}_{it}=\left( \Delta y_{it},\Delta x_{it}\right)
^{\prime }$. Then there exists $T_{0}\in \mathbb{N}$ such that the following
conditions are satisfied:

\begin{enumerate}
\item[($i$)] $\sup_{i\in \mathbb{N}\text{, }T>T_{0}}E\left[ \lambda _{\min
}^{-2}\left( \mathbf{B}_{iT}\right) \right] <K$, where $\mathbf{B}%
_{iT}=\Delta \mathbf{\tilde{Z}}_{i}^{\prime }\mathbf{P}_{i}\Delta \mathbf{%
\tilde{Z}}_{i}/T$, $\mathbf{P}_{i}$ is given by (\ref{pi}).

\item[($ii$)] $\sup_{i\in \mathbb{N}\text{, }T>T_{0}}E\left[ \lambda _{\min
}^{-2}\left( \mathbf{A}_{T}\mathbf{\tilde{H}}_{i}^{\ast \prime }\mathbf{%
\tilde{H}}_{i}^{\ast }\mathbf{A}_{T}\right) \right] <K$, where $\mathbf{H}%
_{i}^{\ast }=\left( \mathbf{\tilde{x}}_{i},\Delta \mathbf{\tilde{x}}_{i},%
\mathbf{\tilde{\xi}}_{i,-1}\right) $,%
\begin{equation*}
\mathbf{A}_{T}=\left( 
\begin{array}{ccc}
T^{-1} & 0 & 0 \\ 
0 & T^{-1/2} & 0 \\ 
0 & 0 & T^{-1/2}%
\end{array}%
\right) \text{,}
\end{equation*}%
$\mathbf{\tilde{\xi}}_{i,-1}=\left( \tilde{\xi}_{i,0},\tilde{\xi}_{i,1},...,%
\tilde{\xi}_{i,T-1}\right) ^{\prime }$, $\tilde{\xi}_{i,t-1}=\tilde{y}%
_{i,t-1}-\beta \tilde{x}_{i,t-1}.$
\end{enumerate}
\end{assumption}

\begin{remark}
Under Assumptions \ref{As1}-\ref{As3} (and without Assumption \ref{As4}), we
have $\limfunc{plim}_{T\rightarrow \infty }\mathbf{B}_{i,T}=\mathbf{B}_{i}$,
where $\mathbf{B}_{i}$ is nonsingular (see Lemma \ref{lh2} in Appendix A).
Similarly, it can be shown that Assumptions \ref{As1}-\ref{As3} are
sufficient for $\limfunc{plim}_{T\rightarrow \infty }\mathbf{A}_{T}\mathbf{%
\tilde{H}}_{i}^{\ast \prime }\mathbf{\tilde{H}}_{i}^{\ast }\mathbf{A}_{T}$
to exist and to be nonsingular. However, these results are not sufficient
for the moments of $\left\Vert \mathbf{B}_{i}^{-1}\right\Vert $ and $%
\left\Vert \left( \mathbf{A}_{T}\mathbf{H}_{i}^{\ast \prime }\mathbf{H}%
_{i}^{\ast }\mathbf{A}_{T}\right) ^{-1}\right\Vert $ to exist, which we
require for the derivations of the asymptotic distribution of the PB
estimator. This is ensured by Assumption \ref{As4}.
\end{remark}

\subsection{Asymptotic results}

Substituting $\mathbf{\tilde{y}}_{i}=\mathbf{\tilde{x}}_{i}\beta +\Delta 
\mathbf{\tilde{Z}}_{i}\mathbf{\psi }_{i}+\alpha _{i}^{-1}\mathbf{\tilde{v}}%
_{i}$ in (\ref{pmg}), and using $\mathbf{M}_{i}\Delta \mathbf{\tilde{Z}}_{i}=%
\mathbf{0}$, we have%
\begin{equation}
T\sqrt{n}\left( \hat{\beta}-\beta \right) =\left( \frac{1}{n}\sum_{i=1}^{n}%
\frac{\mathbf{\tilde{x}}_{i}^{\prime }\mathbf{M}_{i}\mathbf{\tilde{x}}_{i}}{%
T^{2}}\right) ^{-1}\left( \frac{1}{\sqrt{n}}\sum_{i=1}^{n}\frac{\mathbf{%
\tilde{x}}_{i}^{\prime }\mathbf{M}_{i}\mathbf{\tilde{v}}_{i}}{T\alpha _{i}}%
\right) \text{.}  \label{ts1}
\end{equation}%
Consider the first term on the right side of (\ref{ts1}) first. Since $%
\mathbf{M}_{i}$ is an orthogonal projection matrix, $\mathbf{\tilde{x}}%
_{i}^{\prime }\mathbf{M}_{i}\mathbf{\tilde{x}}_{i}/T^{2}$ is bounded by $%
\mathbf{\tilde{x}}_{i}^{\prime }\mathbf{\tilde{x}}_{i}/T^{2}$. The second
moments of $\mathbf{\tilde{x}}_{i}^{\prime }\mathbf{\tilde{x}}_{i}/T^{2}$
are bounded, and, in addition, $\mathbf{\tilde{x}}_{i}^{\prime }\mathbf{M}%
_{i}\mathbf{\tilde{x}}_{i}/T^{2}$ is cross-sectionally independent. It
follows that $\frac{1}{n}\sum_{i=1}^{n}\mathbf{\tilde{x}}_{i}^{\prime }%
\mathbf{M}_{i}\mathbf{\tilde{x}}_{i}/T^{2}$ converges to a constant, which
we denote by $\omega _{x}^{2}$, as $n,T\rightarrow \infty $. Lemma \ref{slp}
in Appendix A establishes the expression for $\omega _{x}^{2}=\sigma
_{x}^{2}/6$, where $\sigma _{x}^{2}=\lim_{n\rightarrow \infty
}n^{-1}\sum_{i=1}^{n}\sigma _{xi}^{2}$, but the specific expression for $%
\omega _{x}^{2}$ is not relevant for the inference approach that we adopt
below. Consider the second term of (\ref{ts1}) next,%
\begin{eqnarray}
\frac{1}{\sqrt{n}}\sum_{i=1}^{n}\frac{\mathbf{\tilde{x}}_{i}^{\prime }%
\mathbf{M}_{i}\mathbf{\tilde{v}}_{i}}{\alpha _{i}T} &=&\frac{1}{\sqrt{n}}%
\sum_{i=1}^{n}\left[ \frac{\mathbf{\tilde{x}}_{i}^{\prime }\mathbf{M}_{i}%
\mathbf{\tilde{v}}_{i}}{\alpha _{i}T}-E\left( \frac{\mathbf{\tilde{x}}%
_{i}^{\prime }\mathbf{M}_{i}\mathbf{\tilde{v}}_{i}}{\alpha _{i}T}\right) %
\right]  \label{sdab} \\
&&+\frac{1}{\sqrt{n}}\sum_{i=1}^{n}E\left( \frac{\mathbf{\tilde{x}}%
_{i}^{\prime }\mathbf{M}_{i}\mathbf{\tilde{v}}_{i}}{\alpha _{i}T}\right) 
\text{,}  \notag
\end{eqnarray}%
The term in the square brackets has zero mean and is independently
distributed over $i$. For the asymptotic distribution to be correctly
centered we need%
\begin{equation}
\frac{1}{\sqrt{n}}\sum_{i=1}^{n}E\left( \frac{\mathbf{\tilde{x}}_{i}^{\prime
}\mathbf{M}_{i}\mathbf{\tilde{v}}_{i}}{\alpha _{i}T}\right) \rightarrow 0%
\text{,}  \label{bp0}
\end{equation}%
as $n$ and $T\rightarrow \infty $. This condition holds so long as $n=\Theta
\left( T^{\theta }\right) $ for some $0<\theta <2$. See Lemma \ref{leb} for
a proof. The asymptotic distribution of the first term in (\ref{sdab}) is in
turn established by Lemma \ref{lad}, see (\ref{rcd}). The following theorem
now follows for the asymptotic distribution of $\hat{\beta}$.

\begin{theorem}
\label{T1}Let $\left( y_{it},x_{it}\right) $ be generated by model (\ref{y}%
)-(\ref{x}), suppose Assumptions \ref{As1}-\ref{As4} hold, and $%
n,T\rightarrow \infty $ such that $n=\Theta \left( T^{\theta }\right) $, for
some $0<\theta <2$. Consider the PB estimator $\hat{\beta}$ given by (\ref%
{pmg}). Then,%
\begin{equation}
T\sqrt{n}\left( \hat{\beta}-\beta \right) \rightarrow _{d}N\left( 0,\Omega
\right) \text{, }\Omega =\omega _{x}^{-4}\omega _{v}^{2}\text{,}  \label{c1}
\end{equation}%
where $\omega _{x}^{2}=\sigma _{x}^{2}/6$, $\sigma
_{x}^{2}=\lim_{n\rightarrow \infty }n^{-1}\sum_{i=1}^{n}\sigma _{xi}^{2}$
and $\omega _{v}^{2}=\lim {}_{n\rightarrow \infty
}n^{-1}\sum_{i=1}^{n}\sigma _{xi}^{2}\sigma _{vi}^{2}/\left( 6\alpha
_{i}^{2}\right) $.
\end{theorem}

\begin{remark}
Like the PMG estimator in 
\citeN{PesaranShinSmith1999}%
, the PB estimator will also work when variables are integrated of order 0
(the I(0) case), which is not pursued in this paper. In the I(0) case, the
PB\ estimator converges at the standard rate of $\sqrt{nT}$.
\end{remark}

To conduct inference, let%
\begin{equation}
\hat{\omega}_{x}^{2}=n^{-1}\sum_{i=1}^{n}\frac{\mathbf{x}_{i}^{\prime }%
\mathbf{M}_{i}\mathbf{x}_{i}}{T^{2}}\text{,}  \label{omx}
\end{equation}%
and $\hat{\omega}_{v}^{2}=n^{-1}\sum_{i=1}^{n}\left( \frac{\mathbf{x}%
_{i}^{\prime }\mathbf{M}_{i}\mathbf{\hat{v}}_{i}^{\ast }}{T}\right) ^{2}$,
where $\mathbf{\hat{v}}_{i}^{\ast }=\mathbf{M}_{i}\left( \mathbf{y}_{i}-\hat{%
\beta}\mathbf{x}_{i}\right) $, and $\mathbf{M}_{i}$ is defined by (\ref{mi}%
). Accordingly, we propose the following estimator of $\Omega $:%
\begin{equation}
\hat{\Omega}=\hat{\omega}_{x}^{-4}\hat{\omega}_{v}^{2}\text{.}  \label{ve}
\end{equation}

\subsection{Bias mitigation and bootstrapping critical values for robust
inference\label{BM}}

When $n$ is sufficiently large relative to $T$, specifically when $\sqrt{n}%
/T\rightarrow K>0$, then $\sqrt{n}T\left( \hat{\beta}-\beta \right) $ is no
longer asymptotically distributed with zero mean. The asymptotic bias is due
to the nonzero mean of $T^{-1}\mathbf{\tilde{x}}_{i}^{\prime }\mathbf{M}_{i}%
\mathbf{\tilde{v}}_{i}$, and it can be of some relevance for finite sample
performance, as the Monte Carlo evidence in Section \ref{MC} illustrates.
Monte Carlo evidence also reveals that the inference based on PB and other
existing estimators in the literature can suffer from serious size
distortions in finite samples. To deal with these problems, we consider
bootstrapping critical values using sieve wild bootstrap for more accurate
and more robust inference that allows for cross-sectional dependence of
errors. In addition, we also consider two bias-correction techniques - a
bootstrap one as well as the split-panel jackknife method. The same
bias-correction methods are also applied to the three other estimators,
namely PMG, PDOLS, and FMOLS, considered in the paper. In what follows we
focus on the PB\ estimator. A description of bias-corrections applied to the
other three estimators are given in Section \ref{Sbcp} of Appendix B.

\subsubsection{Bootstrap bias reduction\label{Sbbc}}

Once an estimate of the bias of $\hat{\beta}$ is available, denoted as $\hat{%
b}$, then the bias-corrected PB estimator is given by%
\begin{equation}
\tilde{\beta}=\hat{\beta}-\hat{b}\text{.}  \label{bs}
\end{equation}%
One possibility of estimating the bias in the literature is by bootstrap. We
adopt the following sieve wild bootstrap algorithm for generating simulated
data.

\begin{enumerate}
\item Given $\hat{\beta}$, estimate the remaining unknown coefficients in (%
\ref{y})-(\ref{x}) by least squares, and compute residuals denoted by $\hat{u%
}_{y,it},\hat{u}_{x,it}$.

\item For each $r=1,2,...,R$, generate new draws for $\hat{u}_{y,it}^{\left(
r\right) }=a_{t}^{\left( r\right) }\hat{u}_{y,it}$, and $\hat{u}%
_{x,it}^{\left( r\right) }=a_{t}^{\left( r\right) }\hat{u}_{x,it}$, where $%
a_{t}^{\left( r\right) }$ is randomly drawn from Rademacher distribution (%
\citeANP{Liu1988}, \citeyearNP{Liu1988}%
),%
\begin{equation*}
a_{t}^{\left( r\right) }=\left\{ 
\begin{array}{cc}
-1\text{,} & \text{with probability 1/2} \\ 
1\text{,} & \text{with probability 1/2}%
\end{array}%
\right. \text{.}
\end{equation*}%
Given the estimated parameters of (\ref{y})-(\ref{x}) from Step 1, and the
initial values $\left\{ y_{i1},x_{i1}\text{ for }i=1,2,...,n\right\} $
generate the simulated series $y_{it}^{\left( r\right) },x_{it}^{\left(
r\right) }$ for $t=2,3,...,T$ , and $i=1,2,...,n$, and the bootstrap
estimates $\hat{\beta}^{(r)}$ for $r=1,2,...,R$.
\end{enumerate}

Using simulated data with $R=10,000$, we compute an estimate of the bias $%
\hat{b}_{R}=\left[ R^{-1}\sum_{r=1}^{R}\hat{\beta}^{\left( r\right) }-\hat{%
\beta}\right] $. We then compute the $\alpha $ percent critical values using
the $1-\alpha $ percent quantile of $\left\{ \left\vert t^{\left( r\right)
}\right\vert \right\} _{r=1}^{R}$, where $t^{\left( r\right) }=\tilde{\beta}%
^{\left( r\right) }/se\left( \tilde{\beta}^{\left( r\right) }\right) $, $%
\tilde{\beta}^{\left( r\right) }=\hat{\beta}^{\left( r\right) }-\hat{b}$ is
the bias-corrected estimate of $\beta $ using the $r$-$th$ draw of the
simulated data, $se\left( \tilde{\beta}^{\left( r\right) }\right)
=T^{-1}n^{-1/2}\hat{\Omega}^{\left( r\right) }$ is the corresponding
standard error estimate, and $\hat{\Omega}^{\left( r\right) }$ is computed
in the same way as $\hat{\Omega}$ in (\ref{ve}) but using the simulated data.

\subsubsection{Jackknife bias reduction}

The split-panel jackknife bias correction method is given by%
\begin{equation}
\tilde{\beta}_{jk}=\tilde{\beta}_{jk}\left( \kappa \right) =\hat{\beta}%
-\kappa \left( \frac{\hat{\beta}_{a}+\hat{\beta}_{b}}{2}-\hat{\beta}\right) ,
\label{bjk}
\end{equation}%
where $\hat{\beta}$ is the full sample PB estimator, $\hat{\beta}_{a}$ and $%
\hat{\beta}_{b}$ are the first and the second half sub-sample PB estimators,
and $\kappa $ is a suitably chosen weighting parameter. In a stationary
setting, where the bias is of order $O\left( T^{-1}\right) $, $\kappa $ is
chosen to be one, so that $\frac{K}{T}-\kappa \cdot \left( \frac{K}{T/2}-%
\frac{K}{T}\right) =0,$ for any arbitrary choice of $K$. See, for example, 
\citeN{DhaeneJochmansy2015}
and 
\citeN{ChudikPesaranYang2018}%
.

In general, when the bias is of order $O\left( T^{-\epsilon }\right) $ for
some $\epsilon >0$, then $\kappa $ can be chosen to solve $\frac{K}{%
T^{\epsilon }}-\kappa \cdot \left( \frac{K}{\left( T/2\right) ^{\epsilon }}-%
\frac{K}{T^{\epsilon }}\right) =0$, which yields $\kappa =1/\left(
2^{\epsilon }-1\right) $. Under our setup with I(1) variables, we need to
correct $\hat{\beta}$ for its $O\left( T^{-2}\right) $ bias, namely $%
\epsilon =2$, which yields $\kappa =1/3$.

Inference using $\tilde{\beta}^{jk}$ can be conducted based on (\ref{ve})
but with $\hat{\omega}_{v}^{2}$ replaced by%
\begin{equation}
\tilde{\omega}_{v}^{2}=\hat{\omega}_{v}^{2}=\frac{1}{n}\sum_{i=1}^{n}\left( 
\frac{\left[ \left( 1+\kappa \right) \mathbf{x}_{i}^{\prime }\mathbf{M}%
_{i}-2\kappa \mathbf{x}_{ab,i}^{\prime }\mathbf{M}_{ab,i}\right] \mathbf{%
\tilde{v}}_{i}^{\ast }}{T}\right) ^{2}\text{,}  \label{omvj}
\end{equation}%
where $\mathbf{\tilde{v}}_{i}^{\ast }=\mathbf{M}_{i}\left( \mathbf{y}_{i}-%
\tilde{\beta}^{jk}\mathbf{x}_{i}\right) $, 
\begin{equation*}
\mathbf{x}_{ab,i}^{\prime }=\left( 
\begin{array}{c}
\mathbf{x}_{a,i}^{\prime } \\ 
\mathbf{x}_{b,i}^{\prime }%
\end{array}%
\right) \text{, }\mathbf{M}_{ab,i}=\left( 
\begin{array}{c}
\mathbf{M}_{a,i} \\ 
\mathbf{M}_{b,i}%
\end{array}%
\right) \text{,}
\end{equation*}%
$\mathbf{x}_{a,i}^{\prime }$ $\left( \mathbf{x}_{b,i}^{\prime }\right) $ and 
$\mathbf{M}_{a,i}$ ($\mathbf{M}_{b,i}$) are defined in the same way as $%
\mathbf{x}_{i}$, and $\mathbf{M}_{i}$ but using only the first (second) half
of the sample.

We compute bootstrapped critical values to conduct more accurate and robust
small sample inference. Specifically, the $\alpha $ percent critical value
is computed as the $1-\alpha $ percent quantile of $\left\{ \left\vert
t_{jk}^{\left( r\right) }\right\vert \right\} _{r=1}^{R}$, where $%
t_{jk}^{\left( r\right) }=\tilde{\beta}_{jk}^{\left( r\right) }/se\left( 
\tilde{\beta}_{jk}^{\left( r\right) }\right) $, $\tilde{\beta}_{jk}^{\left(
r\right) }$ is the jackknife estimate of $\beta $ using the $r$-$th$ draw of
the simulated data generated using the algorithm described in Subsection \ref%
{Sbbc}, $se\left( \tilde{\beta}_{jk}^{\left( r\right) }\right) $ is the
corresponding standard error estimate, namely $se\left( \tilde{\beta}%
_{jk}^{\left( r\right) }\right) =T^{-1}n^{-1/2}\hat{\Omega}_{jk}^{\left(
r\right) }$, $\hat{\Omega}_{jk}^{\left( r\right) }=\hat{\omega}_{x,\left(
r\right) }^{-4}\tilde{\omega}_{v,\left( r\right) }^{2}$, in which $\tilde{%
\omega}_{v,\left( r\right) }$ and $\hat{\omega}_{x,\left( r\right) }^{2}$
are computed using the simulated data, based on expressions (\ref{omvj}) and
(\ref{omx}), respectively.

\section{Monte Carlo Evidence\label{MC}}

\subsection{Design\label{MCD}}

The Data Generating Process (DGP) is given by (\ref{y})-(\ref{x}), for $%
i=1,2,...,n,$ $T=1,2,...,T$, with starting values satisfying Assumption \ref%
{As3} with $\mathbf{\mu }_{i}\sim IIDN\left( \mathbf{\tau }_{2},\mathbf{I}%
_{2}\right) $, and $c_{i}=$ $\alpha _{i}\mu _{i,1}-\alpha _{i}\beta \mu
_{i,2}$. We generate $\alpha _{i}\sim IIDU\left[ 0.2,0.3\right] $. We
consider two DGPs based on the cross-sectional dependence of errors. In the
cross-sectionally independent DGP, we generate $u_{y,it}=\sigma
_{y,i}e_{y,it}$, $u_{x,it}=\sigma _{x,i}e_{x,it}$, $\sigma _{y,i}^{2},\sigma
_{x,i}^{2}\sim IIDU\left[ 0.8,1.2\right] $,%
\begin{equation*}
\left( 
\begin{array}{c}
e_{y,it} \\ 
e_{x,it}%
\end{array}%
\right) \sim IIDN\left( \mathbf{0}_{2},\mathbf{\Sigma }_{e}\right) \text{, }%
\mathbf{\Sigma }_{e}\sim \left( 
\begin{array}{cc}
1 & \rho _{i} \\ 
\rho _{i} & 1%
\end{array}%
\right) \text{, and }\rho _{i}\sim IIDU\left[ 0.3,0.7\right] \text{.}
\end{equation*}%
In the DGP with cross-sectionally dependent errors, we generate $e_{y,it}$
to contain a factor structure including strong, semi-strong and weak factors:%
\begin{equation*}
e_{y,it}=\varkappa _{i}\left( \varepsilon _{y,it}+\sum_{\ell =1}^{m}\gamma
_{i\ell }f_{\ell t}\right) \text{, }
\end{equation*}%
where $\varepsilon _{y,it}\sim IIDN\left( 0,1\right) $, $f_{\ell ,t}\sim
IIDN\left( 0,1\right) $, $\gamma _{\ell }\sim IIDU\left[ 0,\gamma _{\max
,\ell }\right] $, for $\ell =1,2,...,m$. We choose $m=5$ factors and $\gamma
_{\max ,\ell }=2n^{\alpha _{\ell }-1}$ with $\alpha _{\ell
}=1,0.9,0.8,0.7,0.6$, for $\ell =1,2,...,5$, respectively. Scaling constant $%
\varkappa _{i}$ is set to ensure $E\left( e_{y,it}^{2}\right) =1$, namely $%
\varkappa _{i}=\left( 1+\sum_{\ell =1}^{m}\gamma _{i\ell }^{2}\right)
^{-1/2} $. We generate $e_{x,it}$ to ensure unit variance and $cov\left(
e_{y,it},e_{x,it}\right) =\rho _{i}$. Specifically, $e_{x,it}=\rho
_{i}e_{y,it}+\sqrt{1-\rho _{i}^{2}}\varepsilon _{x,it}$, $\varepsilon
_{x,it}\sim IIDN\left( 0,1\right) $. Both designs features heteroskedastic
(over $i$) and correlated (over $y$ \& $x$ equations) errors, namely $%
E\left( u_{y,it}^{2}\right) =\sigma _{y,i}^{2}$, $E\left(
u_{x,it}^{2}\right) =\sigma _{x,i}^{2}$, and $corr\left(
u_{y,it},u_{x,it}\right) =\rho _{i}$. We consider $n,T=20,30,40,50$ and
compute $R_{MC}=2000$ Monte Carlo replications.

\subsection{Bias, RMSE and inference}

We report bias, root mean square error (RMSE), size ($H_{0}:\beta =1$, 5\%
nominal level) and power ($H_{1}:\beta =0.9$, 5\% nominal level) findings
for the PB estimator $\hat{\beta}$ given by (\ref{pmg}), with variance
estimated using (\ref{ve}). Moreover, we also report findings for the two
bias corrected versions of PB estimator as described in Subsection \ref{BM}
with bootstrapped critical values for inference robust to error
cross-sectional dependence. We compare the performance of the PB estimator
with the PMG estimator by 
\citeN{PesaranShinSmith1999}%
, panel dynamic OLS (PDOLS) estimator by 
\citeN{MarkSul2003}%
, and the group-mean fully modified OLS (FMOLS) estimator by 
\citeANP{Pedroni1996} (\citeyearNP{Pedroni1996}, \citeyearNP{Pedroni2001ReStat})%
. Similarly to the PB estimator, we also consider jackknife and bootstrap
based bias-corrected versions of the PMG, PDOLS and FMOLS estimators with
cross-sectionally robust bootstrapped critical values, described in Appendix
B. We use $R_{b}=10,000$ bootstrap replications (within each MC replication)
for bootstrap bias correction and for computation of robust and more
accurate bootstrapped critical values.

\subsection{Findings}

Table 1 report the results for the original (without bias-correction)
estimators. PB estimator stands out as the most precise estimator in terms
of having the lowest RMSE values among the four estimators. The second best
is PMG estimator with RMSE values 1 to 21 percent larger compared with the
PB estimator, the third is PDOLS with RMSE values 23 to 66 percent larger
compared with PB, and the FMOLS comes last with RMSE\ values 95 to 180
percent larger compared with PB. In terms of the bias alone, the ordering of
the estimators is slightly different with PMG and PB switching their places.
For $T=20$, the bias of PMG estimator is -0.016 to -0.020, the bias of PB
estimator is in the range -0.034 to -0.037, the bias of the PDOLS estimator
is in the range -0.052 to -0.056 and the bias of the FMOLS estimator is in
the range -0.104 to -0.110. For such a small value of $T$, the bias is not
very large, and, as expected, it declines with an increase in $T$.

All four estimators suffer from varying degrees of size distortions. The
inference based on the PB estimator is the most accurate. Specifically, the
size distortions for the PB estimator are lowest among the four estimators -
with reported size in the range between 9.9 and 25.2 percent, exceeding the
chosen nominal value of 5 percent. Size distortions diminish with an
increase in $T$.\bigskip \bigskip

\begin{center}
\textbf{Table 1: }MC findings for the estimation of long-run coefficient $%
\beta $ in experiments with cross-sectionally independent errors.

Estimators without bias correction and inference conducted using standard
critical values.

\bigskip

\footnotesize%

\renewcommand{\arraystretch}{1.0}\setlength{\tabcolsep}{3pt}%
\scriptsize%
\begin{tabular}{rrrrrccrrrrccrrrrccrrrr}
\hline\hline
& \multicolumn{4}{c}{\textbf{Bias (}$\mathbf{\times }$\textbf{\ 100)}} &  & 
& \multicolumn{4}{c}{\textbf{RMSE (}$\mathbf{\times }$\textbf{\ 100)}} &  & 
& \multicolumn{4}{c}{\textbf{Size (5\% level)}} &  &  & \multicolumn{4}{c}{%
\textbf{Power (5\% level)}} \\ 
\cline{2-5}\cline{8-11}\cline{14-17}\cline{20-23}
$n\backslash T$ & \textbf{20} & \textbf{30} & \textbf{40} & \textbf{50} & 
\multicolumn{1}{r}{} & \multicolumn{1}{r}{} & \textbf{20} & \textbf{30} & 
\textbf{40} & \textbf{50} & \multicolumn{1}{r}{} & \multicolumn{1}{r}{} & 
\textbf{20} & \textbf{30} & \textbf{40} & \textbf{50} & \multicolumn{1}{r}{}
& \multicolumn{1}{r}{} & \textbf{20} & \textbf{30} & \textbf{40} & \textbf{50%
} \\ \cline{2-23}
& \multicolumn{22}{l}{PB} \\ \hline
\textbf{20} & -3.69 & -1.75 & -1.07 & -0.73 &  &  & 6.43 & 4.12 & 3.04 & 2.45
&  &  & 18.40 & 13.35 & 11.80 & 11.40 &  &  & 34.00 & 68.30 & 89.95 & 97.70
\\ 
\textbf{30} & -3.39 & -1.79 & -1.04 & -0.74 &  &  & 5.55 & 3.54 & 2.58 & 2.03
&  &  & 19.40 & 14.50 & 11.95 & 10.10 &  &  & 43.70 & 81.50 & 96.50 & 99.60
\\ 
\textbf{40} & -3.56 & -1.87 & -1.06 & -0.74 &  &  & 5.18 & 3.25 & 2.33 & 1.81
&  &  & 21.25 & 15.55 & 12.30 & 10.45 &  &  & 45.90 & 87.35 & 99.05 & 99.95
\\ 
\textbf{50} & -3.58 & -1.90 & -1.09 & -0.74 &  &  & 4.96 & 3.05 & 2.18 & 1.66
&  &  & 25.20 & 15.55 & 13.40 & 9.95 &  &  & 54.05 & 93.25 & 99.65 & 100.00
\\ \hline
& \multicolumn{22}{l}{PMG} \\ \hline
\textbf{20} & -1.97 & -0.89 & -0.51 & -0.32 &  &  & 7.77 & 4.79 & 3.40 & 2.57
&  &  & 39.45 & 28.15 & 21.40 & 17.85 &  &  & 63.40 & 82.20 & 93.75 & 98.95
\\ 
\textbf{30} & -1.56 & -0.97 & -0.41 & -0.33 &  &  & 6.32 & 3.99 & 2.78 & 2.08
&  &  & 41.10 & 28.45 & 22.50 & 16.55 &  &  & 71.20 & 89.60 & 98.10 & 99.85
\\ 
\textbf{40} & -1.64 & -0.86 & -0.44 & -0.31 &  &  & 5.71 & 3.44 & 2.47 & 1.85
&  &  & 43.10 & 29.25 & 23.05 & 18.10 &  &  & 77.25 & 95.15 & 99.60 & 100.00
\\ 
\textbf{50} & -1.70 & -0.93 & -0.48 & -0.31 &  &  & 5.23 & 3.10 & 2.26 & 1.67
&  &  & 42.25 & 28.60 & 23.75 & 16.85 &  &  & 81.00 & 97.00 & 99.80 & 100.00
\\ \cline{2-23}
& \multicolumn{22}{l}{PDOLS} \\ \hline
\textbf{20} & -5.60 & -3.64 & -2.82 & -2.32 &  &  & 7.93 & 5.28 & 4.02 & 3.31
&  &  & 21.75 & 19.10 & 17.30 & 19.30 &  &  & 17.35 & 43.15 & 72.95 & 90.15
\\ 
\textbf{30} & -5.25 & -3.60 & -2.75 & -2.26 &  &  & 7.05 & 4.81 & 3.64 & 2.96
&  &  & 24.10 & 23.10 & 23.10 & 24.20 &  &  & 21.40 & 55.90 & 84.85 & 97.30
\\ 
\textbf{40} & -5.47 & -3.77 & -2.84 & -2.31 &  &  & 6.78 & 4.69 & 3.53 & 2.86
&  &  & 29.90 & 30.60 & 28.30 & 30.10 &  &  & 21.70 & 62.60 & 91.35 & 99.10
\\ 
\textbf{50} & -5.46 & -3.78 & -2.88 & -2.30 &  &  & 6.57 & 4.52 & 3.45 & 2.75
&  &  & 35.10 & 34.30 & 34.40 & 36.45 &  &  & 25.25 & 72.40 & 95.85 & 99.85
\\ 
& \multicolumn{22}{l}{FMOLS} \\ \hline
\textbf{20} & -11.01 & -7.16 & -5.45 & -4.25 &  &  & 12.56 & 8.44 & 6.55 & 
5.18 &  &  & 89.25 & 78.25 & 69.90 & 64.15 &  &  & 45.00 & 56.60 & 79.05 & 
92.60 \\ 
\textbf{30} & -10.44 & -7.06 & -5.30 & -4.17 &  &  & 11.58 & 7.99 & 6.09 & 
4.83 &  &  & 93.80 & 86.15 & 77.90 & 71.95 &  &  & 44.90 & 63.20 & 88.75 & 
98.20 \\ 
\textbf{40} & -10.78 & -7.31 & -5.50 & -4.26 &  &  & 11.59 & 8.00 & 6.08 & 
4.77 &  &  & 97.45 & 92.95 & 85.65 & 82.20 &  &  & 44.60 & 67.85 & 92.70 & 
99.20 \\ 
\textbf{50} & -10.76 & -7.35 & -5.50 & -4.22 &  &  & 11.44 & 7.91 & 5.98 & 
4.65 &  &  & 98.70 & 96.25 & 91.35 & 86.85 &  &  & 46.10 & 73.70 & 95.35 & 
99.85 \\ \hline\hline
\end{tabular}%
\vspace{-0.2in}
\end{center}

\begin{flushleft}
\footnotesize%
\singlespacing%
Notes: DGP is given by $\Delta y_{it}=c_{i}-\alpha _{i}\left(
y_{i,t-1}-\beta x_{i,t-1}\right) +u_{y,it}$ and $\Delta x_{it}=u_{x,it}$,
for $i=1,2,...,n,$ $T=1,2,...,T$, with $\beta =1$ and $\alpha _{i}\sim IIDU%
\left[ 0.2,0.3\right] $. Errors $u_{y,it}$, $u_{x,it}$ are cross-sectionally
independent, heteroskedastic over $i$, and correlated over $y$ \& $x$
equations. See Section \ref{MCD} for complete description of the DGP. The
pooled Bewley estimator is given by (\ref{pmg}), with variance estimated
using (\ref{ve}). PMG is the Pooled Mean Group estimator proposed by 
\citeN{PesaranShinSmith1999}%
. PDOLS is panel dynamic OLS estimator by 
\citeN{MarkSul2003}%
. FMOLS is the group-mean fully modified OLS estimator by 
\citeANP{Pedroni1996} (\citeyearNP{Pedroni1996}, \citeyearNP{Pedroni2001ReStat})%
. The size and power findings are computed using 5\% nominal level and the
reported power is the rejection frequency for testing the hypothesis $\beta
=0.9$. 
\normalsize%
\pagebreak
\end{flushleft}

\begin{center}
\textbf{Table 2: }MC findings for the estimation of long-run coefficient $%
\beta $ in experiments with cross-sectionally independent errors.

Bias corrected estimators and inference conducted using bootstrapped
critical values.

\bigskip

\footnotesize%

\renewcommand{\arraystretch}{0.75}\setlength{\tabcolsep}{3pt}%
\scriptsize%
\begin{tabular}{rrrrrccrrrrccrrrrccrrrr}
\hline\hline
& \multicolumn{4}{c}{\textbf{Bias (}$\mathbf{\times }$\textbf{\ 100)}} &  & 
& \multicolumn{4}{c}{\textbf{RMSE (}$\mathbf{\times }$\textbf{\ 100)}} &  & 
& \multicolumn{4}{c}{\textbf{Size (5\% level)}} &  &  & \multicolumn{4}{c}{%
\textbf{Power (5\% level)}} \\ 
\cline{2-5}\cline{8-11}\cline{14-17}\cline{20-23}
$n\backslash T$ & \textbf{20} & \textbf{30} & \textbf{40} & \textbf{50} & 
\multicolumn{1}{r}{} & \multicolumn{1}{r}{} & \textbf{20} & \textbf{30} & 
\textbf{40} & \textbf{50} & \multicolumn{1}{r}{} & \multicolumn{1}{r}{} & 
\textbf{20} & \textbf{30} & \textbf{40} & \textbf{50} & \multicolumn{1}{r}{}
& \multicolumn{1}{r}{} & \textbf{20} & \textbf{30} & \textbf{40} & \textbf{50%
} \\ \hline
& \multicolumn{22}{l}{\textbf{Jackknife bias-corrected estimators}} \\ \hline
& \multicolumn{22}{l}{PB} \\ \hline
\textbf{20} & -1.60 & -0.52 & -0.25 & -0.14 &  &  & 6.27 & 4.18 & 3.10 & 2.52
&  &  & 5.95 & 5.15 & 5.40 & 5.30 &  &  & 25.10 & 56.75 & 82.95 & 94.30 \\ 
\textbf{30} & -1.34 & -0.55 & -0.25 & -0.18 &  &  & 5.24 & 3.45 & 2.58 & 2.06
&  &  & 6.70 & 5.05 & 5.10 & 5.00 &  &  & 37.45 & 73.70 & 93.10 & 99.05 \\ 
\textbf{40} & -1.45 & -0.60 & -0.24 & -0.15 &  &  & 4.59 & 3.03 & 2.29 & 1.81
&  &  & 5.85 & 5.05 & 4.50 & 4.90 &  &  & 42.80 & 84.10 & 98.15 & 99.85 \\ 
\textbf{50} & -1.53 & -0.62 & -0.26 & -0.15 &  &  & 4.23 & 2.72 & 2.07 & 1.61
&  &  & 6.20 & 4.60 & 5.65 & 4.25 &  &  & 51.10 & 90.80 & 99.55 & 100.00 \\ 
\hline
& \multicolumn{22}{l}{PMG} \\ \hline
\textbf{20} & -0.55 & -0.20 & -0.05 & 0.01 &  &  & 9.26 & 5.59 & 3.78 & 2.86
&  &  & 14.90 & 11.05 & 9.10 & 7.20 &  &  & 33.45 & 57.15 & 83.05 & 95.60 \\ 
\textbf{30} & -0.10 & -0.30 & 0.02 & -0.04 &  &  & 7.35 & 4.48 & 3.09 & 2.28
&  &  & 14.00 & 10.95 & 8.50 & 7.35 &  &  & 45.60 & 72.05 & 93.10 & 99.30 \\ 
\textbf{40} & -0.27 & -0.13 & -0.01 & -0.01 &  &  & 6.73 & 3.83 & 2.72 & 2.06
&  &  & 15.40 & 9.30 & 8.90 & 8.05 &  &  & 48.90 & 82.95 & 97.10 & 99.75 \\ 
\textbf{50} & -0.42 & -0.21 & -0.06 & 0.00 &  &  & 6.13 & 3.43 & 2.47 & 1.83
&  &  & 16.45 & 9.65 & 9.40 & 7.55 &  &  & 57.55 & 88.55 & 98.85 & 100.00 \\ 
\cline{2-23}
& \multicolumn{22}{l}{PDOLS} \\ \hline
\textbf{20} & -4.22 & -2.56 & -1.97 & -1.59 &  &  & 8.00 & 5.05 & 3.77 & 3.04
&  &  & 8.40 & 5.80 & 5.40 & 4.75 &  &  & 8.70 & 27.50 & 55.25 & 77.90 \\ 
\textbf{30} & -3.91 & -2.55 & -1.92 & -1.56 &  &  & 6.89 & 4.42 & 3.30 & 2.64
&  &  & 8.70 & 5.75 & 5.95 & 4.70 &  &  & 10.15 & 37.30 & 69.75 & 91.00 \\ 
\textbf{40} & -4.09 & -2.69 & -1.98 & -1.60 &  &  & 6.35 & 4.17 & 3.08 & 2.47
&  &  & 7.70 & 5.80 & 4.90 & 4.05 &  &  & 10.05 & 41.10 & 75.20 & 94.95 \\ 
\textbf{50} & -4.11 & -2.70 & -2.02 & -1.58 &  &  & 6.01 & 3.91 & 2.93 & 2.30
&  &  & 7.90 & 5.90 & 4.55 & 3.55 &  &  & 11.30 & 45.75 & 83.75 & 97.60 \\ 
& \multicolumn{22}{l}{FMOLS} \\ \hline
\textbf{20} & -8.70 & -5.07 & -3.70 & -2.76 &  &  & 11.19 & 7.19 & 5.52 & 
4.33 &  &  & 10.85 & 5.95 & 4.55 & 4.00 &  &  & 1.05 & 1.40 & 8.10 & 25.75
\\ 
\textbf{30} & -8.17 & -5.02 & -3.60 & -2.73 &  &  & 10.02 & 6.60 & 4.94 & 
3.86 &  &  & 9.95 & 5.00 & 4.65 & 3.00 &  &  & 0.50 & 1.65 & 10.20 & 35.20
\\ 
\textbf{40} & -8.49 & -5.22 & -3.78 & -2.78 &  &  & 9.84 & 6.40 & 4.79 & 3.68
&  &  & 9.20 & 5.00 & 3.45 & 3.25 &  &  & 0.15 & 0.75 & 8.85 & 38.70 \\ 
\textbf{50} & -8.50 & -5.28 & -3.78 & -2.74 &  &  & 9.61 & 6.26 & 4.62 & 3.50
&  &  & 10.65 & 5.00 & 3.75 & 2.85 &  &  & 0.10 & 0.95 & 10.55 & 45.75 \\ 
\hline
& \multicolumn{22}{l}{\textbf{Bootstrap bias-corrected estimators}} \\ \hline
& \multicolumn{22}{l}{PB} \\ \hline
\textbf{20} & -1.28 & -0.33 & -0.16 & -0.11 &  &  & 5.87 & 3.93 & 2.93 & 2.39
&  &  & 7.25 & 6.20 & 6.05 & 5.75 &  &  & 36.20 & 67.60 & 88.30 & 96.95 \\ 
\textbf{30} & -0.98 & -0.39 & -0.15 & -0.13 &  &  & 4.90 & 3.24 & 2.43 & 1.93
&  &  & 7.75 & 6.35 & 6.35 & 5.80 &  &  & 52.20 & 84.00 & 96.90 & 99.55 \\ 
\textbf{40} & -1.07 & -0.43 & -0.13 & -0.10 &  &  & 4.24 & 2.83 & 2.15 & 1.70
&  &  & 7.00 & 6.55 & 5.95 & 5.10 &  &  & 59.75 & 91.35 & 99.45 & 99.95 \\ 
\textbf{50} & -1.09 & -0.44 & -0.15 & -0.10 &  &  & 3.89 & 2.56 & 1.96 & 1.52
&  &  & 8.95 & 6.25 & 6.45 & 5.55 &  &  & 69.40 & 95.50 & 99.90 & 100.00 \\ 
\hline
& \multicolumn{22}{l}{PMG} \\ \hline
\textbf{20} & -1.28 & -0.44 & -0.21 & -0.11 &  &  & 7.88 & 4.83 & 3.41 & 2.57
&  &  & 14.10 & 10.45 & 7.90 & 6.80 &  &  & 35.25 & 63.40 & 86.70 & 97.30 \\ 
\textbf{30} & -0.88 & -0.55 & -0.12 & -0.13 &  &  & 6.40 & 3.99 & 2.79 & 2.08
&  &  & 13.10 & 10.75 & 8.25 & 6.85 &  &  & 47.45 & 77.75 & 95.85 & 99.65 \\ 
\textbf{40} & -0.96 & -0.44 & -0.14 & -0.11 &  &  & 5.73 & 3.42 & 2.47 & 1.85
&  &  & 14.95 & 9.05 & 7.80 & 6.90 &  &  & 53.55 & 87.60 & 98.40 & 99.85 \\ 
\textbf{50} & -1.02 & -0.51 & -0.19 & -0.10 &  &  & 5.20 & 3.06 & 2.24 & 1.66
&  &  & 15.15 & 9.85 & 9.05 & 6.95 &  &  & 60.90 & 91.65 & 99.30 & 100.00 \\ 
\hline
& \multicolumn{22}{l}{PDOLS} \\ \hline
\textbf{20} & -2.19 & -0.90 & -0.59 & -0.42 &  &  & 6.75 & 4.31 & 3.12 & 2.55
&  &  & 10.35 & 7.80 & 7.25 & 7.25 &  &  & 28.70 & 64.45 & 86.65 & 95.85 \\ 
\textbf{30} & -1.89 & -0.96 & -0.59 & -0.43 &  &  & 5.70 & 3.62 & 2.63 & 2.08
&  &  & 10.75 & 8.90 & 7.70 & 7.15 &  &  & 39.90 & 77.10 & 95.35 & 99.30 \\ 
\textbf{40} & -2.00 & -1.04 & -0.59 & -0.41 &  &  & 4.98 & 3.24 & 2.34 & 1.84
&  &  & 10.05 & 8.65 & 7.05 & 7.35 &  &  & 46.40 & 85.75 & 98.65 & 99.95 \\ 
\textbf{50} & -2.00 & -1.05 & -0.64 & -0.40 &  &  & 4.61 & 2.93 & 2.14 & 1.65
&  &  & 10.70 & 9.40 & 8.65 & 6.80 &  &  & 52.85 & 91.55 & 99.50 & 100.00 \\ 
\hline
& \multicolumn{22}{l}{FMOLS} \\ \hline
\textbf{20} & -4.59 & -1.97 & -1.30 & -0.80 &  &  & 8.84 & 5.62 & 4.28 & 3.40
&  &  & 17.10 & 10.90 & 9.20 & 7.40 &  &  & 20.20 & 41.10 & 63.20 & 81.70 \\ 
\textbf{30} & -4.19 & -2.04 & -1.28 & -0.84 &  &  & 7.54 & 4.86 & 3.61 & 2.83
&  &  & 18.50 & 10.90 & 9.15 & 7.70 &  &  & 24.85 & 51.60 & 78.60 & 93.15 \\ 
\textbf{40} & -4.36 & -2.14 & -1.34 & -0.81 &  &  & 6.89 & 4.36 & 3.22 & 2.51
&  &  & 20.10 & 12.05 & 9.80 & 7.75 &  &  & 26.10 & 60.40 & 86.25 & 96.80 \\ 
\textbf{50} & -4.41 & -2.22 & -1.35 & -0.78 &  &  & 6.56 & 4.10 & 2.98 & 2.29
&  &  & 23.80 & 14.35 & 10.50 & 8.90 &  &  & 30.45 & 69.45 & 92.05 & 99.20
\\ \hline\hline
\end{tabular}%
\vspace{-0.2in}
\end{center}

\begin{flushleft}
\footnotesize%
\singlespacing%
Notes: See the notes to Table 1. Bias-corrected versions of the PB estimator
are described in Subsection \ref{BM}. Bias-corrected versions of the PMG,
PDOLS and FMOLS estimator are described in Appendix B. Inference is
conducted using bootstrapped critical values.%
\normalsize%
\pagebreak
\end{flushleft}

We consider next the bias-corrected versions of the four estimators with
inference carried out using robust bootstrap critical values. Upper panel of
Table 2 reports findings for estimators corrected for bias using the
jackknife procedure, and the bottom panel reports on bootstrap bias
corrected estimators. Bias correction did not change the overall ranking of
estimators -- PB continues to be the most precise (lowest RMSE). Both bias
correction approaches are quite effective in reducing the bias. The bias of
PB\ and PMG estimators for any of the two bias corrections are very low. In
addition to reducing the bias, in many cases the bias-correction also
resulted in reduced RMSE values. In the case of the PB estimator, using
bootstrap bias correction resulted in improved RMSE performance for all
choices of $n,T$ - by about 2 to 22 percent. Results in Table 2 also show
notable improvement to inference comes from using bootstrapped critical
values - with PB having virtually no size distortions and size distortions
of the remaining estimators are relatively minor.

Last but not least, we consider the DGP with cross-sectionally correlated
errors. The corresponding results, reported in Tables B1 and B2 in Appendix
B, reveal the same ranking of the four estimators, and, importantly, the
bootstrapped critical values continue to deliver correct size, despite the
error cross-sectional dependence.

The Monte Carlo results show that PB\ estimator can perform better (in terms
of overall precision as measured by RMSE, and in terms of accuracy of
inference) than existing estimators (PMG, PDOLS, and FMOLS) in finite sample
sizes of interest, whether or not bias correction is considered. Of' course,
our results do not imply that PB estimator will always be better, but that
it can be a useful addition to the existing literature as a complement to
PMG, PDOLS, and FMOLS estimators. Bias corrections and bootstrapping
critical values are helpful for all four estimators, resulting not only in
reduced bias, but sometimes also in better RMSE. In all cases, they result
in more accurate inference in our experiments.

\section{Empirical Application\label{EA}}

This section revisits consumption function empirical application undertaken
by 
\citeN{PesaranShinSmith1999}%
, hereafter PSS. The long-run consumption function is assumed to be given by%
\begin{equation*}
c_{it}=d_{i}+\beta _{1}y_{it}^{d}+\beta _{2}\pi _{it}+\vartheta _{it}\text{,}
\end{equation*}%
for country $i=1,2,...,n$, where $c_{it}$ is the logarithm of real
consumption per capita, $y_{it}^{d}$ is the logarithm of real per capita
disposable income, $\pi _{it}$ is the rate of inflation, and $\vartheta
_{it} $ is an $I\left( 0\right) $ process. We take the dataset from PSS,
which consists of $n=24$ countries and a slightly unbalanced time period
covering 1960-1993. PSS estimate $\beta _{1}$ and $\beta _{2}$ using an
ARDL(1,1,1) specification, which can be written as error-correcting panel
regressions%
\begin{equation}
\Delta c_{it}=-\alpha _{i}\left( c_{i,t-1}-d_{i}-\beta
_{1}y_{i,t-1}^{d}-\beta _{2}\pi _{i,t-1}\right) +\delta _{i1}\Delta
y_{it}^{d}+\delta _{i2}\Delta \pi _{it}+v_{it}\text{,}  \label{cf}
\end{equation}%
for $i=1,2,...,n$, where all coefficients, except the long-run coefficients $%
\beta _{1}$ and $\beta _{2}$ are country-specific.

Table 3 presents alternative estimates of the long-run coefficients. The
upper panel presents findings for estimators without bias correction and
standard confidence intervals. The middle and lower panels present jackknife
and bootstrap\textbf{\ }bias-corrected estimates with confidence intervals
based on bootstrapped critical values. Results differ widely across
different approaches to estimation and inference. Depending on which bias
correction approach is conducted, the PB\ estimates of the long-run
coefficient on real income ($\beta _{1}$) is estimated to be 0.921 or 0.926,
and the long-run coefficient on the inflation variable ($\beta _{2}$) is
estimated to be -0.120 or -0.125. The null hypothesis that the coefficient
on $y_{it}^{d}$ is unity cannot be rejected at the 5 percent nominal level,
nor is the hypothesis that the long run coefficient on inflation is zero.
From an economic perspective, unit long-run real income elasticity and no
long-run effects of inflation on consumption seem both plausible - the
former hypothesis is in line with balanced growth path models, and the
latter in line with monetary policy neutrality in the long-run. A different
conclusion would be reached according to PMG estimates - namely both the
unit coefficients on the real income variable and zero coefficient on
inflation would be rejected at the 5 percent nominal level. The results
based on the PDOLS are in line with the PB estimates and do not reject unit
real income and zero inflation long run coefficients. FMOLS estimates of $%
\beta _{1}$ are larger than the other estimates, but the unit coefficient on
the income variable still cannot be rejected. The FMOLS estimates of $\beta
_{2}$ are also quite large. The choice of estimation method clearly matters
in this empirical illustration.\pagebreak

\begin{center}
\textbf{Table 3: Estimated consumption function coefficients for OECD
countries}\bigskip

\small%
\renewcommand{\arraystretch}{0.75}%
\begin{tabular}{lllll}
\hline\hline
& $\beta _{1}$: Income & 95\% Conf. Int. & $\beta _{2}$: Inflation & 95\%
Conf. Int. \\ \hline
& \multicolumn{4}{l}{Estimator without bias correction} \\ \hline
PB & .912 & [.845,.980] & -.134 & [-.260,-.008] \\ 
PMG & .904 & [.889,.919] & -.466 & [-.566,-.365] \\ 
PDOLS & .923 & [.798,1.047] & -.187 & [-.407,.033] \\ 
FMOLS & .951 & [.942,.959] & -.336 & [-.408,-.265] \\ \hline
& \multicolumn{4}{l}{Jackknife bias-corrected estimators} \\ \hline
PB & .926 & [.835,1.017] & -.120 & [-.345,.105] \\ 
PMG & .915 & [.880,.949] & -.403 & [-616.,-.190] \\ 
PDOLS & .940 & [.737,1.143] & -.184 & [-.530,.161] \\ 
FMOLS & .983 & [.912,1.053] & -.397 & [-1.370,.576] \\ \hline
& \multicolumn{4}{l}{Bootstrap bias-corrected estimators} \\ \hline
PB & .921 & [.830,1.012] & -.125 & [-.314,.065] \\ 
PMG & .905 & [.875,.936] & -.477 & [-.657,-.297] \\ 
PDOLS & .932 & [.746,1.118] & -.183 & [-.499,.133] \\ 
FMOLS & .985 & [.941,1.028] & -.438 & [-1.047,.171] \\ \hline\hline
\end{tabular}%
\vspace{-0.4cm}
\end{center}

\begin{flushleft}
\singlespacing%
\scriptsize%

Notes: This table revisits empirical application in Table 1 of 
\citeN{PesaranShinSmith1999}%
, reporting estimates of long-run income elasticity ($\beta _{1}$) and
inflation effect ($\beta _{2}$) coefficients and their 95\% confidence
intervals in the ARDL(1,1,1) consumption functions (\ref{cf}) for OECD
countries using the dataset from 
\citeN{PesaranShinSmith1999}%
. PB stands for pooled Bewley estimator developed in this paper. PMG is the
Pooled Mean Group estimator proposed by 
\citeN{PesaranShinSmith1999}%
. PDOLS is panel dynamic OLS estimator by 
\citeN{MarkSul2003}%
. FMOLS is the group-mean fully modified OLS estimator by 
\citeANP{Pedroni1996} (\citeyearNP{Pedroni1996}, \citeyearNP{Pedroni2001ReStat})%
. Description of bias correction methods is provided in Subsection \ref{BM}\
for PB estimator and in Appendix B for PMG, PDOLS and FMOLS estimators.
Inference in the case of original estimators uncorrected for bias is
conducted using the standard asymptotic critical values, and it is valid
only when errors are not cross-sectionally dependent. Inference in the case
of bias-corrected estimators is conducted using bootstrapped critical values
following 
\citeN{ChudikPesaranSmith2022}%
, and it is robust to cross-section dependence of errors.
\end{flushleft}

\doublespacing%
\normalsize%

\section{Conclusion\label{CON}}

This paper proposes the pooled Bewley (PB)\ estimator of long-run
relationships in heterogeneous dynamic panels. Relative to existing
estimators in the literature -- namely PMG, PDOLS and FMOLS -- Monte Carlo
evidence reveals that PB can perform well in small samples. While we
developed the asymptotic theory of PB estimator under a similar setting to
the PMG estimator, notably we assumed cross-sectionally independent errors,
we have also shown the benefit of bootstrapping critical values for
inference when errors are cross-sectionally correlated for all four
estimators.

While the asymptotic distribution of the other estimators are derived for
the case where $n$ is fixed and $T\rightarrow \infty ,$ we derive the joint $%
\left( n,T\right) $ asymptotic distribution of the PB estimator, when both $%
n $ and $T$ diverge to infinity jointly such that $n=\Theta \left( T^{\theta
}\right) $, for $0<\theta <2.$ This covers a broader range of empirical
applications where both $n$ and $T$ are large. The small sample and
asymptotic results suggest that the PB estimator is a useful addition to
estimators for long run effects in single equation dynamic heterogeneous
panels, where the direction of long-run causality is known.\bigskip

\pagebreak 
\small%

\bibliographystyle{chicago}
\bibliography{ref_CPS_PB}

\pagebreak

\noindent 
\appendix%

\numberwithin{equation}{section}%

\numberwithin{lemma}{section}%

\onehalfspacing%

\LARGE%

\begin{flushleft}
\textbf{Appendices}
\end{flushleft}

\small%

\section{Mathematical derivations}

This appendix is organized in four sections. Section \ref{A1} introduces
some notations and definitions. Section \ref{A2} presents lemmas and proofs
needed for the proof of Theorem \ref{T1} presented in the body of the paper.

\subsection{Notations and definitions\label{A1}}

Let $\mathbf{z}_{it}=\left( y_{it},x_{it}\right) ^{\prime }$, and define $%
\mathbf{C}_{i}\left( L\right) =\sum_{\ell =0}^{\infty }\mathbf{C}_{i\ell
}L^{\ell }$ and $\mathbf{C}_{i}^{\ast }\left( L\right) =\sum_{\ell
=0}^{\infty }\mathbf{C}_{i\ell }^{\ast }L^{\ell }$, where%
\begin{eqnarray*}
\mathbf{C}_{i0} &=&\mathbf{I}_{2}\text{,} \\
\mathbf{C}_{i\ell } &=&\left( \mathbf{\Phi }_{i}-\mathbf{I}_{2}\right) 
\mathbf{\Phi }_{i}^{\ell -1}\text{, }\ell =1,2,....\text{,}
\end{eqnarray*}%
\begin{equation}
\mathbf{\Phi }_{i}=\left( 
\begin{array}{cc}
1-\alpha _{i} & \alpha _{i}\beta \\ 
0 & 1%
\end{array}%
\right) \text{,}
\end{equation}%
\begin{equation*}
\mathbf{C}_{i}(1)=\mathbf{C}_{i0}+\mathbf{C}_{i1}+.....=\lim_{\ell
\rightarrow \infty }\mathbf{\Phi }_{i}^{\ell }=\left( 
\begin{array}{cc}
0 & \beta \\ 
0 & 1%
\end{array}%
\right) \text{,}
\end{equation*}%
and 
\begin{eqnarray*}
\mathbf{C}_{i0}^{\ast } &=&\mathbf{C}_{i0}-\mathbf{C}_{i}(1)=\left( 
\begin{array}{cc}
1 & -\beta \\ 
0 & 0%
\end{array}%
\right) \text{,} \\
\mathbf{C}_{i\ell }^{\ast } &=&\mathbf{C}_{i,\ell -1}^{\ast }+\mathbf{C}%
_{i\ell }=\left( 
\begin{array}{cc}
\left( 1-\alpha _{i}\right) ^{\ell } & -\left( 1-\alpha _{i}\right) ^{\ell
}\beta \\ 
0 & 0%
\end{array}%
\right) \text{, for }\ell =1,2,...\text{.}
\end{eqnarray*}

Model (\ref{y})-(\ref{x}) can be equivalently written as%
\begin{equation*}
\mathbf{\Phi }_{i}\left( L\right) \mathbf{z}_{it}=\mathbf{c}_{i}+\mathbf{u}%
_{it}\text{,}
\end{equation*}%
for $i=1,2,...,n$ and $t=1,2,...,T$, where $\mathbf{c}_{i}=\left(
c_{i},0\right) ^{\prime }$, 
\begin{equation}
\mathbf{\Phi }_{i}\left( L\right) =\mathbf{I}_{2}-\mathbf{\Phi }_{i}L\text{,}
\label{Phil}
\end{equation}%
and $\mathbf{I}_{2}$ is a $2\times 2$ identity matrix. The lag polynomial $%
\mathbf{\Phi }_{i}\left( L\right) $ can be re-written in the following
(error correcting) form%
\begin{equation}
\mathbf{\Phi }_{i}\left( L\right) =-\Pi _{i}L+\left( 1-L\right) \mathbf{I}%
_{2}\text{,}
\end{equation}%
where 
\begin{equation}
\Pi _{i}=-\left( \mathbf{I}_{2}-\mathbf{\Phi }_{i}\right) =\left( 
\begin{array}{cc}
-\alpha _{i} & \alpha _{i}\beta \\ 
0 & 0%
\end{array}%
\right) \text{.}
\end{equation}%
The VAR model (\ref{1}) can be also rewritten in the following form%
\begin{equation}
\mathbf{\Phi }_{i}\left( L\right) \left( \mathbf{z}_{it}-\mathbf{\mu }%
_{i}\right) =\mathbf{u}_{it}\text{,}  \label{1}
\end{equation}
where $\mathbf{c}_{i}=-\Pi _{i}\mathbf{\mu }_{i}=\left( c_{i},0\right)
^{\prime }$, namely $c_{i}=\alpha _{i}\mu _{i,1}-\alpha _{i}\beta \mu _{i,2}$%
.

Using Granger representation theorem, the process $\mathbf{z}_{it}$ under
the assumptions \ref{As1}-\ref{As3} has representation%
\begin{eqnarray}
y_{it} &=&\mu _{yi}+\beta s_{it}+\sum_{\ell =0}^{\infty }\left( 1-\alpha
_{i}\right) ^{\ell }\left( u_{y,i,t-\ell }-\beta u_{x,i,t-\ell }\right) 
\text{,}  \label{ma_y} \\
x_{it} &=&\mu _{xi}+s_{it}\text{,}  \label{ma_x}
\end{eqnarray}%
where%
\begin{equation}
s_{it}=\sum_{\ell =1}^{t}u_{x,it}\text{,}  \label{s}
\end{equation}%
is the stochastic trend.

\subsection{Lemmas: Statements and proofs\label{A2}}

\begin{lemma}
\label{sla}Suppose Assumptions \ref{As2} and \ref{As3} hold, and consider $%
\mathbf{\tilde{x}}_{i}=\left( \tilde{x}_{i,1},\tilde{x}_{i,2},...,\tilde{x}%
_{i,T}\right) ^{\prime }$, where $\tilde{x}_{it}=x_{it}-\bar{x}_{i}$, $%
x_{it}=\sum_{s=1}^{t}u_{x,it}$, and $\bar{x}_{i}=T^{-1}\sum_{t=1}^{T}x_{it}$%
. Then%
\begin{equation}
n^{-1}\sum_{i=1}^{n}\frac{\mathbf{\tilde{x}}_{i}^{\prime }\mathbf{\tilde{x}}%
_{i}}{T^{2}}\rightarrow _{p}\omega _{x}^{2}=\frac{\sigma _{x}^{2}}{6}\text{,
as }n,T\rightarrow \infty \text{,}  \label{sl1}
\end{equation}%
where $\sigma _{x}^{2}=\lim_{n\rightarrow \infty }n^{-1}\sum_{i=1}^{n}\sigma
_{xi}^{2}$.
\end{lemma}

\begin{proof}
Recall that $\mathbf{M}_{\tau }=\mathbf{I}_{T}-T^{-1}\mathbf{\tau }_{T}%
\mathbf{\tau }_{T}^{\prime }$, where $\mathbf{I}_{T}$ is $T\times T$
identity matrix and $\mathbf{\tau }_{T}$ is $T\times 1$ vector of ones.
Since $\mathbf{\tilde{x}}_{i}=\mathbf{M}_{\tau }\mathbf{x}_{i}$, and $%
\mathbf{M}_{\tau }$ is symmetric and idempotent ($\mathbf{M}_{\tau }^{\prime
}\mathbf{M}_{\tau }=\mathbf{M}_{\tau }=\mathbf{M}_{\tau }^{\prime }$) we can
write $\mathbf{\tilde{x}}_{i}^{\prime }\mathbf{\tilde{x}}_{i}$ as $\mathbf{%
\tilde{x}}_{i}^{\prime }\mathbf{\tilde{x}}_{i}=\mathbf{x}_{i}^{\prime }%
\mathbf{M}_{\tau }^{\prime }\mathbf{M}_{\tau }\mathbf{x}_{i}=\mathbf{x}%
_{i}^{\prime }\mathbf{M}_{\tau }^{\prime }\mathbf{x}_{i}=\mathbf{\tilde{x}}%
_{i}^{\prime }\mathbf{x}_{i}$. Denote $S_{i,T}=\mathbf{\tilde{x}}%
_{i}^{\prime }\mathbf{x}_{i}/T^{2}$. We have%
\begin{equation}
n^{-1}\sum_{i=1}^{n}\frac{\mathbf{\tilde{x}}_{i}^{\prime }\mathbf{\tilde{x}}%
_{i}}{T^{2}}=n^{-1}\sum_{i=1}^{n}S_{i,T}=n^{-1}\sum_{i=1}^{n}E\left(
S_{i,T}\right) +n^{-1}\sum_{i=1}^{n}\left[ S_{i,T}-E\left( S_{i,T}\right) %
\right] \text{.}  \label{pl1}
\end{equation}%
Consider $E\left( S_{i,T}\right) $ first. Noting that $\tilde{x}%
_{it}=\sum_{s=1}^{t}u_{x,it}-\bar{x}_{i}$, $\bar{x}_{i}=T^{-1}\sum_{s=1}^{T}%
\left( T-s+1\right) u_{st}$, and $x_{it}=\sum_{s=1}^{t}u_{x,it}$, $S_{i,T}$
can be written as 
\begin{eqnarray*}
S_{i,T} &=&\frac{1}{T^{2}}\sum_{t=1}^{T}\tilde{x}_{it}x_{it}\text{,} \\
&=&\frac{1}{T^{2}}\sum_{t=1}^{T}\left[ \left( \sum_{s=1}^{t}u_{x,is}\right)
^{2}-\bar{x}_{i}\sum_{s=1}^{t}u_{x,is}\right] \text{,} \\
&=&\frac{1}{T^{2}}\sum_{t=1}^{T}\left[ \left( \sum_{s=1}^{t}u_{x,is}\right)
^{2}-\sum_{s=1}^{t}\frac{T-s+1}{T}u_{x,is}\cdot \sum_{s=1}^{t}u_{x,is}\right]
\text{.}
\end{eqnarray*}%
Taking expectations, we obtain%
\begin{equation*}
E\left( S_{i,T}\right) =\frac{\sigma _{xi}^{2}}{T^{2}}\sum_{t=1}^{T}\left[
t-\sum_{s=1}^{t}\frac{T-s+1}{T}\right] \text{.}
\end{equation*}%
Using $\sum_{s=1}^{t}\frac{T-s+1}{T}=\sum_{s=1}^{t}\left( 1-s/T+1/T\right)
=t-\left( t+1\right) t/\left( 2T\right) +t/T$, we have%
\begin{equation*}
E\left( S_{i,T}\right) =\frac{\sigma _{xi}^{2}}{T^{2}}\sum_{t=1}^{T}\left[
t-t+\frac{\left( t+1\right) t}{2T}-\frac{t}{T}\right] =\frac{\sigma _{xi}^{2}%
}{T^{2}}\sum_{t=1}^{T}\frac{\left( t+1\right) t}{2T}-\frac{t}{T}\text{.}
\end{equation*}%
Finally, noting that $\sum_{t=1}^{T}\left( t+1\right) t=\left( T+2\right)
\left( T+1\right) T/3$, and $\sum_{t=1}^{T}t=\left( T+1\right) T/2$, we
obtain%
\begin{equation}
E\left( S_{i,T}\right) =\sigma _{xi}^{2}\varkappa _{T}<K<\infty \text{,}
\label{plb}
\end{equation}%
for all $T>0$, where 
\begin{equation}
\varkappa _{T}=\left[ \frac{\left( T+2\right) \left( T+1\right) T}{6T^{3}}-%
\frac{\left( T+1\right) T}{2T^{3}}\right] \text{.}  \label{kt}
\end{equation}%
In addition, $\varkappa _{T}\rightarrow 1/6$, as $T\rightarrow \infty $, and 
\begin{equation*}
\frac{1}{n}\sum_{i=1}^{n}E\left( S_{i,T}\right) =\varkappa _{T}\frac{1}{n}%
\sum_{i=1}^{n}\sigma _{xi}^{2}\rightarrow \frac{\sigma _{x}^{2}}{6}\text{,}
\end{equation*}%
as $n,T\rightarrow \infty $. This establishes the limit of the first term on
the right side of (\ref{pl1}). Consider the second term next. Since $E\left[
S_{i,T}-E\left( S_{i,T}\right) \right] =0$, and $S_{i,T}$ is independent
over $i$, we have%
\begin{equation*}
E\left\{ n^{-1}\sum_{i=1}^{n}\left[ S_{i,T}-E\left( S_{i,T}\right) \right]
\right\} ^{2}=\frac{1}{n^{2}}\sum_{i=1}^{n}E\left( S_{i,T}^{2}\right) -\frac{%
1}{n^{2}}\sum_{i=1}^{n}\left[ E\left( S_{i,T}\right) \right] ^{2}\text{.}
\end{equation*}%
But it follows from (\ref{plb}) that there exist finite positive constant $%
K_{1}<\infty $ (which does not depend on $n,T$) such that $\left[ E\left(
S_{i,T}\right) \right] ^{2}<K_{1}$. In addition, due to existence of
uniformly bounded fourth moments of $u_{x,it}$, it also can be shown that $%
E\left( S_{i,T}^{2}\right) <K_{2}<\infty $. Hence, $E\left\{
n^{-1}\sum_{i=1}^{n}\left[ S_{i,T}-E\left( S_{i,T}\right) \right] \right\}
^{2}=O\left( n^{-1}\right) $, which implies $n^{-1}\sum_{i=1}^{n}\left[
S_{i,T}-E\left( S_{i,T}\right) \right] \rightarrow _{p}0$, as $%
n,T\rightarrow \infty $. This completes the proof.
\end{proof}

\begin{lemma}
\label{lb1}Suppose Assumptions \ref{As1}-\ref{As2} hold. Then there exists
finite positive constant $K$ that does not depend on $i$ and/or $T$ such that%
\begin{equation}
E\left( \frac{1}{T}\sum_{t=1}^{T}u_{x,it}\tilde{x}_{it}\right) ^{\varrho }<K%
\text{,}  \label{ba}
\end{equation}%
and%
\begin{equation}
E\left( \frac{1}{T}\sum_{t=1}^{T}\Delta y_{it}\tilde{x}_{it}\right)
^{\varrho }<K\text{,}  \label{bb}
\end{equation}%
for $\varrho =4$, where $\tilde{x}_{it}=x_{it}-\bar{x}_{i}$, $%
x_{it}=\sum_{s=1}^{t}u_{x,it}$, $\bar{x}_{i}=T^{-1}\sum_{t=1}^{T}x_{it}$,
and $\Delta y_{it}=\delta _{i}u_{x,it}+v_{it}-\alpha _{i}\sum_{\ell
=1}^{\infty }\left( 1-\alpha _{i}\right) ^{\ell -1}\left[ v_{i,t-\ell
}+\left( \delta _{i}-\beta \right) u_{x,i,t-\ell }\right] $.
\end{lemma}

\begin{proof}
Consider $\sum_{t=1}^{T}u_{it}\tilde{x}_{it}/T$ and $\varrho =2$ first, and
note that $\tilde{x}_{it}=\sum_{s=1}^{t}u_{x,is}-\bar{x}_{i}$, where $\bar{x}%
_{i}=T^{-1}\sum_{s=1}^{T}\left( T-s+1\right) u_{x,is}$. We have%
\begin{eqnarray*}
\left( \frac{1}{T}\sum_{t=1}^{T}u_{x,it}\tilde{x}_{it}\right) ^{2} &=&\frac{1%
}{T^{2}}\sum_{t=1}^{T}\sum_{t^{\prime }=1}^{T}u_{x,it}u_{x,it^{\prime }}%
\tilde{x}_{it}\tilde{x}_{it^{\prime }}\text{,} \\
&=&\frac{1}{T^{2}}\sum_{t=1}^{T}\sum_{t^{\prime
}=1}^{T}u_{x,it}u_{x,it^{\prime }}\left( \sum_{s=1}^{t}u_{x,is}-\bar{x}%
_{i}\right) \left( \sum_{s=1}^{t^{\prime }}u_{x,is}-\bar{x}_{i}\right) \text{%
,} \\
&=&A_{i,T,1}+A_{i,T,2}-A_{i,T,3}-A_{i,T,4}\text{,}
\end{eqnarray*}%
where 
\begin{eqnarray*}
A_{i,T,1} &=&\frac{1}{T^{2}}\sum_{t=1}^{T}\sum_{t^{\prime
}=1}^{T}u_{x,it}u_{x,it^{\prime }}\left( \sum_{s=1}^{t}u_{x,is}\right)
\left( \sum_{s=1}^{t^{\prime }}u_{x,is}\right) \text{,} \\
A_{i,T,2} &=&\frac{1}{T^{2}}\sum_{t=1}^{T}\sum_{t^{\prime
}=1}^{T}u_{x,it}u_{x,it^{\prime }}\bar{x}_{i}^{2}\text{,} \\
A_{i,T,3} &=&\frac{1}{T^{2}}\sum_{t=1}^{T}\sum_{t^{\prime
}=1}^{T}u_{x,it}u_{x,it^{\prime }}\bar{x}_{i}\sum_{s=1}^{t}u_{x,is}\text{,}
\\
A_{i,T,4} &=&\frac{1}{T^{2}}\sum_{t=1}^{T}\sum_{t^{\prime
}=1}^{T}u_{x,it}u_{x,it^{\prime }}\bar{x}_{i}\sum_{s=1}^{t^{\prime }}u_{x,is}%
\text{.}
\end{eqnarray*}%
Taking expectations and noting that $u_{x,it}$ is independent of $%
u_{x,it^{\prime }}$ for any $t\neq t^{\prime }$, we have%
\begin{equation*}
E\left( A_{i,T,1}\right) =\frac{1}{T^{2}}\left(
\sum_{t=1}^{T}\sum_{t^{\prime }=1}^{t-1}\sigma
_{ix}^{2}+\sum_{t=1}^{T}E\left( u_{x,it}^{4}\right)
+\sum_{t=1}^{T}\sum_{t^{\prime }=1}^{t-1}\sigma _{ix}^{2}\right) \text{.}
\end{equation*}%
Under Assumption \ref{As2}, there exists a finite constant $K$ that does not
depend on $i$ and/or $t$, such that $\sigma _{ix}^{2}<K$ and $E\left(
u_{x,it}^{4}\right) <K$. Hence $\left\vert E\left( A_{i,T,1}\right)
\right\vert <K$. Similarly, we can bound the remaining elements, $\left\vert
E\left( A_{i,T,j}\right) \right\vert <K$, for $j=2,3,4$. It now follows that 
$E\left( \frac{1}{T}\sum_{t=1}^{T}u_{x,it}\tilde{x}_{it}\right) ^{2}<K$,
where the upper bound $K$ does not depend on $i$ or $T$. This establishes (%
\ref{ba}) hold for $\varrho =2$. Sufficient conditions for (\ref{ba}) to
hold when $\varrho =4$ are: $E\left( A_{i,T,j}^{2}\right) <K$ for $j=1,2,3,4$%
. These conditions follow from uniformly bounded eights moments of $u_{x,it}$%
. This completes the proof of (\ref{ba}). Result (\ref{bb}) can be
established in the same way by using the first difference of representation (%
\ref{ma_y}).
\end{proof}

\bigskip

\begin{lemma}
\label{slb}Suppose Assumptions \ref{As1}-\ref{As4} hold, and consider $%
s_{iT} $ given by 
\begin{equation}
s_{iT}=\mathbf{\tilde{x}}_{i}^{\prime }\Delta \mathbf{\tilde{Z}}_{i}\left(
\Delta \mathbf{\tilde{Z}}_{i}^{\prime }\mathbf{P}_{i}\Delta \mathbf{\tilde{Z}%
}_{i}\right) ^{-1}\Delta \mathbf{\tilde{Z}}_{i}^{\prime }\mathbf{\tilde{x}}%
_{i}\text{,}  \label{sit}
\end{equation}%
where $\mathbf{P}_{i}$ is given by (\ref{pi}), and $\mathbf{\tilde{x}}_{i}$
and $\Delta \mathbf{\tilde{Z}}_{i}$ are defined below (\ref{i}). Then, 
\begin{equation}
n^{-1}\sum_{i=1}^{n}\frac{s_{iT}}{T^{2}}\rightarrow _{p}0\text{, as }%
n,T\rightarrow \infty \text{. }  \label{rb1}
\end{equation}
\end{lemma}

\begin{proof}
Consider $s_{i,T}/T$, which can be written as 
\begin{equation}
\frac{s_{iT}}{T}\leq \mathbf{a}_{iT}^{\prime }\mathbf{B}_{iT}^{-1}\mathbf{a}%
_{iT}\text{,}  \label{pb1}
\end{equation}%
where 
\begin{equation}
\mathbf{a}_{iT}=\frac{\Delta \mathbf{\tilde{Z}}_{i}^{\prime }\mathbf{\tilde{x%
}}_{i}}{T}=\frac{\Delta \mathbf{Z}_{i}^{\prime }\mathbf{\tilde{x}}_{i}}{T}%
\text{,}  \label{ait}
\end{equation}%
and%
\begin{equation}
\mathbf{B}_{iT}=\frac{\Delta \mathbf{\tilde{Z}}_{i}^{\prime }\mathbf{P}%
_{i}\Delta \mathbf{\tilde{Z}}_{i}}{T}\text{.}  \label{bit}
\end{equation}%
Using these notations, we have%
\begin{equation*}
E\left\vert \frac{1}{n}\sum_{i=1}^{n}\frac{s_{iT}}{T^{2}}\right\vert \leq 
\frac{1}{nT}\sum_{i=1}^{n}E\left\vert \mathbf{a}_{iT}^{\prime }\mathbf{B}%
_{iT}^{-1}\mathbf{a}_{iT}\right\vert \text{.}
\end{equation*}%
Using $\left\vert \mathbf{a}_{iT}^{\prime }\mathbf{B}_{iT}^{-1}\mathbf{a}%
_{iT}\right\vert \leq \lambda _{\min }^{-1}\left( \mathbf{B}_{iT}\right) 
\mathbf{a}_{iT}^{\prime }\mathbf{a}_{iT}$, and Cauchy-Schwarz inequality, we
obtain%
\begin{equation*}
E\left\vert \frac{1}{n}\sum_{i=1}^{n}\frac{s_{iT}}{T^{2}}\right\vert \leq 
\frac{1}{nT}\sum_{i=1}^{n}\sqrt{E\left[ \left( \mathbf{a}_{iT}^{\prime }%
\mathbf{a}_{iT}\right) ^{2}\right] }\sqrt{E\left[ \lambda _{\min
}^{-2}\left( \mathbf{B}_{iT}\right) \right] }\text{.}
\end{equation*}%
Lemma \ref{lb1} implies the fourth moments of the individual elements of $%
\mathbf{a}_{i,T}$ are uniformly bounded in $i$ and $T$, which is sufficient
for $E\left[ \left( \mathbf{a}_{iT}^{\prime }\mathbf{a}_{iT}\right) ^{2}%
\right] <K$. In addition, $E\left[ \lambda _{\min }^{-2}\left( \mathbf{B}%
_{iT}\right) \right] <K$ by Assumption \ref{As4}. Hence, there exists $%
K<\infty $, which does not depend on $\left( n,T\right) $ such that 
\begin{equation}
E\left\vert n^{-1}\sum_{i=1}^{n}\frac{s_{iT}}{T^{2}}\right\vert <\frac{K}{T}%
\text{,}
\end{equation}%
and result (\ref{rb1}) follows.
\end{proof}

\bigskip

\begin{lemma}
\label{slp}Suppose Assumptions \ref{As1}-\ref{As4} hold. Then%
\begin{equation}
n^{-1}\sum_{i=1}^{n}\frac{\mathbf{\tilde{x}}_{i}^{\prime }\mathbf{M}_{i}%
\mathbf{\tilde{x}}_{i}}{T^{2}}\rightarrow _{p}\omega _{x}^{2}=\frac{\sigma
_{x}^{2}}{6}\text{, as }n,T\rightarrow \infty ,  \label{cr1}
\end{equation}%
where $\sigma _{x}^{2}=\lim_{n\rightarrow \infty }n^{-1}\sum_{i=1}^{n}\sigma
_{xi}^{2}$, $\mathbf{M}_{i}$ is defined in (\ref{mi}) and $\mathbf{\tilde{x}}%
_{i}$ is defined below (\ref{i}).
\end{lemma}

\begin{proof}
Noting that $\mathbf{\tilde{x}}_{i}$ is one of the column vectors of $%
\mathbf{H}_{i}$, we have $\mathbf{P}_{i}\mathbf{\tilde{x}}_{i}=\mathbf{%
\tilde{x}}_{i}$, and $\mathbf{\tilde{x}}_{i}^{\prime }\mathbf{M}_{i}\mathbf{%
\tilde{x}}_{i}$ can be written as 
\begin{equation}
\mathbf{\tilde{x}}_{i}^{\prime }\mathbf{M}_{i}\mathbf{\tilde{x}}_{i}=\mathbf{%
\tilde{x}}_{i}^{\prime }\mathbf{\tilde{x}}_{i}-s_{i,T}\text{,}
\end{equation}%
where $s_{i,T}$ is given by (\ref{sit}). Sufficient conditions for result (%
\ref{cr1}) are:%
\begin{equation}
n^{-1}\sum_{i=1}^{n}\frac{\mathbf{\tilde{x}}_{i}^{\prime }\mathbf{\tilde{x}}%
_{i}}{T^{2}}\rightarrow _{p}\omega _{x}^{2}=\frac{\sigma _{x}^{2}}{6}\text{,
as }n,T\rightarrow \infty ,  \label{rs1}
\end{equation}%
and%
\begin{equation}
n^{-1}\sum_{i=1}^{n}\frac{s_{i,T}}{T^{2}}\rightarrow _{p}0\text{, as }%
n,T\rightarrow \infty \text{. }  \label{rs2}
\end{equation}%
Condition (\ref{rs1}) is established by Lemma \ref{sla}, and condition (\ref%
{rs2}) is established by Lemma \ref{slb}.
\end{proof}

\begin{lemma}
\label{sld}Let Assumptions \ref{As1}-\ref{As3} hold. Then%
\begin{equation}
\frac{1}{\sqrt{n}}\sum_{i=1}^{n}\frac{\mathbf{\tilde{x}}_{i}^{\prime }%
\mathbf{\tilde{v}}_{i}}{T\alpha _{i}}\rightarrow _{d}N\left( 0,\omega
_{v}^{2}\right) \text{, as }n,T\rightarrow \infty \text{,}
\end{equation}%
where $\omega _{v}^{2}=\lim {}_{n\rightarrow \infty
}n^{-1}\sum_{i=1}^{n}\sigma _{xi}^{2}\sigma _{vi}^{2}/\left( 6\alpha
_{i}^{2}\right) $, and $\mathbf{\tilde{x}}_{i}$ and $\mathbf{\tilde{v}}_{i}$
are defined below (\ref{i}).
\end{lemma}

\begin{proof}
Recall $\mathbf{M}_{\tau }=\mathbf{I}_{T}-T^{-1}\mathbf{\tau }_{T}\mathbf{%
\tau }_{T}^{\prime }$, where $\mathbf{I}_{T}$ is $T\times T$ identity matrix
and $\mathbf{\tau }_{T}$ is $T\times 1$ vector of ones. Since $\mathbf{M}%
_{\tau }^{\prime }\mathbf{M}_{\tau }=\mathbf{M}_{\tau }^{\prime }$,\textbf{\ 
}we have%
\begin{equation*}
\mathbf{\tilde{x}}_{i}^{\prime }\mathbf{\tilde{v}}_{i}=\mathbf{x}%
_{i}^{\prime }\mathbf{M}_{\tau }^{\prime }\mathbf{M}_{\tau }\mathbf{v}_{i}=%
\mathbf{x}_{i}^{\prime }\mathbf{M}_{\tau }^{\prime }\mathbf{v}_{i}=\mathbf{%
\tilde{x}}_{i}^{\prime }\mathbf{v}_{i}\text{.}
\end{equation*}%
Let $C_{i}=\frac{\sigma _{xi}\sigma _{vi}}{\alpha _{i}}$ and $%
Q_{i,T}=C_{i}^{-1}\frac{\mathbf{\tilde{x}}_{i}^{\prime }\mathbf{v}_{i}}{%
T\alpha _{i}}$. We have $E\left( Q_{i,T}\right) =0$, and (under independence
of $v_{it}$ over $t$ and independence of $v_{it}$ and $u_{x,it^{\prime }}$
for any $t,t^{\prime }$) 
\begin{equation*}
E\left[ \left( \frac{\mathbf{\tilde{x}}_{i}^{\prime }\mathbf{v}_{i}}{T\alpha
_{i}}\right) ^{2}\right] =\frac{1}{T^{2}\alpha _{i}^{2}}\sum_{t=1}^{T}E%
\left( \tilde{x}_{it}^{2}\right) E\left( v_{it}^{2}\right) \text{,}
\end{equation*}%
where $E\left( v_{it}^{2}\right) =\sigma _{vi}^{2}$. In addition, (\ref{plb}%
) established that $\frac{1}{T^{2}}\sum_{t=1}^{T}E\left( \tilde{x}%
_{it}^{2}\right) =\sigma _{xi}^{2}\varkappa _{T}$, where $\varkappa _{T}$ is
given by (\ref{kt}). Hence,%
\begin{equation*}
E\left[ \left( \frac{\mathbf{\tilde{x}}_{i}^{\prime }\mathbf{v}_{i}}{T\alpha
_{i}}\right) ^{2}\right] =\frac{\sigma _{vi}^{2}\sigma _{xi}^{2}}{\alpha
_{i}^{2}}\varkappa _{T}=C_{i}^{2}\varkappa _{T}\text{.}
\end{equation*}%
It follows that 
\begin{equation*}
E\left( Q_{i,T}^{2}\right) =\varkappa _{T}\text{,}
\end{equation*}%
where $\varkappa _{T}\rightarrow 1/6<\infty $. Finite fourth moments of $%
u_{x,it}$ and $v_{it}$ imply $Q_{i,T}^{4}$ is uniformly bounded in $T$, and
therefore $Q_{i,T}^{2}$ is uniformly integrable in $T$. We can apply Theorem
3 of 
\citeN{PhillipsMoon1999}
to obtain%
\begin{equation*}
\frac{1}{\sqrt{n}}\sum_{i=1}^{n}C_{i}Q_{i,T}=\frac{1}{\sqrt{n}}\sum_{i=1}^{n}%
\frac{\mathbf{\tilde{x}}_{i}^{\prime }\mathbf{v}_{i}}{T\alpha _{i}}%
\rightarrow _{d}N\left( 0,\omega _{v}^{2}\right) \text{, as }n,T\rightarrow
\infty \text{,}
\end{equation*}%
where $\omega _{v}^{2}=\lim_{n\rightarrow \infty }C_{i}^{2}\varkappa
_{T}=\lim {}_{n\rightarrow \infty }n^{-1}\sum_{i=1}^{n}\sigma
_{xi}^{2}\sigma _{vi}^{2}/\left( 6\alpha _{i}^{2}\right) $.
\end{proof}

\bigskip

\begin{lemma}
\label{lra} Suppose Assumptions \ref{As1}-\ref{As4} hold, and consider $%
\mathbf{q}_{iT}=\alpha _{i}^{-1}\Delta \mathbf{\tilde{Z}}_{i}^{\prime }%
\mathbf{P}_{i}\mathbf{\tilde{v}}_{i}/\sqrt{T}$. Then, 
\begin{equation}
E\left\Vert \mathbf{q}_{iT}\right\Vert _{2}^{4}<K\text{,}  \label{rap}
\end{equation}%
and%
\begin{equation}
\left\vert E\left( \mathbf{q}_{iT}\right) \right\vert <\frac{K}{\sqrt{T}}%
\text{.}  \label{rbp}
\end{equation}
\end{lemma}

\begin{proof}
Denote the individual elements of $2\times 1$ vector $\mathbf{q}_{iT}$ as $%
q_{iT,j}$, $j=1,2$. Sufficient conditions for (\ref{rap}) to hold are 
\begin{equation}
E\left( q_{iT,j}\right) ^{4}<K\text{, for }j=1,2\text{.}  \label{ra1}
\end{equation}%
We establish (\ref{ra1}) for $j=1$ first. We have%
\begin{equation*}
q_{iT,1}=\frac{\Delta \mathbf{\tilde{y}}_{i}^{\prime }\mathbf{P}_{i}\mathbf{%
\tilde{v}}_{i}}{\alpha _{i}\sqrt{T}}\text{,}
\end{equation*}%
where $\Delta \mathbf{y}_{i}$ can be written as%
\begin{equation}
\Delta \mathbf{\tilde{y}}_{i}=-\alpha _{i}\mathbf{\tilde{\xi}}_{i,-1}+\delta
_{i}\Delta \mathbf{\tilde{x}}_{i}+\mathbf{v}_{i}\text{,}  \label{yd}
\end{equation}%
where $\mathbf{\tilde{\xi}}_{i,-1}=\mathbf{\tilde{y}}_{i,-1}-\mathbf{\tilde{x%
}}_{i,-1}$. Note that $\mathbf{P}_{i}=\mathbf{\tilde{H}}_{i}\left( \mathbf{%
\tilde{H}}_{i}^{\prime }\mathbf{\tilde{H}}_{i}\right) ^{-1}\mathbf{\tilde{H}}%
_{i}^{\prime }$ and $\mathbf{\tilde{H}}_{i}=\left( \mathbf{\tilde{y}}_{i,-1},%
\mathbf{\tilde{x}}_{i},\mathbf{\tilde{x}}_{i,-1}\right) $. Hence $\Delta 
\mathbf{\tilde{x}}_{i}^{\prime }\mathbf{P}_{i}=\Delta \mathbf{\tilde{x}}%
_{i}^{\prime }$ and $\mathbf{\tilde{\xi}}_{i,-1}^{\prime }\mathbf{P}_{i}=%
\mathbf{\tilde{\xi}}_{i,-1}^{\prime }$, since $\Delta \mathbf{\tilde{x}}_{i}$
and $\mathbf{\tilde{\xi}}_{i,-1}$ can be both obtained as a linear
combinations of the column vectors of $\mathbf{\tilde{H}}_{i}$. Hence%
\begin{equation}
q_{iT,1}=-\frac{\mathbf{\tilde{\xi}}_{i,-1}^{\prime }\mathbf{\tilde{v}}_{i}}{%
\sqrt{T}}+\frac{\Delta \mathbf{\tilde{x}}_{i}^{\prime }\mathbf{\tilde{v}}_{i}%
}{\alpha _{i}\sqrt{T}}+\frac{\mathbf{\tilde{v}}_{i}^{\prime }\mathbf{P}_{i}%
\mathbf{\tilde{v}}_{i}}{\alpha _{i}\sqrt{T}}\equiv \varsigma
_{a,iT}+\varsigma _{b,iT}+\varsigma _{c,iT}\text{,}  \label{q1e}
\end{equation}%
where we simplified notations by introducing $\varsigma _{a,iT}=-\mathbf{%
\tilde{\xi}}_{i,-1}^{\prime }\mathbf{\tilde{v}}_{i}/\sqrt{T}$, $\varsigma
_{b,iT}=\alpha _{i}^{-1}\Delta \mathbf{\tilde{x}}_{i}^{\prime }\mathbf{%
\tilde{v}}_{i}/\sqrt{T}$ and $\varsigma _{c,iT}=\alpha _{i}^{-1}\Delta 
\mathbf{\tilde{v}}_{i}^{\prime }\mathbf{P}_{i}\mathbf{\tilde{v}}_{i}/\sqrt{T}
$ to denote the individual terms in the expression (\ref{q1e}) for $q_{iT,1}$%
. Sufficient conditions for $E\left( q_{iT,1}^{4}\right) <K$ are $E\left(
\varsigma _{s,iT}^{4}\right) <K$ for $s\in \left\{ a,b,c\right\} $.

For $s=a$, we have%
\begin{equation*}
\varsigma _{a,iT}=-\frac{\mathbf{\tilde{\xi}}_{i,-1}^{\prime }\mathbf{\tilde{%
v}}_{i}}{\sqrt{T}}=-\frac{1}{\sqrt{T}}\sum_{i=1}^{T}\left( \xi _{i,t-1}-\bar{%
\xi}_{i,-1}\right) \left( v_{it}-\bar{v}_{i}\right) =-\frac{1}{\sqrt{T}}%
\sum_{i=1}^{T}\xi _{i,t-1}v_{it}+\sqrt{T}\bar{\xi}_{i,-1}\bar{v}_{i}\text{,}
\end{equation*}%
where $\bar{\xi}_{i,-1}=T^{-1}\sum_{t=1}^{T}\xi _{i,t-1}$, and%
\begin{eqnarray*}
\xi _{it} &=&\sum_{\ell =0}^{\infty }\left( 1-\alpha _{i}\right) ^{\ell
}\left( u_{y,i,t-\ell }-\beta u_{x,i,t-\ell }\right) \text{,} \\
&=&\sum_{\ell =0}^{\infty }\left( 1-\alpha _{i}\right) ^{\ell }\left( \delta
_{i}-\beta \right) u_{x,i,t-\ell }+\sum_{\ell =0}^{\infty }\left( 1-\alpha
_{i}\right) ^{\ell }v_{it}\text{.}
\end{eqnarray*}%
Noting that $\sup_{i}\left\vert 1-\alpha _{i}\right\vert <1$ under
Assumption \ref{As1}, and fourth moments of $u_{x,i,t}$ and eights moments
of $v_{it}$ are bounded, we obtain%
\begin{equation*}
E\left[ \left( \frac{1}{\sqrt{T}}\sum_{i=1}^{T}\xi _{i,t-1}v_{it}\right) ^{4}%
\right] \leq K\text{,}
\end{equation*}%
and%
\begin{equation*}
T^{2}\cdot E\left( \bar{\xi}_{i,-1}^{4}\bar{v}_{i}^{4}\right) \leq K\text{,}
\end{equation*}%
which are sufficient conditions for $E\left( \varsigma _{a,iT}^{4}\right)
\leq K$.

For $s=b$, we have%
\begin{equation*}
\varsigma _{b,iT}=\frac{\Delta \mathbf{\tilde{x}}_{i}^{\prime }\mathbf{%
\tilde{v}}_{i}}{\alpha _{i}\sqrt{T}}=\frac{1}{\alpha _{i}\sqrt{T}}%
\sum_{t=1}^{T}\left( u_{x,it}-\bar{u}_{x,i}\right) \left( v_{it}-\bar{v}%
_{i}\right) =\frac{1}{\alpha _{i}\sqrt{T}}\sum_{t=1}^{T}u_{x,it}v_{it}-\frac{%
\sqrt{T}}{\alpha _{i}}\bar{u}_{x,i}\bar{v}_{i}\text{.}
\end{equation*}%
Using Assumption \ref{As2}, we obtain the following upper bound%
\begin{equation}
E\left( \varsigma _{b,iT}^{4}\right) \leq \left\vert \alpha
_{i}^{-4}\right\vert \frac{1}{T}\sum_{t=1}^{T}E\left( u_{x,it}^{4}\right)
E\left( v_{it}^{4}\right) +\left\vert \alpha _{i}^{-4}\right\vert TE\left( 
\bar{u}_{x,i}^{4}\right) E\left( \bar{v}_{i}^{4}\right) \leq K\text{,}
\label{rvb}
\end{equation}%
where $\left\vert \alpha _{i}^{-4}\right\vert <K$, $E\left(
u_{x,it}^{4}\right) <K$, $E\left( v_{it}^{4}\right) <K$, $E\left( \bar{u}%
_{x,i}^{4}\right) <K/T^{2}$, and $E\left( \bar{v}_{i}^{4}\right) <K/T^{2}$.

For $s=c$, we have%
\begin{equation*}
\varsigma _{c,iT}=\frac{\mathbf{\tilde{v}}_{i}^{\prime }\mathbf{P}_{i}%
\mathbf{\tilde{v}}_{i}}{\alpha _{i}\sqrt{T}}=\frac{\Delta \mathbf{\tilde{v}}%
_{i}^{\prime }\mathbf{\tilde{H}}_{i}\left( \mathbf{\tilde{H}}_{i}^{\prime }%
\mathbf{\tilde{H}}_{i}\right) ^{-1}\mathbf{\tilde{H}}_{i}^{\prime }\mathbf{%
\tilde{v}}_{i}}{\alpha _{i}\sqrt{T}}\text{.}
\end{equation*}%
Consider $\mathbf{\tilde{H}}_{i}^{\ast }=\left( \mathbf{\tilde{x}}%
_{i},\Delta \mathbf{\tilde{x}}_{i},\mathbf{\tilde{\xi}}_{i,-1}\right) $ and
note that%
\begin{equation*}
\mathbf{\tilde{H}}_{i}=\left( \mathbf{\tilde{y}}_{i,-1},\mathbf{\tilde{x}}%
_{i},\mathbf{\tilde{x}}_{i,-1}\right) =\mathbf{B}^{\ast }\mathbf{\tilde{H}}%
_{i}^{\ast }\text{,}
\end{equation*}%
where 
\begin{equation*}
\mathbf{B}^{\ast }=\left( 
\begin{array}{ccc}
-\beta & \beta & 1 \\ 
1 & 0 & 0 \\ 
-1 & 1 & 0%
\end{array}%
\right) \text{,}
\end{equation*}%
is nonsingular (for any $\beta $). Hence $\mathbf{P}_{i}=\mathbf{\tilde{H}}%
_{i}\left( \mathbf{\tilde{H}}_{i}^{\prime }\mathbf{\tilde{H}}_{i}\right)
^{-1}\mathbf{\tilde{H}}_{i}^{\prime }=\mathbf{\tilde{H}}_{i}^{\ast }\left( 
\mathbf{\tilde{H}}_{i}^{\ast \prime }\mathbf{\tilde{H}}_{i}^{\ast }\right)
^{-1}\mathbf{\tilde{H}}_{i}^{\ast \prime }$, and we can write $\varsigma
_{c,iT}$ as%
\begin{equation*}
\varsigma _{c,iT}=\frac{\mathbf{\tilde{v}}_{i}^{\prime }\mathbf{\tilde{H}}%
_{i}^{\ast }\left( \mathbf{\tilde{H}}_{i}^{\ast \prime }\mathbf{\tilde{H}}%
_{i}^{\ast }\right) ^{-1}\mathbf{\tilde{H}}_{i}^{\ast \prime }\mathbf{\tilde{%
v}}_{i}}{\alpha _{i}\sqrt{T}}\text{.}
\end{equation*}%
Consider the scaling matrix 
\begin{equation}
\mathbf{A}_{T}=\left( 
\begin{array}{ccc}
T^{-1} & 0 & 0 \\ 
0 & T^{-1/2} & 0 \\ 
0 & 0 & T^{-1/2}%
\end{array}%
\right) \text{.}
\end{equation}%
We have%
\begin{equation*}
\varsigma _{c,iT}=\frac{1}{\alpha _{i}\sqrt{T}}\mathbf{\tilde{v}}%
_{i}^{\prime }\mathbf{\tilde{H}}_{i}^{\ast }\mathbf{A}_{T}\left( \mathbf{A}%
_{T}\mathbf{\tilde{H}}_{i}^{\ast \prime }\mathbf{\tilde{H}}_{i}^{\ast }%
\mathbf{A}_{T}\right) ^{-1}\mathbf{A}_{T}\mathbf{\tilde{H}}_{i}^{\ast \prime
}\mathbf{\tilde{v}}_{i}\geq 0\text{.}
\end{equation*}%
Using the inequality $\mathbf{x}^{\prime }\mathbf{A}^{-1}\mathbf{x}\leq
\lambda _{\min }\left( \mathbf{A}\right) \left\Vert \mathbf{x}\right\Vert
^{2}$, we have%
\begin{equation*}
0\leq \varsigma _{c,iT}\leq \frac{1}{\alpha _{i}\sqrt{T}}\lambda _{\min
}^{-1}\left( \mathbf{A}_{T}\mathbf{\tilde{H}}_{i}^{\ast \prime }\mathbf{%
\tilde{H}}_{i}^{\ast }\mathbf{A}_{T}\right) \left\Vert \mathbf{A}_{T}\mathbf{%
\tilde{H}}_{i}^{\ast \prime }\mathbf{\tilde{v}}_{i}\right\Vert _{2}^{2}\text{%
.}
\end{equation*}%
Using Cauchy-Schwarz inequality, we obtain%
\begin{equation*}
E\left( \varsigma _{c,iT}^{4}\right) \leq \frac{1}{\alpha _{i}^{4}T^{2}}%
\sqrt{E\left[ \lambda _{\min }^{-4}\left( \mathbf{A}_{T}\mathbf{\tilde{H}}%
_{i}^{\ast \prime }\mathbf{\tilde{H}}_{i}^{\ast }\mathbf{A}_{T}\right) %
\right] }\sqrt{E\left\Vert \mathbf{A}_{T}\mathbf{\tilde{H}}_{i}^{\ast \prime
}\mathbf{\tilde{v}}_{i}\right\Vert _{2}^{8}}\text{.}
\end{equation*}%
But $\alpha _{i}^{-4}<K$ under Assumption \ref{As1}, and $E\left[ \lambda
_{\min }^{-4}\left( \mathbf{A}_{T}\mathbf{\tilde{H}}_{i}^{\ast \prime }%
\mathbf{\tilde{H}}_{i}^{\ast }\mathbf{A}_{T}\right) \right] <K$ under
Assumption \ref{As4}. It follows%
\begin{equation*}
E\left( \varsigma _{c,iT}^{4}\right) \leq \frac{K}{T^{2}}\sqrt{E\left\Vert 
\mathbf{A}_{T}\mathbf{\tilde{H}}_{i}^{\ast \prime }\mathbf{\tilde{v}}%
_{i}\right\Vert _{2}^{8}}\text{.}
\end{equation*}%
Let $\mathbf{A}_{T}\mathbf{\tilde{H}}_{i}^{\ast \prime }\mathbf{\tilde{v}}%
_{i}\equiv \mathbf{h}_{viT}$ and consider the individual elements of $%
\mathbf{h}_{viT}$, denoted as $h_{viT,j}$ for $j=1,2,3$,%
\begin{equation*}
\mathbf{h}_{viT}=\mathbf{A}_{T}\mathbf{\tilde{H}}_{i}^{\ast \prime }\mathbf{%
\tilde{v}}_{i}=\left( 
\begin{array}{c}
h_{viT,1} \\ 
h_{viT,2} \\ 
h_{viT,3}%
\end{array}%
\right) =\left( 
\begin{array}{c}
\frac{1}{T}\sum_{t=1}^{T}\tilde{x}_{it}\tilde{v}_{it} \\ 
\frac{1}{\sqrt{T}}\sum_{t=1}^{T}\tilde{u}_{it}\tilde{v}_{it} \\ 
\frac{1}{\sqrt{T}}\sum_{t=1}^{T}\bar{\xi}_{i,t-1}\tilde{v}_{it}%
\end{array}%
\right) \text{.}
\end{equation*}%
Under Assumption \ref{As2}, it can be shown that%
\begin{equation*}
E\left( h_{viT,j}^{8}\right) <K\text{, for }j=1,2,3\text{,}
\end{equation*}%
which is sufficient for $E\left\Vert \mathbf{A}_{T}\mathbf{\tilde{H}}%
_{i}^{\ast \prime }\mathbf{\tilde{v}}_{i}\right\Vert _{2}^{8}<K$. It follows
that 
\begin{equation}
E\left( \varsigma _{c,iT}^{4}\right) <\frac{K}{T^{2}}\text{.}  \label{zc}
\end{equation}%
This completes the proof of (\ref{rap}) for $j=1$. Consider next (\ref{rap})
for $j=2$, and note $q_{iT,2}$ is the same as $\varsigma _{b,iT}$, namely 
\begin{equation*}
q_{iT,2}=\frac{\Delta \mathbf{\tilde{x}}_{i}^{\prime }\mathbf{P}_{i}\mathbf{%
\tilde{v}}_{i}}{\alpha _{i}\sqrt{T}}=\frac{\Delta \mathbf{\tilde{x}}%
_{i}^{\prime }\mathbf{\tilde{v}}_{i}}{\alpha _{i}\sqrt{T}}=\varsigma _{b,iT}%
\text{.}
\end{equation*}%
But $E\left( \varsigma _{b,iT}^{4}\right) <K$, see (\ref{rvb}). This
completes the proof of (\ref{rap}).

Next we establish (\ref{rbp}). As before we consider the individual elements
of $2\times 1$ vector $\mathbf{q}_{iT}$, denoted as $q_{iT,s}$ for $s=1,2$,
separately. For $s=1$ we have (using the individual terms in expression (\ref%
{q1e})) 
\begin{equation}
\left\vert E\left( q_{iT,1}\right) \right\vert =\left\vert E\left( -\frac{%
\mathbf{\tilde{\xi}}_{i,-1}^{\prime }\mathbf{\tilde{v}}_{i}}{\sqrt{T}}+\frac{%
\Delta \mathbf{\tilde{x}}_{i}^{\prime }\mathbf{\tilde{v}}_{i}}{\alpha _{i}%
\sqrt{T}}+\frac{\Delta \mathbf{\tilde{v}}_{i}^{\prime }\mathbf{P}_{i}\mathbf{%
\tilde{v}}_{i}}{\alpha _{i}\sqrt{T}}\right) \right\vert \leq \left\vert
E\left( \varsigma _{a,iT}\right) \right\vert +\left\vert E\left( \varsigma
_{b,iT}\right) \right\vert +\left\vert E\left( \varsigma _{c,iT}\right)
\right\vert \text{.}  \label{rbp1}
\end{equation}%
For the first term in (\ref{rbp1}), we obtain%
\begin{eqnarray*}
\left\vert E\left( \varsigma _{a,iT}\right) \right\vert &=&\left\vert E 
\left[ -\frac{\mathbf{\tilde{\xi}}_{i,-1}^{\prime }\mathbf{\tilde{v}}_{i}}{%
\sqrt{T}}=-\frac{1}{\sqrt{T}}\sum_{i=1}^{T}\left( \xi _{i,t-1}-\bar{\xi}%
_{i,-1}\right) \left( v_{it}-\bar{v}_{i}\right) \right] \right\vert \text{,}
\\
&\leq &\frac{1}{\sqrt{T}}\sum_{i=1}^{T}\left\vert E\left( \xi
_{i,t-1}v_{it}\right) \right\vert +\sqrt{T}E\left\vert \bar{\xi}_{i,-1}\bar{v%
}_{i}\right\vert \text{.}
\end{eqnarray*}%
But $E\left( \xi _{i,t-1}v_{it}\right) =0$ and $E\left\vert \bar{\xi}_{i,-1}%
\bar{v}_{i}\right\vert <K/T$ under Assumptions \ref{As1}-\ref{As2}. Hence,%
\begin{equation*}
\left\vert E\left( \varsigma _{a,iT}\right) \right\vert \leq \frac{K}{\sqrt{T%
}}\text{.}
\end{equation*}%
For the second term in (\ref{rbp1}), we obtain%
\begin{eqnarray*}
\left\vert E\left( \varsigma _{b,iT}\right) \right\vert &=&\left\vert
E\left( \frac{1}{\alpha _{i}\sqrt{T}}\sum_{t=1}^{T}u_{x,it}v_{it}-\frac{%
\sqrt{T}}{\alpha _{i}}\bar{u}_{x,i}\bar{v}_{i}\right) \right\vert \text{,} \\
&\leq &K\frac{1}{\sqrt{T}}\sum_{t=1}^{T}\left\vert E\left(
u_{x,it}v_{it}\right) \right\vert +K\sqrt{T}\left\vert E\left( \bar{u}_{x,i}%
\bar{v}_{i}\right) \right\vert \text{.}
\end{eqnarray*}%
But $E\left( u_{x,it}v_{it}\right) =0$ and $E\left( \bar{u}_{x,i}\bar{v}%
_{i}\right) =0$ under Assumption \ref{As2}. Hence%
\begin{equation*}
\left\vert E\left( \varsigma _{b,iT}\right) \right\vert =0\text{.}
\end{equation*}%
Finally, for the last term we note that 
\begin{equation*}
\left\vert E\left( \varsigma _{c,iT}\right) \right\vert \leq E\left\vert
\varsigma _{c,iT}\right\vert \leq \sqrt{E\left( \varsigma _{c,iT}^{2}\right) 
}\text{,}
\end{equation*}%
and using result (\ref{zc}), we obtain%
\begin{equation*}
\left\vert E\left( \varsigma _{c,iT}\right) \right\vert <\frac{K}{\sqrt{T}}%
\text{.}
\end{equation*}%
It now follows that $\left\vert E\left( q_{iT,1}\right) \right\vert <K/\sqrt{%
T}$, as desired.

Consider $\left\vert E\left( q_{iT,s}\right) \right\vert $ for $s=2$ next.
We have%
\begin{equation*}
\left\vert E\left( q_{iT,2}\right) \right\vert =\left\vert E\left( \varsigma
_{b,iT}\right) \right\vert =0\text{.}
\end{equation*}%
This completes the proof of result (\ref{rbp}).
\end{proof}

\bigskip

\begin{lemma}
\label{lh2} Let Assumptions \ref{As1}-\ref{As4} hold, and consider $\mathbf{B%
}_{iT}$ defined by (\ref{bit}). Then we have 
\begin{equation}
T^{\varphi }\left\Vert \mathbf{B}_{iT}-\mathbf{B}_{i}\right\Vert \rightarrow
_{p}0\text{ as }T\rightarrow \infty \text{, for any }\varphi <1/2\text{,}
\label{lh2r}
\end{equation}%
where 
\begin{equation}
\mathbf{B}_{i}=\limfunc{plim}_{T\rightarrow \infty }\mathbf{B}_{iT}=\left( 
\begin{array}{cc}
\alpha _{i}^{2}E\left( \xi _{it}^{2}\right) +\delta _{i}^{2}\sigma _{xi}^{2}
& \delta _{i}\sigma _{xi}^{2} \\ 
\delta _{i}\sigma _{xi}^{2} & \sigma _{xi}^{2}%
\end{array}%
\right) \text{,}  \label{bi}
\end{equation}%
and $\xi _{it}=\sum_{\ell =0}^{\infty }\left( 1-\alpha _{i}\right) ^{\ell
}\left( u_{y,i,t-\ell }-\beta u_{x,i,t-\ell }\right) $.
\end{lemma}

\begin{proof}
We have%
\begin{equation*}
\mathbf{B}_{iT}=\frac{\Delta \mathbf{\tilde{Z}}_{i}^{\prime }\mathbf{P}%
_{i}\Delta \mathbf{\tilde{Z}}_{i}}{T}=\frac{1}{T}\left( 
\begin{array}{cc}
\Delta \mathbf{\tilde{y}}_{i}^{\prime }\mathbf{P}_{i}\Delta \mathbf{\tilde{y}%
}_{i} & \Delta \mathbf{\tilde{y}}_{i}^{\prime }\mathbf{P}_{i}\Delta \mathbf{%
\tilde{x}}_{i} \\ 
\Delta \mathbf{\tilde{x}}_{i}^{\prime }\mathbf{P}_{i}\Delta \mathbf{\tilde{y}%
}_{i} & \Delta \mathbf{\tilde{x}}_{i}^{\prime }\mathbf{P}_{i}\Delta \mathbf{%
\tilde{x}}_{i}%
\end{array}%
\right) =\left( 
\begin{array}{cc}
b_{iT,11} & b_{iT,12} \\ 
b_{iT,21} & b_{iT,22}%
\end{array}%
\right) \text{.}
\end{equation*}%
Consider the element $b_{iT,22}$ first. Since $\Delta \mathbf{\tilde{x}}%
_{i}^{\prime }\mathbf{P}_{i}=\Delta \mathbf{\tilde{x}}_{i}^{\prime }$, and $%
\Delta \tilde{x}_{it}=u_{x,it}-\bar{u}_{x,i}$, we have%
\begin{equation*}
b_{iT,12}=\left( \frac{1}{T}\sum_{t=1}^{T}u_{x,it}^{2}\right) -\bar{u}%
_{x,i}^{2}\text{.}
\end{equation*}%
Under Assumption \ref{As2}, $u_{x,it}\sim IID\left( 0,\sigma
_{xi}^{2}\right) $ with finite fourth order moments, and therefore%
\begin{equation*}
T^{\varphi }\left( \frac{1}{T}\sum_{t=1}^{T}u_{x,it}^{2}-\sigma
_{xi}^{2}\right) \overset{p}{\rightarrow }0\text{, for any }\varphi <1/2%
\text{. }
\end{equation*}%
In addition, $E\left( \bar{u}_{x,i}^{2}\right) <K/T$, which implies $%
T^{\varphi }\bar{u}_{x,i}^{2}\overset{p}{\rightarrow }0$, for any $\varphi
<1/2$. It follows 
\begin{equation}
T^{\varphi }\left( b_{iT,22}-\sigma _{xi}^{2}\right) \overset{p}{\rightarrow 
}0\text{, for any }\varphi <1/2\text{. }  \label{b22}
\end{equation}%
Consider the element $b_{iT,11}$ next. We will use similar arguments as in
the proof of Lemma \ref{lra}. In particular, $\Delta \mathbf{\tilde{y}}_{i}$
can be written as in (\ref{yd}), and, since $\mathbf{P}_{i}\mathbf{\tilde{\xi%
}}_{i,-1}=\mathbf{\tilde{\xi}}_{i,-1}$ and $\mathbf{P}_{i}\Delta \mathbf{%
\tilde{x}}_{i}=\Delta \mathbf{\tilde{x}}_{i}$, we have%
\begin{equation}
b_{i,T,11}=\frac{\Delta \mathbf{\tilde{y}}_{i}^{\prime }\mathbf{P}_{i}\Delta 
\mathbf{\tilde{y}}_{i}}{T}=\zeta _{aa,iT}+\zeta _{bb,iT}+\zeta
_{cc,iT}+2\zeta _{ab,iT}+2\zeta _{ac,iT}+2\zeta _{bc,iT}\text{,}
\label{b11s}
\end{equation}%
where%
\begin{eqnarray*}
\zeta _{aa,iT} &=&\alpha _{i}^{2}\frac{\mathbf{\tilde{\xi}}_{i,-1}^{\prime }%
\mathbf{\tilde{\xi}}_{i,-1}}{T}\text{, }\zeta _{bb,iT}=\delta _{i}^{2}\frac{%
\Delta \mathbf{\tilde{x}}_{i}^{\prime }\Delta \mathbf{\tilde{x}}_{i}}{T}%
\text{,} \\
\zeta _{cc,iT} &=&\frac{\mathbf{\tilde{v}}_{i}^{\prime }\mathbf{P}_{i}%
\mathbf{\tilde{v}}_{i}}{T}\text{,}
\end{eqnarray*}%
and the cross-product terms are%
\begin{equation*}
\zeta _{ab,iT}=\alpha _{i}\delta _{i}\frac{\mathbf{\tilde{\xi}}%
_{i,-1}^{\prime }\Delta \mathbf{\tilde{x}}_{i}}{T},\text{ }\zeta
_{ac,iT}=\alpha _{i}\frac{\mathbf{\tilde{\xi}}_{i,-1}^{\prime }\mathbf{%
\tilde{v}}_{i}}{T}\text{, }\zeta _{bc,iT}=\delta _{i}\frac{\Delta \mathbf{%
\tilde{x}}_{i}^{\prime }\mathbf{\tilde{v}}_{i}}{T}\text{.}
\end{equation*}%
We consider these individual terms $\zeta $ next. Note that 
\begin{eqnarray*}
\xi _{it} &=&\sum_{\ell =0}^{\infty }\left( 1-\alpha _{i}\right) ^{\ell
}\left( u_{y,i,t-\ell }-\beta u_{x,i,t-\ell }\right) \text{,} \\
&=&\sum_{\ell =0}^{\infty }\left( 1-\alpha _{i}\right) ^{\ell }v_{i,t-\ell
}+\left( \delta _{i}-\beta \right) \sum_{\ell =0}^{\infty }\left( 1-\alpha
_{i}\right) ^{\ell }u_{x,i,t-\ell }\text{,}
\end{eqnarray*}%
where $\sup_{i}\left\vert 1-\alpha _{i}\right\vert <1$ under Assumption \ref%
{As1}, and innovations $v_{it}$ and $u_{xit}$ have finite fourth order
moments under Assumption \ref{As2}. Hence, $T^{\varphi }\left[
T^{-1}\sum_{t=1}^{T}\xi _{i,t-1}^{2}-E\left( \xi _{i,t-1}^{2}\right) \right]
\rightarrow _{p}0$, $E\left( \bar{\xi}_{i,-1}^{2}\right) <K/T$, and we obtain%
\begin{equation}
T^{\varphi }\left[ \zeta _{aa,iT}-\alpha _{i}^{2}E\left( \bar{\xi}%
_{i,-1}^{2}\right) \right] \rightarrow _{p}0\text{, for any }\varphi <1/2%
\text{.}  \label{pr1}
\end{equation}%
Noting that $\zeta _{bb,iT}=\delta _{i}^{2}b_{i,T,12}$, and using result (%
\ref{b22}), we have%
\begin{equation}
T^{\varphi }\left[ \zeta _{bb,iT}-\delta _{i}^{2}\sigma _{xi}^{2}\right]
\rightarrow _{p}0\text{, for any }\varphi <1/2\text{.}
\end{equation}%
Consider $\zeta _{cc,iT}$ and note that $\zeta _{cc,iT}=\frac{\alpha _{i}}{%
\sqrt{T}}\varsigma _{c,iT}$, where $\varsigma _{c,iT}=\alpha _{i}^{-1}\Delta 
\mathbf{\tilde{v}}_{i}^{\prime }\mathbf{P}_{i}\mathbf{\tilde{v}}_{i}/\sqrt{T}
$ was introduced in (\ref{q1e}) in proof of Lemma \ref{lra}. But $E\left(
\varsigma _{c,iT}^{2}\right) <\frac{K}{T}$ by (\ref{zc}), and it follows%
\begin{equation}
T^{\varphi }\zeta _{cc,iT}\rightarrow _{p}0\text{, for any }\varphi <1/2%
\text{.}
\end{equation}%
Using similar arguments, we obtain for the cross-product terms,%
\begin{equation}
T^{\varphi }\zeta _{ab,iT}\rightarrow _{p}0\text{, }T^{\varphi }\zeta
_{ac,iT}\rightarrow _{p}0\text{, and }T^{\varphi }\zeta _{bc,iT}\rightarrow
_{p}0\text{, for any }\varphi <1/2\text{, as }T\rightarrow \infty \text{.}
\label{pr4}
\end{equation}%
Using (\ref{pr1})-(\ref{pr4}) in (\ref{b11s}), we obtain%
\begin{equation}
T^{\varphi }\left( b_{i,T,11}-\alpha _{i}^{2}E\left( \xi _{it}^{2}\right)
-\delta _{i}^{2}\sigma _{xi}^{2}\right) \overset{p}{\rightarrow }0\text{,
for any }\varphi <1/2\text{. }  \label{b11}
\end{equation}%
Using the same arguments for the last term $b_{i,T,12}=b_{i,T,21}$, we obtain%
\begin{equation*}
T^{\varphi }\left( b_{i,T,12}-\delta _{i}\sigma _{xi}^{2}\right) \overset{p}{%
\rightarrow }0\text{, for any }\varphi <1/2\text{. }
\end{equation*}%
This completes the proof of (\ref{lh2r}).\bigskip
\end{proof}

\bigskip

\begin{lemma}
\label{lh} Let Assumptions \ref{As1}-\ref{As4} hold, and consider $\mathbf{B}%
_{iT}$ defined by (\ref{bit}) and $\mathbf{B}_{i}=\limfunc{plim}%
_{T\rightarrow \infty }\mathbf{B}_{iT}$ defined by (\ref{bi}). Then we have 
\begin{equation}
T^{\varphi }\left\Vert \mathbf{B}_{iT}^{-1}-\mathbf{B}_{i}^{-1}\right\Vert
\rightarrow _{p}0,\text{ as }T\rightarrow \infty \text{, for any }\varphi
<1/2\text{.}  \label{hb}
\end{equation}
\end{lemma}

\begin{proof}
This proof closely follows proof of Lemma A.8 in 
\citeN{PesaranChudik2010}%
. Let $p=\left\Vert \mathbf{B}_{i}^{-1}\right\Vert $, $q=\left\Vert \mathbf{B%
}_{iT}^{-1}-\mathbf{B}_{i}^{-1}\right\Vert $, and $r=\left\Vert \mathbf{B}%
_{iT}-\mathbf{B}_{i}\right\Vert $. We suppressed subscripts $i,T$ to
simplify the notations, but it is understood that the terms $p,q,r$ depend
on $\left( i,T\right) $. Using the triangle inequality and the
submultiplicative property of matrix norm $\left\Vert .\right\Vert $, we have%
\begin{eqnarray*}
q &=&\left\Vert \mathbf{B}_{iT}^{-1}\left( \mathbf{B}_{i}-\mathbf{B}%
_{iT}\right) \mathbf{B}_{i}^{-1}\right\Vert \text{,} \\
&\leq &\left\Vert \mathbf{B}_{iT}^{-1}\right\Vert rp\text{,} \\
&\leq &\left\Vert \left( \mathbf{B}_{iT}^{-1}-\mathbf{B}_{i}^{-1}\right) +%
\mathbf{B}_{i}^{-1}\right\Vert rp\text{,} \\
&\leq &\left( p+q\right) rp\text{.}
\end{eqnarray*}%
Subtracting $rpq$ from both sides and multiplying by $T^{\varphi }$, we
have, for any $\varphi <1/2$,%
\begin{equation}
\left( 1-rp\right) \left( T^{\varphi }q\right) \leq p^{2}\left( T^{\varphi
}r\right) \text{.}  \label{2e1}
\end{equation}%
Note that $T^{\varphi }r\overset{p}{\rightarrow }0$ by Lemma \ref{lh2}, and $%
\left\vert p\right\vert <K$ since $\mathbf{B}_{i}$ is invertible and $%
\lambda _{\min }\left( \mathbf{B}_{i}\right) $ is bounded away from zero
(this follows from observing that both $\sigma _{xi}^{2}$ and $E\left( \xi
_{it}^{2}\right) $ as well as $\alpha _{i}^{2}$ in (\ref{bi}) are bounded
away from zero). Hence,%
\begin{equation}
\left( 1-rp\right) \overset{p}{\rightarrow }1\text{,}  \label{2e2}
\end{equation}%
and%
\begin{equation}
p^{2}\left( T^{\varphi }r\right) \overset{p}{\rightarrow }0.  \label{2e3}
\end{equation}%
(\ref{2e1})-(\ref{2e3}) imply $T^{\varphi }q\overset{p}{\rightarrow }0$.
This establishes result (\ref{hb}).
\end{proof}

\bigskip

\begin{lemma}
\label{sle} Let Assumptions \ref{As1}-\ref{As4} hold, and consider $\xi
_{iT} $ defined by 
\begin{equation}
\xi _{iT}=\frac{\mathbf{\tilde{x}}_{i}^{\prime }\Delta \mathbf{\tilde{Z}}_{i}%
}{T}\left( \frac{\Delta \mathbf{\tilde{Z}}_{i}^{\prime }\mathbf{P}_{i}\Delta 
\mathbf{\tilde{Z}}_{i}}{T}\right) ^{-1}\frac{\Delta \mathbf{\tilde{Z}}%
_{i}^{\prime }\mathbf{P}_{i}\mathbf{\tilde{v}}_{i}}{\alpha _{i}\sqrt{T}}%
\text{.}  \label{exi}
\end{equation}%
where $\mathbf{P}_{i}$ is given by (\ref{pi}), and $\mathbf{\tilde{x}}_{i}$
and $\Delta \mathbf{\tilde{Z}}_{i}$ are defined below (\ref{i}). Then%
\begin{equation}
\frac{1}{\sqrt{nT}}\sum_{i=1}^{n}\xi _{iT}\rightarrow _{p}0\text{,}
\label{e1}
\end{equation}%
as $n,T\rightarrow \infty $ such that $n=\Theta \left( T^{\theta }\right) $
for some $0<\theta <2$.
\end{lemma}

\begin{proof}
Term $\xi _{iT}$ can be written as 
\begin{equation}
\xi _{iT}=\mathbf{a}_{iT}^{\prime }\mathbf{B}_{iT}^{-1}\mathbf{q}_{iT}\text{,%
}
\end{equation}%
where $\mathbf{a}_{iT}$ is given by (\ref{ait}), $\mathbf{B}_{iT}$ is given
by (\ref{bit}), and 
\begin{equation}
\mathbf{q}_{iT}=\frac{\Delta \mathbf{\tilde{Z}}_{i}^{\prime }\mathbf{P}_{i}%
\mathbf{\tilde{v}}_{i}}{\alpha _{i}\sqrt{T}}\text{.}
\end{equation}%
We have%
\begin{equation}
\frac{1}{\sqrt{nT}}\sum_{i=1}^{n}\xi _{i,T}=\frac{1}{\sqrt{nT}}\sum_{i=1}^{n}%
\mathbf{a}_{iT}^{\prime }\left( \mathbf{B}_{iT}^{-1}-\mathbf{B}%
_{i}^{-1}\right) \mathbf{q}_{iT}+\frac{1}{\sqrt{nT}}\sum_{i=1}^{n}\mathbf{a}%
_{iT}^{\prime }\mathbf{B}_{i}^{-1}\mathbf{q}_{iT}\text{.}  \label{ea}
\end{equation}%
Consider the two terms on the right side of (\ref{ea}) in turn. Lemma \ref%
{lb1} established fourth moments of $\mathbf{a}_{iT}$ are bounded, which is
sufficient for $\left\Vert \mathbf{a}_{iT}\right\Vert =O_{p}\left( 1\right) $%
. Result (\ref{rap}) of Lemma \ref{lra} established second moments of
individual elements of $\mathbf{q}_{iT}$ are bounded, which is sufficient
for $\left\Vert \mathbf{q}_{iT}\right\Vert =O_{p}\left( 1\right) $. In
addition, Lemma \ref{lh} established 
\begin{equation*}
T^{\varphi }\left\Vert \mathbf{B}_{iT}^{-1}-\mathbf{B}_{i}^{-1}\right\Vert
\rightarrow _{p}0\text{ as }T\rightarrow \infty \text{, for any }\varphi <1/2%
\text{.}
\end{equation*}%
Set $\varphi =\left( \theta -1\right) /2<1/2$ . Then we obtain%
\begin{eqnarray}
\frac{1}{\sqrt{nT}}\sum_{i=1}^{n}\mathbf{a}_{iT}^{\prime }\left( \mathbf{B}%
_{iT}^{-1}-\mathbf{B}_{i}^{-1}\right) \mathbf{q}_{iT} &=&\frac{\sqrt{n}}{%
\sqrt{T}}\frac{1}{T^{\varphi }}\frac{1}{n}\sum_{i=1}^{n}\mathbf{a}%
_{iT}^{\prime }\left[ T^{\varphi }\left( \mathbf{B}_{iT}^{-1}-\mathbf{B}%
_{i}^{-1}\right) \right] \mathbf{q}_{iT}\text{,}  \notag \\
&\leq &\frac{\sqrt{n}}{T^{\theta /2}}\left( \frac{1}{n}\sum_{i=1}^{n}\left%
\Vert \mathbf{a}_{iT}\right\Vert \left( T^{\varphi }\left\Vert \mathbf{B}%
_{iT}^{-1}-\mathbf{B}_{i}^{-1}\right\Vert \right) \left\Vert \mathbf{q}%
_{iT}\right\Vert \right) \rightarrow _{p}0\text{,}  \label{an0}
\end{eqnarray}%
as $n,T\rightarrow \infty $ such that $\theta <2$, where we used $\sqrt{T}%
T^{\varphi }=T^{1/2+\varphi }=T^{\theta /2}$, and $\frac{\sqrt{n}}{T^{\theta
/2}}<K$ since $n=\Theta \left( T^{\theta }\right) $.

Consider next the second term on the right side of (\ref{ea}). Let $\mu
_{iT}^{\ast }=E\left( \mathbf{a}_{iT}^{\prime }\mathbf{B}_{i}^{-1}\mathbf{q}%
_{iT}\right) $, and consider the variance of $\left( nT\right)
^{-1/2}\sum_{i=1}^{n}\mathbf{a}_{iT}^{\prime }\mathbf{B}_{i}^{-1}\mathbf{q}%
_{iT}$. By independence of $\mathbf{a}_{iT}^{\prime }\mathbf{B}_{i}^{-1}%
\mathbf{q}_{iT}$ across $i$, 
\begin{eqnarray}
Var\left( \frac{1}{\sqrt{nT}}\sum_{i=1}^{n}\mathbf{a}_{iT}^{\prime }\mathbf{B%
}_{i}^{-1}\mathbf{q}_{iT}\right) &=&\frac{1}{nT}\sum_{i=1}^{n}Var\left( 
\mathbf{a}_{iT}^{\prime }\mathbf{B}_{i}^{-1}\mathbf{q}_{iT}\right) \text{,} 
\notag \\
&\leq &\frac{1}{nT}\sum_{i=1}^{n}E\left( \mathbf{a}_{iT}^{\prime }\mathbf{B}%
_{i}^{-1}\mathbf{q}_{iT}\right) ^{2}\text{.}  \label{ubb}
\end{eqnarray}%
Denoting individual elements of $\mathbf{B}_{i}^{-1}$ as $b_{i,sj}^{-}$,
individual elements of $\mathbf{a}_{iT}$ as $a_{iT,j}$, and individual
elements of $\mathbf{q}_{iT}$ as $q_{iT,s}$, for $s,j=1,2$, we have%
\begin{eqnarray}
\mathbf{a}_{iT}^{\prime }\mathbf{B}_{i}^{-1}\mathbf{q}_{iT}
&=&\sum_{s=1}^{2}\sum_{j=1}^{2}b_{i,sj}^{-}a_{iT,s}q_{iT,j}\text{,}  \notag
\\
&=&b_{i,11}^{-}a_{iT,1}q_{iT,1}+b_{i,21}^{-}a_{iT,2}q_{iT,1}+b_{i,12}^{-}a_{iT,1}q_{iT,2}+b_{i,22}^{-}a_{iT,2}q_{iT,2}%
\text{,}  \label{z3}
\end{eqnarray}%
where%
\begin{equation}
a_{iT,1}=\frac{1}{T}\sum_{t=1}^{T}\tilde{x}_{it}\Delta \tilde{y}_{it}=\frac{1%
}{T}\sum_{t=1}^{T}\left( x_{it}-\bar{x}_{i}\right) \Delta y_{it}\text{,}
\label{ai1}
\end{equation}%
\begin{equation}
a_{iT,2}=\frac{1}{T}\sum_{t=1}^{T}\tilde{x}_{it}\Delta \tilde{x}_{it}=\frac{1%
}{T}\sum_{t=1}^{T}\left( x_{it}-\bar{x}_{i}\right) u_{x,it}\text{,}
\label{ai2}
\end{equation}%
\begin{equation}
q_{iT,1}=\frac{\Delta \mathbf{\tilde{y}}_{i}^{\prime }\mathbf{P}_{i}\mathbf{%
\tilde{v}}_{i}}{\alpha _{i}\sqrt{T}}\text{,}  \label{qi1}
\end{equation}%
and%
\begin{equation}
q_{iT,2}=\frac{1}{\sqrt{T}}\sum_{t=1}^{T}\frac{\tilde{u}_{x,it}\tilde{v}_{it}%
}{\alpha _{i}}\text{.}  \label{qi2}
\end{equation}%
Note that $\mathbf{B}_{i}$ is invertible and $\inf_{i}\lambda _{\min }\left( 
\mathbf{B}_{i}\right) $ is bounded away from zero (this follows from
observing that both $\sigma _{xi}^{2}$ and $E\left( \xi _{it}^{2}\right) $
as well as $\alpha _{i}^{2}$ in (\ref{bi}) are bounded away from zero). It
follows $\sup_{i}\left\Vert \mathbf{B}_{i}^{-1}\right\Vert <K$, and
therefore $\left\vert \left( b_{i,sj}^{-}\right) ^{2}\right\vert <K$. Using
this result and Cauchy-Schwarz inequality for the individual summands on the
right side of (\ref{z3}), we obtain%
\begin{equation}
E\left( \mathbf{a}_{iT}^{\prime }\mathbf{B}_{i}^{-1}\mathbf{q}_{iT}\right)
^{2}\leq K\sum_{s=1}^{2}\sum_{j=1}^{2}\sqrt{E\left( a_{iT,s}^{4}\right) }%
\sqrt{E\left( q_{iT,j}^{4}\right) }<K\text{,}  \label{uba}
\end{equation}%
where $E\left( a_{iT,s}^{4}\right) <K$ by Lemma \ref{lb1}, and $E\left(
q_{iT,j}^{4}\right) <K$ by result (\ref{rap}) of Lemma \ref{lra}. Using (\ref%
{uba}) in (\ref{ubb}), it follows that%
\begin{equation*}
Var\left( \frac{1}{\sqrt{nT}}\sum_{i=1}^{n}\mathbf{a}_{iT}^{\prime }\mathbf{B%
}_{i}^{-1}\mathbf{q}_{iT}\right) <\frac{K}{T}\text{,}
\end{equation*}%
and therefore%
\begin{equation}
\frac{1}{\sqrt{nT}}\sum_{i=1}^{n}\left( \mathbf{a}_{iT}^{\prime }\mathbf{B}%
_{i}^{-1}\mathbf{q}_{iT}-\mu _{iT}^{\ast }\right) \rightarrow _{q.m.}0\text{
as }n,T\rightarrow \infty \text{.}  \label{en1}
\end{equation}

We establish an upper bound for $\left\vert \mu _{iT}^{\ast }\right\vert $
next. We have (using (\ref{z3}) and noting that $\left\vert
b_{i,sj}^{-}\right\vert <K$) 
\begin{equation*}
\left\vert \mu _{iT}^{\ast }\right\vert <K\cdot
\sum_{s=1}^{2}\sum_{j=1}^{2}\left\vert E\left( a_{iT,s}q_{iT,j}\right)
\right\vert \text{.}
\end{equation*}%
It follows that if we can show that%
\begin{equation}
\left\vert E\left( a_{iT,s}q_{iT,j}\right) \right\vert <\frac{K}{\sqrt{T}}%
\text{, }  \label{zb1}
\end{equation}%
holds for all $s,j=1,2$, then 
\begin{equation}
\left\vert \mu _{iT}^{\ast }\right\vert <\frac{K}{\sqrt{T}}\text{,}
\label{fr}
\end{equation}%
hold. We establish (\ref{zb1}) for $s=j=2$, first, which is the most
convenient case to consider. We have%
\begin{equation}
E\left( a_{iT,2}q_{iT,2}\right) =E\left( \frac{1}{T}\sum_{t=1}^{T}\left(
x_{it}-\bar{x}_{i}\right) u_{x,it}\cdot \frac{1}{\sqrt{T}}\sum_{t=1}^{T}%
\frac{u_{x,it}v_{it}}{\alpha _{i}}\right) =0\text{,}  \label{rp1}
\end{equation}%
since $v_{it}\,$is independently distributed of $u_{x,it^{\prime }}$ for any 
$t,t^{\prime }$. Consider next $s=1$, $j=2$. We have%
\begin{equation}
E\left( a_{iT,1}q_{iT,2}\right) =E\left( \frac{1}{T}\sum_{t=1}^{T}\left(
x_{it}-\bar{x}_{i}\right) \Delta y_{it}\cdot \frac{1}{\sqrt{T}}\sum_{t=1}^{T}%
\frac{u_{x,it}v_{it}}{\alpha _{i}}\right) \text{,}
\end{equation}%
where (first-differencing (\ref{ma_y}) and substituting (\ref{uy}))%
\begin{eqnarray}
\Delta y_{it} &=&\delta _{i}u_{x,it}+v_{it}-\alpha _{i}\sum_{\ell
=1}^{\infty }\left( 1-\alpha _{i}\right) ^{\ell -1}\left[ v_{i,t-\ell
}+\left( \delta _{i}-\beta \right) u_{x,i,t-\ell }\right] \text{,}  \notag \\
&=&\eta _{u,it}+\eta _{v,it}\text{,}  \label{dyd}
\end{eqnarray}%
in which%
\begin{equation}
\eta _{u,it}=\delta _{i}u_{x,it}-\alpha _{i}\sum_{\ell =1}^{\infty }\left(
1-\alpha _{i}\right) ^{\ell -1}\left( \delta _{i}-\beta \right)
u_{x,i,t-\ell }\text{,}
\end{equation}%
and%
\begin{equation}
\eta _{v,it}=v_{it}-\alpha _{i}\sum_{\ell =1}^{\infty }\left( 1-\alpha
_{i}\right) ^{\ell -1}v_{i,t-\ell }\text{.}
\end{equation}%
Hence, $E\left( a_{iT,1}q_{iT,2}\right) $ can be written as%
\begin{eqnarray*}
E\left( a_{iT,1}q_{iT,2}\right) &=&E\left( \frac{1}{T}\sum_{t=1}^{T}\left(
x_{it}-\bar{x}_{i}\right) \eta _{u,it}\cdot \frac{1}{\sqrt{T}}\sum_{t=1}^{T}%
\frac{u_{x,it}v_{it}}{\alpha _{i}}\right) \\
&&+E\left( \frac{1}{T}\sum_{t=1}^{T}\left( x_{it}-\bar{x}_{i}\right) \eta
_{v,it}\cdot \frac{1}{\sqrt{T}}\sum_{t=1}^{T}\frac{u_{x,it}v_{it}}{\alpha
_{i}}\right) \text{.}
\end{eqnarray*}%
The first term is equal to $0$, since $v_{it}\,$is independently distributed
of $u_{x,it^{\prime }}$ for any $t,t^{\prime }$. Consider the second term.
Noting that $E\left[ \left( x_{it}-\bar{x}_{i}\right) u_{x,is}\right] <K$
and $\left\vert \alpha _{i}^{-1}\right\vert <K$ for any $i,t,s$, we obtain 
\begin{eqnarray*}
E\left( \frac{1}{T}\sum_{t=1}^{T}\left( x_{it}-\bar{x}_{i}\right) \eta
_{v,it}\cdot \frac{1}{\sqrt{T}}\sum_{t=1}^{T}\frac{u_{x,it}v_{it}}{\alpha
_{i}}\right) &=&\frac{1}{T^{3/2}}\sum_{t=1}^{T}\sum_{s=1}^{T}\alpha
_{i}^{-1}E\left[ \left( x_{it}-\bar{x}_{i}\right) u_{x,is}\right] E\left(
\eta _{v,is}v_{it}\right) \text{,} \\
&\leq &\frac{K}{T^{3/2}}\sum_{t=1}^{T}\sum_{s=1}^{T}E\left( \eta
_{v,is}v_{it}\right) \text{.}
\end{eqnarray*}%
But 
\begin{equation*}
E\left( \eta _{v,is}v_{it}\right) =\left\{ 
\begin{array}{c}
0\text{, for }s<t\text{,} \\ 
\sigma _{vi}^{2}<K\text{, for }s=t\text{,} \\ 
\leq K\rho ^{s-t}\text{, for }s>t\text{,}%
\end{array}%
\right.
\end{equation*}%
where $\rho \equiv \sup_{i}\left\vert 1-\alpha _{i}\right\vert <1$ by\
Assumption \ref{As1}. Hence $\left\vert \sum_{s=1}^{T}E\left( \eta
_{v,is}v_{it}\right) \right\vert <K$ for any $t=1,2,...T$, and%
\begin{equation}
\left\vert E\left( \frac{1}{T}\sum_{t=1}^{T}\left( x_{it}-\bar{x}_{i}\right)
\eta _{v,it}\cdot \frac{1}{\sqrt{T}}\sum_{t=1}^{T}\frac{u_{x,it}v_{it}}{%
\alpha _{i}}\right) \right\vert \leq \frac{K}{\sqrt{T}}\text{,}  \label{rp2}
\end{equation}%
as desired. This establish (\ref{zb1}) hold for $s=1$, $j=2$.

Consider next (\ref{zb1}) for $s\in \left\{ 1,2\right\} $ and $j=1$. Using
expression (\ref{q1e}), we can write $a_{iT,s}q_{iT,1}$, for $s=1,2$, as%
\begin{equation}
a_{iT,s}q_{iT,1}=a_{iT,s}\varsigma _{a,iT}+a_{iT,s}\varsigma
_{b,iT}+a_{iT,s}\varsigma _{c,iT}\text{,}
\end{equation}%
where as in the proof of Lemma \ref{lra} $\varsigma _{a,iT}=-\mathbf{\tilde{%
\xi}}_{i,-1}^{\prime }\mathbf{\tilde{v}}_{i}/\sqrt{T}$, $\varsigma
_{b,iT}=\alpha _{i}^{-1}\Delta \mathbf{\tilde{x}}_{i}^{\prime }\mathbf{%
\tilde{v}}_{i}/\sqrt{T}$ and $\varsigma _{c,iT}=\alpha _{i}^{-1}\Delta 
\mathbf{\tilde{v}}_{i}^{\prime }\mathbf{P}_{i}\mathbf{\tilde{v}}_{i}/\sqrt{T}
$. Using similar arguments as in establishing (\ref{rp2}), we obtain%
\begin{equation*}
\left\vert E\left( a_{iT,s}\varsigma _{a,iT}\right) \right\vert <\frac{K}{%
\sqrt{T}}\text{, for }s=1,2\text{.}
\end{equation*}%
Noting next that $\varsigma _{b,iT}=\alpha _{i}^{-1}q_{i,T,2}$, it directly
follows from results (\ref{rp1}) and (\ref{rp2}) that%
\begin{equation*}
\left\vert E\left( a_{iT,s}\varsigma _{b,iT}\right) \right\vert <\frac{K}{%
\sqrt{T}}\text{, for }s=1,2\text{.}
\end{equation*}%
Consider the last term, $a_{i,T,s}\varsigma _{c,iT}$, for $s=1,2$. Using
Cauchy-Schwarz inequality we have%
\begin{equation*}
\left\vert E\left( a_{iT,s}\varsigma _{c,iT}\right) \right\vert \leq \sqrt{%
E\left( a_{iT,s}^{2}\right) }\sqrt{E\left( \varsigma _{c,iT}^{2}\right) }%
\text{, for }s=1,2\text{. }
\end{equation*}%
But $E\left( a_{iT,s}^{2}\right) <K$, for $s=1,2$ by Lemma \ref{lb1}, and $%
E\left( \varsigma _{c,iT}^{2}\right) <K/T$ is implied by (\ref{zc}). Hence%
\begin{equation*}
\left\vert E\left( a_{iT,s}\varsigma _{c,iT}\right) \right\vert \leq \frac{K%
}{\sqrt{T}}\text{, for }s=1,2\text{. }
\end{equation*}%
This completes the proof of (\ref{zb1}) for all $s,j=1,2$, and therefore (%
\ref{fr}) holds. Using (\ref{fr}), we 
\begin{equation}
\left\vert \frac{1}{\sqrt{nT}}\sum_{i=1}^{n}\mu _{iT}^{\ast }\right\vert
\leq \frac{1}{\sqrt{nT}}\sum_{i=1}^{n}\left\vert \mu _{iT}^{\ast
}\right\vert <\frac{1}{\sqrt{nT}}\sum_{i=1}^{n}\frac{K}{\sqrt{T}}=K\frac{%
\sqrt{n}}{T}\rightarrow 0\text{,}  \label{en2}
\end{equation}%
as $n,T\rightarrow \infty $ such that $\sqrt{n}/T\rightarrow 0$. Results (%
\ref{en1}) and (\ref{en2}) imply 
\begin{equation}
\frac{1}{\sqrt{nT}}\sum_{i=1}^{n}\mathbf{a}_{iT}^{\prime }\mathbf{B}_{i}^{-1}%
\mathbf{q}_{iT}\rightarrow _{p}0\text{,}  \label{an1}
\end{equation}%
as $n,T\rightarrow \infty $ such that $\sqrt{n}/T\rightarrow 0$. Finally,
using (\ref{an0}) and (\ref{an1}) in (\ref{ea}), we obtain (\ref{en1}), as
desired.
\end{proof}

\bigskip

\begin{lemma}
\label{leb}\textbf{\ }Let Assumptions \ref{As1}-\ref{As4} hold. Then 
\begin{equation*}
\frac{1}{\sqrt{n}}\sum_{i=1}^{n}E\left( \frac{\mathbf{\tilde{x}}_{i}^{\prime
}\mathbf{M}_{i}\mathbf{\tilde{v}}_{i}}{\alpha _{i}T}\right) \rightarrow 0,
\end{equation*}%
as $n,T\rightarrow \infty $ such that $n=\Theta \left( T^{\theta }\right) $
for some $0<\theta <2$, where $\mathbf{M}_{i}$ is defined in (\ref{mi}), and 
$\mathbf{\tilde{x}}_{i}$ and $\mathbf{\tilde{v}}_{i}$ are defined below (\ref%
{i}).
\end{lemma}

\begin{proof}
We have 
\begin{equation}
\frac{1}{\sqrt{n}}\sum_{i=1}^{n}\frac{\mathbf{\tilde{x}}_{i}^{\prime }%
\mathbf{M}_{i}\mathbf{\tilde{v}}_{i}}{\alpha _{i}T}=\frac{1}{\sqrt{n}}%
\sum_{i=1}^{n}\frac{\mathbf{\tilde{x}}_{i}^{\prime }\mathbf{\tilde{v}}_{i}}{%
\alpha _{i}T}+\frac{1}{\sqrt{nT}}\sum_{i=1}^{n}\xi _{iT}\text{,}  \label{tsb}
\end{equation}%
where $\xi _{iT}$ is defined by (\ref{exi}). For the first term, $E\left( 
\mathbf{\tilde{x}}_{i}^{\prime }\mathbf{\tilde{v}}_{i}\right) =0$. For the
second term, Lemma \ref{sle} established $\frac{1}{\sqrt{nT}}%
\sum_{i=1}^{n}\xi _{iT}\rightarrow _{p}0$, as $n,T\rightarrow \infty $\ such
that $n=\Theta \left( T^{\theta }\right) $, for some $0<\theta <2$. Hence it
follows that $\frac{1}{\sqrt{n}}\sum_{i=1}^{n}E\left( \frac{\mathbf{\tilde{x}%
}_{i}^{\prime }\mathbf{M}_{i}\mathbf{\tilde{v}}_{i}}{\alpha _{i}T}\right)
\rightarrow 0$, as $n$,$T\rightarrow \infty $\ such that $n=\Theta \left(
T^{\theta }\right) $, for some $0<\theta <2$.
\end{proof}

\bigskip

\begin{lemma}
\label{lad} \textbf{\ }Suppose conditions of Theorem \ref{T1} hold. Then 
\begin{equation}
\frac{1}{\sqrt{n}}\sum_{i=1}^{n}\left[ \frac{\mathbf{\tilde{x}}_{i}^{\prime }%
\mathbf{M}_{i}\mathbf{\tilde{v}}_{i}}{\alpha _{i}T}-E\left( \frac{\mathbf{%
\tilde{x}}_{i}^{\prime }\mathbf{M}_{i}\mathbf{\tilde{v}}_{i}}{\alpha _{i}T}%
\right) \right] \rightarrow _{d}N\left( 0,\omega _{v}^{2}\right) \text{,}
\label{rcd}
\end{equation}%
where $\omega _{v}^{2}=\lim {}_{n\rightarrow \infty
}n^{-1}\sum_{i=1}^{n}\sigma _{xi}^{2}\sigma _{vi}^{2}/\left( 6\alpha
_{i}^{2}\right) $, $\mathbf{M}_{i}$ is defined by (\ref{mi}), and $\mathbf{%
\tilde{x}}_{i}$ and $\mathbf{\tilde{v}}_{i}$ are defined below (\ref{i}).
\end{lemma}

\begin{proof}
It is convenient to use (\ref{tsb}) in (\ref{rcd}) to obtain%
\begin{equation*}
\frac{1}{\sqrt{n}}\sum_{i=1}^{n}\left[ \frac{\mathbf{\tilde{x}}_{i}^{\prime }%
\mathbf{M}_{i}\mathbf{\tilde{v}}_{i}}{\alpha _{i}T}-E\left( \frac{\mathbf{%
\tilde{x}}_{i}^{\prime }\mathbf{M}_{i}\mathbf{\tilde{v}}_{i}}{\alpha _{i}T}%
\right) \right] =\frac{1}{\sqrt{n}}\sum_{i=1}^{n}\frac{\mathbf{\tilde{x}}%
_{i}^{\prime }\mathbf{\tilde{v}}_{i}}{\alpha _{i}T}+\frac{1}{\sqrt{nT}}%
\sum_{i=1}^{n}\left[ \xi _{iT}-E\left( \xi _{iT}\right) \right] \text{,}
\end{equation*}%
where $E\left( \mathbf{\tilde{x}}_{i}^{\prime }\mathbf{\tilde{v}}_{i}\right)
=0$. It follows from Lemmas \ref{sle} and \ref{leb} that $\frac{1}{\sqrt{nT}}%
\sum_{i=1}^{n}\left[ \xi _{iT}-E\left( \xi _{iT}\right) \right] \rightarrow
_{p}0$, as $n,T\rightarrow \infty $\ such that $n=\Theta \left( T^{\theta
}\right) $, for some $0<\theta <2$. Hence, under the conditions of Theorem %
\ref{T1}, the asymptotic distribution of $\frac{1}{\sqrt{n}}\sum_{i=1}^{n}%
\left[ \frac{\mathbf{\tilde{x}}_{i}^{\prime }\mathbf{M}_{i}\mathbf{\tilde{v}}%
_{i}}{\alpha _{i}T}-E\left( \frac{\mathbf{\tilde{x}}_{i}^{\prime }\mathbf{M}%
_{i}\mathbf{\tilde{v}}_{i}}{\alpha _{i}T}\right) \right] $ is given by the
first term, $\frac{1}{\sqrt{n}}\sum_{i=1}^{n}\frac{\mathbf{\tilde{x}}%
_{i}^{\prime }\mathbf{\tilde{v}}_{i}}{\alpha _{i}T}$, alone. Lemma \ref{sld}
establishes the asymptotic normality of this term.%
\begin{equation}
\frac{1}{\sqrt{n}}\sum_{i=1}^{n}\frac{\mathbf{\tilde{x}}_{i}^{\prime }%
\mathbf{M}_{i}\mathbf{\tilde{v}}_{i}}{\alpha _{i}T}\rightarrow _{d}N\left(
0,\omega _{v}^{2}\right) \text{,}
\end{equation}%
as $n,T\rightarrow \infty $. This completes the proof.
\end{proof}

\bigskip

\section{Estimation algorithms}

\noindent This appendix describes implementation of the Pooled Mean Group
(PMG) estimator, which we compute iteratively, in Section \ref{EAP1}.
Section \ref{Sbcp} discusses implementation of bias-correction methods and
bootstrapping of critical values for the PMG, PDOLS and FMOLS estimators.
Section \ref{A4} provides tables with Monte Carlo findings for experiments
with cross-sectionally dependent errors.

\subsection{\label{EAP1}Computation of PMG estimator}

Consider the same illustrative panel ARDL model as in the paper, namely the
model given by equations (\ref{y})-(\ref{x}). PMG estimator of the long-run
coefficient $\beta $, as originally proposed by 
\citeN{PesaranShinSmith1999}%
, is computed by solving the following equations iteratively:%
\begin{equation}
\hat{\beta}_{PMG}=-\left( \sum_{i=1}^{n}\frac{\hat{\phi}_{i}^{2}}{\hat{\sigma%
}_{i}^{2}}\mathbf{x}_{i}^{\prime }\mathbf{H}_{x,i}\mathbf{x}_{i}\right)
^{-1}\sum_{i=1}^{n}\frac{\hat{\phi}_{i}^{2}}{\hat{\sigma}_{i}^{2}}\mathbf{x}%
_{i}^{\prime }\mathbf{H}_{x,i}\left( \Delta \mathbf{y}_{i}-\hat{\phi}_{i}%
\mathbf{y}_{i,-1}\right) \text{,}  \label{pmg1}
\end{equation}%
\begin{equation}
\hat{\phi}_{i}=\left( \mathbf{\hat{\xi}}_{i}^{\prime }\mathbf{H}_{x,i}%
\mathbf{\hat{\xi}}_{i}\right) ^{-1}\mathbf{\hat{\xi}}_{i}^{\prime }\mathbf{H}%
_{x,i}\Delta \mathbf{y}_{i}\text{, }i=1,2,...,n\text{,}  \label{pmg2}
\end{equation}%
and%
\begin{equation}
\hat{\sigma}_{i}^{2}=T^{-1}\left( \Delta \mathbf{y}_{i}-\hat{\phi}_{i}%
\mathbf{\hat{\xi}}_{i}\right) ^{\prime }\mathbf{H}_{x,i}\left( \Delta 
\mathbf{y}_{i}-\hat{\phi}_{i}\mathbf{\hat{\xi}}_{i}\right) \text{, }%
i=1,2,...,n\text{,}  \label{pmg3}
\end{equation}%
where $\mathbf{\hat{\xi}}_{i}=\mathbf{y}_{i,-1}-\mathbf{x}_{i}\hat{\beta}%
_{PMG}$, $\mathbf{x}_{i}=\left( x_{i,1},x_{i,2},...,x_{i,T}\right) ^{\prime
} $, $\Delta \mathbf{y}_{i}=\mathbf{y}_{i}-\mathbf{y}_{i,-1}$, $\mathbf{y}%
_{i}=\left( y_{i,1},y_{i,2},...,y_{i,T}\right) ^{\prime }$, $\mathbf{y}%
_{i,-1}=\left( y_{i,0},y_{i,1},...,y_{i,T-1}\right) ^{\prime }$, $\mathbf{H}%
_{x,i}=\mathbf{I}_{T}-\Delta \mathbf{x}_{i}\left( \Delta \mathbf{x}%
_{i}^{\prime }\Delta \mathbf{x}_{i}\right) ^{-1}\Delta \mathbf{x}%
_{i}^{\prime }$, $\Delta \mathbf{x}_{i}=\mathbf{x}_{i}-\mathbf{x}_{i,-1}$,
and $\mathbf{x}_{i,-1}\mathbf{=}\left( x_{i,0},x_{i,1},...,x_{i,T-1}\right)
^{\prime }$.

To solve (\ref{pmg1})-(\ref{pmg3}) iteratively, we set $\hat{\beta}%
_{PMG,\left( 0\right) }$ to the pooled Engle-Granger estimator, and given
the initial estimate $\hat{\beta}_{PMG,\left( 0\right) }$, we compute $%
\mathbf{\hat{\xi}}_{i,\left( 0\right) }=\mathbf{y}_{i,-1}-\mathbf{x}_{i}\hat{%
\beta}_{PMG,\left( 0\right) }$, $\hat{\phi}_{i,\left( 0\right) }$ and $\hat{%
\sigma}_{i,\left( 0\right) }^{2}$, for $i=1,2,...,n$ using (\ref{pmg2})-(\ref%
{pmg3}). Next we compute $\hat{\beta}_{PMG,\left( 1\right) }$ using (\ref%
{pmg1}) and given values $\hat{\phi}_{i,\left( 0\right) }$ and $\hat{\sigma}%
_{i,\left( 0\right) }^{2}$. Then we iterate - for a given value of $\hat{%
\beta}_{PMG,\left( \ell \right) }$ we compute $\mathbf{\hat{\xi}}_{i,\left(
\ell \right) }$, $\hat{\phi}_{i,\left( \ell \right) }$ and $\hat{\sigma}%
_{i,\left( \ell \right) }^{2}$; and for given values of $\hat{\phi}%
_{i,\left( \ell \right) }$ and $\hat{\sigma}_{i,\left( \ell \right) }^{2}$
we compute $\hat{\beta}_{PMG,\left( \ell +1\right) }$. If convergence is not
achieved, we increase $\ell $ by one and repeat. We define convergence by $%
\left\vert \hat{\beta}_{PMG,\left( \ell +1\right) }-\hat{\beta}_{PMG,\left(
\ell \right) }\right\vert <10^{-4}$.

Inference is conducted using equation (17) of 
\citeN{PesaranShinSmith1999}%
. In particular, 
\begin{equation*}
T\sqrt{n}\left( \hat{\beta}_{PMG}-\beta _{0}\right) \sim N\left( 0,\Omega
_{PMG}\right) \text{, }
\end{equation*}%
where%
\begin{equation*}
\Omega _{PMG}=\left( \frac{1}{n}\sum_{i=1}^{n}\frac{\phi _{i,0}}{\sigma
_{i,0}^{2}}r_{x_{i},x_{i}}\right) ^{-1}\text{, and }r_{x_{i},x_{i}}=plim_{T%
\rightarrow \infty }T^{-2}\mathbf{x}_{i}^{\prime }\mathbf{H}_{x,i}\mathbf{x}%
_{i}\text{ .}
\end{equation*}%
Standard error of $\hat{\beta}_{PMG}$, denoted as $se\left( \hat{\beta}%
_{PMG}\right) $, is estimated as 
\begin{equation*}
\widehat{se}\left( \hat{\beta}_{PMG}\right) =T^{-1}n^{-1/2}\hat{\Omega}_{PMG}%
\text{, }
\end{equation*}%
where%
\begin{equation}
\hat{\Omega}_{PMG}=\left( \frac{1}{n}\sum_{i=1}^{n}\frac{\hat{\phi}_{i,0}}{%
\hat{\sigma}_{i,0}^{2}}\hat{r}_{x_{i},x_{i}}\right) ^{-1}\text{ and }\hat{r}%
_{x_{i},x_{i}}=T^{-2}\mathbf{x}_{i}^{\prime }\mathbf{H}_{x,i}\mathbf{x}_{i}%
\text{.}  \label{omp}
\end{equation}

\subsection{Bias-corrected PMG, PDOLS and FMOLS estimators, and bootstrapped
critical values\label{Sbcp}}

Similarly to the bootstrap bias-corrected PB estimator, we consider the
following bootstrap bias-corrected PMG, PDOLS and FMOLS estimators. Let the
original (uncorrected) estimators be denoted as $\hat{\beta}_{e}$ for $%
e=PMG,PDOLS$, and $FMOLS$, respectively. Bootstrap bias corrected version of
these estimators is given by%
\begin{equation}
\tilde{\beta}_{e}=\hat{\beta}_{e}-\hat{b}_{e}\text{,}  \label{bspmg}
\end{equation}%
for $e=PMG,PDOLS$, and $FMOLS$, where $\hat{b}_{e}$ an estimate of the bias
obtained by the following sieve wild bootstrap algorithm, which resembles
the algorithm in Subsection \ref{Sbbc}. For e$=PMG,PDOLS$, and $FMOLS:$

\begin{enumerate}
\item Compute $\hat{\beta}_{e}$. Given $\hat{\beta}_{e}$, estimate the
remaining unknown coefficients of (\ref{y})-(\ref{x}) by least squares, and
compute residuals $\hat{u}_{y,it}^{e},\hat{u}_{x,it}^{e}$.

\item For each $r=1,2,...,R$, generate new draws for $\hat{u}%
_{y,it}^{e,\left( r\right) }=a_{t}^{\left( r\right) }\hat{u}_{y,it}^{e}$,
and $\hat{u}_{x,it}^{e,\left( r\right) }=a_{t}^{\left( r\right) }\hat{u}%
_{x,it}^{e}$, where $a_{t}^{\left( r\right) }$ are randomly drawn from
Rademacher distribution (%
\citeANP{Liu1988}, \citeyearNP{Liu1988}%
) namely%
\begin{equation*}
a_{t}^{\left( r\right) }=\left\{ 
\begin{array}{cc}
-1\text{,} & \text{with probability 1/2} \\ 
1\text{,} & \text{with probability 1/2}%
\end{array}%
\right. \text{.}
\end{equation*}%
Given the estimated parameters of (\ref{y})-(\ref{x}) from Step 1, and
initial values $y_{i1},x_{i1}$ generate simulated data $y_{it}^{e,\left(
r\right) },x_{it}^{e,\left( r\right) }$ for $t=2,3,...,T$ and $i=1,2,...,n$.
Using the generated data compute $\hat{\beta}_{e}^{\left( r\right) }$.

\item Compute $\hat{b}_{e}=\left[ R^{-1}\sum_{r=1}^{R}\hat{\beta}%
_{e}^{\left( r\right) }-\hat{\beta}_{e}\right] $.
\end{enumerate}

The possibility of iterating the algorithm above by using the bias-corrected
estimate $\tilde{\beta}_{e}$ in Step 1 is not considered in this paper.

We conduct inference by using the $1-\alpha $ confidence interval $%
C_{1-\alpha }\left( \tilde{\beta}_{e}\right) =\tilde{\beta}_{e}\pm \hat{k}%
_{e}\widehat{se}\left( \hat{\beta}_{e}\right) =\tilde{\beta}_{e}\pm
T^{-1}n^{-1/2}\hat{k}_{e}\hat{\Omega}_{e}$, where $\hat{k}_{e}$ is the $%
1-\alpha $ percent quantile of $\left\{ \left\vert t_{e}^{\left( r\right)
}\right\vert \right\} _{r=1}^{R}$, in which $t_{e}^{\left( r\right) }=\tilde{%
\beta}_{e}^{\left( r\right) }/\widehat{se}\left( \hat{\beta}_{e}^{\left(
r\right) }\right) =T^{-1}n^{-1/2}\tilde{\beta}_{e}^{\left( r\right) }/\hat{%
\Omega}_{e}^{\left( r\right) }$, $\tilde{\beta}_{e}^{\left( r\right) }=\hat{%
\beta}_{e}^{\left( r\right) }-\hat{b}_{e}$ is the bias-corrected PMG
estimate of $\beta $ in the $r$-$th$ draw of the bootstrap data in the
algorithm above, and $\hat{\Omega}_{e}^{\left( r\right) }$ is estimated
standard error using the bootstrap data.

\subsubsection{Jackknife bias-corrections}

We consider similar jackknife bias correction for PMG, PDOLS and FMOLS
estimator as for the PB estimator in Section \ref{BM}. In particular,

\begin{equation*}
\tilde{\beta}_{jk,e}=\tilde{\beta}_{jk,e}\left( \kappa \right) =\hat{\beta}%
_{e}-\kappa \left( \frac{\hat{\beta}_{e,a}+\hat{\beta}_{e,b}}{2}-\hat{\beta}%
_{e}\right) \text{,}
\end{equation*}%
for $e=PMG,PDOLS$, and $FMOLS$, where $\hat{\beta}_{e}$ is the full sample
estimator, $\hat{\beta}_{e,a}$ and $\hat{\beta}_{e,b}$ are the first and the
second half sub-sample estimators, and $\kappa =1/3$ is the same weighting
parameter as in Section \ref{BM}.

We conduct inference by using the $1-\alpha $ confidence interval $%
C_{1-\alpha }\left( \tilde{\beta}_{jk,e}\right) =\tilde{\beta}_{jk,e}\pm 
\hat{k}_{jk,e}\widehat{se}\left( \hat{\beta}_{e}\right) =\tilde{\beta}%
_{jk,e}\pm \hat{k}_{jk,e}T^{-1}n^{-1/2}\hat{\Omega}_{e}$, where $\hat{k}%
_{jk} $ is the $1-\alpha $ percent quantile of $\left\{ \left\vert
t_{jk,e}^{\left( r\right) }\right\vert \right\} _{r=1}^{R}$, in which $%
t_{jk,e}^{\left( r\right) }=\tilde{\beta}_{jk,e}^{\left( r\right) }/\widehat{%
se}\left( \hat{\beta}_{e}^{\left( r\right) }\right) =T^{-1}n^{-1/2}\tilde{%
\beta}_{jk,e}^{\left( r\right) }/\hat{\Omega}_{e}^{\left( r\right) }$, $%
\tilde{\beta}_{jk,e}^{\left( r\right) }$ is the jackknife bias-corrected
estimate of $\beta $ using the $r$-$th$ draw of the bootstrap data generated
using the same algorithm as in Subsection \ref{Sbcp}, and $\hat{\Omega}%
_{e}^{\left( r\right) }$ is estimated standard error using the bootstrap
data.

\bigskip

\subsection{\label{A4}Monte Carlo results for experiments with
cross-sectionally dependent errors}

\bigskip \pagebreak

\begin{center}
\textbf{Table B1: }MC findings for the estimation of long-run coefficient $%
\beta $ in experiments with cross-sectionally dependent errors.

Estimators without bias correction and inference conducted using standard
critical values.

\bigskip

\footnotesize%

\renewcommand{\arraystretch}{1.0}\setlength{\tabcolsep}{3pt}%
\scriptsize%
\begin{tabular}{rrrrrccrrrrccrrrrccrrrr}
\hline\hline
& \multicolumn{4}{c}{\textbf{Bias (}$\mathbf{\times }$\textbf{\ 100)}} &  & 
& \multicolumn{4}{c}{\textbf{RMSE (}$\mathbf{\times }$\textbf{\ 100)}} &  & 
& \multicolumn{4}{c}{\textbf{Size (5\% level)}} &  &  & \multicolumn{4}{c}{%
\textbf{Power (5\% level)}} \\ 
\cline{2-5}\cline{8-11}\cline{14-17}\cline{20-23}
$n\backslash T$ & \textbf{20} & \textbf{30} & \textbf{40} & \textbf{50} & 
\multicolumn{1}{r}{} & \multicolumn{1}{r}{} & \textbf{20} & \textbf{30} & 
\textbf{40} & \textbf{50} & \multicolumn{1}{r}{} & \multicolumn{1}{r}{} & 
\textbf{20} & \textbf{30} & \textbf{40} & \textbf{50} & \multicolumn{1}{r}{}
& \multicolumn{1}{r}{} & \textbf{20} & \textbf{30} & \textbf{40} & \textbf{50%
} \\ \cline{2-23}
& \multicolumn{22}{l}{PB} \\ \hline
\textbf{20} & -3.83 & -1.93 & -1.01 & -0.66 & \multicolumn{1}{r}{} & 
\multicolumn{1}{r}{} & 7.58 & 4.92 & 3.57 & 2.78 & \multicolumn{1}{r}{} & 
\multicolumn{1}{r}{} & 28.00 & 21.95 & 17.55 & 16.80 & \multicolumn{1}{r}{}
& \multicolumn{1}{r}{} & 37.30 & 64.50 & 86.90 & 96.85 \\ 
\textbf{30} & -3.57 & -1.98 & -1.17 & -0.69 & \multicolumn{1}{r}{} & 
\multicolumn{1}{r}{} & 6.67 & 4.43 & 3.32 & 2.46 & \multicolumn{1}{r}{} & 
\multicolumn{1}{r}{} & 28.95 & 23.55 & 22.10 & 18.80 & \multicolumn{1}{r}{}
& \multicolumn{1}{r}{} & 45.10 & 75.25 & 92.70 & 99.00 \\ 
\textbf{40} & -3.94 & -2.04 & -1.10 & -0.74 & \multicolumn{1}{r}{} & 
\multicolumn{1}{r}{} & 6.58 & 4.20 & 3.08 & 2.37 & \multicolumn{1}{r}{} & 
\multicolumn{1}{r}{} & 33.60 & 27.25 & 23.75 & 22.90 & \multicolumn{1}{r}{}
& \multicolumn{1}{r}{} & 46.60 & 82.00 & 96.70 & 99.65 \\ 
\textbf{50} & -3.88 & -1.99 & -1.10 & -0.76 & \multicolumn{1}{r}{} & 
\multicolumn{1}{r}{} & 6.35 & 4.01 & 2.89 & 2.23 & \multicolumn{1}{r}{} & 
\multicolumn{1}{r}{} & 38.30 & 29.55 & 25.30 & 23.30 & \multicolumn{1}{r}{}
& \multicolumn{1}{r}{} & 50.70 & 85.60 & 97.75 & 99.95 \\ \hline
& \multicolumn{22}{l}{PMG} \\ \hline
\textbf{20} & -1.96 & -0.90 & -0.40 & -0.26 & \multicolumn{1}{r}{} & 
\multicolumn{1}{r}{} & 9.05 & 5.49 & 3.86 & 3.00 & \multicolumn{1}{r}{} & 
\multicolumn{1}{r}{} & 45.75 & 34.00 & 26.30 & 23.90 & \multicolumn{1}{r}{}
& \multicolumn{1}{r}{} & 62.45 & 78.50 & 92.45 & 98.15 \\ 
\textbf{30} & -1.60 & -1.05 & -0.56 & -0.31 & \multicolumn{1}{r}{} & 
\multicolumn{1}{r}{} & 7.76 & 4.83 & 3.54 & 2.58 & \multicolumn{1}{r}{} & 
\multicolumn{1}{r}{} & 46.35 & 37.25 & 32.75 & 27.35 & \multicolumn{1}{r}{}
& \multicolumn{1}{r}{} & 70.50 & 85.25 & 95.60 & 99.65 \\ 
\textbf{40} & -1.79 & -1.02 & -0.50 & -0.32 & \multicolumn{1}{r}{} & 
\multicolumn{1}{r}{} & 7.07 & 4.51 & 3.27 & 2.44 & \multicolumn{1}{r}{} & 
\multicolumn{1}{r}{} & 49.40 & 40.80 & 35.00 & 29.90 & \multicolumn{1}{r}{}
& \multicolumn{1}{r}{} & 73.05 & 90.25 & 97.65 & 99.85 \\ 
\textbf{50} & -1.76 & -0.89 & -0.46 & -0.32 & \multicolumn{1}{r}{} & 
\multicolumn{1}{r}{} & 6.70 & 4.21 & 3.00 & 2.28 & \multicolumn{1}{r}{} & 
\multicolumn{1}{r}{} & 54.30 & 42.95 & 34.20 & 31.55 & \multicolumn{1}{r}{}
& \multicolumn{1}{r}{} & 76.40 & 93.90 & 98.90 & 100.00 \\ \cline{2-23}
& \multicolumn{22}{l}{PDOLS} \\ \hline
\textbf{20} & -6.20 & -4.10 & -3.08 & -2.40 & \multicolumn{1}{r}{} & 
\multicolumn{1}{r}{} & 9.50 & 6.31 & 4.67 & 3.67 & \multicolumn{1}{r}{} & 
\multicolumn{1}{r}{} & 31.40 & 27.30 & 25.70 & 26.30 & \multicolumn{1}{r}{}
& \multicolumn{1}{r}{} & 22.35 & 39.95 & 65.20 & 85.80 \\ 
\textbf{30} & -5.76 & -4.15 & -3.08 & -2.37 & \multicolumn{1}{r}{} & 
\multicolumn{1}{r}{} & 8.65 & 5.89 & 4.39 & 3.41 & \multicolumn{1}{r}{} & 
\multicolumn{1}{r}{} & 34.85 & 33.45 & 32.15 & 33.35 & \multicolumn{1}{r}{}
& \multicolumn{1}{r}{} & 27.95 & 48.60 & 77.05 & 93.55 \\ 
\textbf{40} & -6.29 & -4.24 & -3.12 & -2.44 & \multicolumn{1}{r}{} & 
\multicolumn{1}{r}{} & 8.67 & 5.82 & 4.27 & 3.35 & \multicolumn{1}{r}{} & 
\multicolumn{1}{r}{} & 42.95 & 41.85 & 38.60 & 38.80 & \multicolumn{1}{r}{}
& \multicolumn{1}{r}{} & 26.95 & 52.35 & 82.15 & 97.05 \\ 
\textbf{50} & -6.12 & -4.19 & -3.09 & -2.47 & \multicolumn{1}{r}{} & 
\multicolumn{1}{r}{} & 8.46 & 5.62 & 4.16 & 3.27 & \multicolumn{1}{r}{} & 
\multicolumn{1}{r}{} & 47.90 & 45.45 & 45.15 & 45.05 & \multicolumn{1}{r}{}
& \multicolumn{1}{r}{} & 31.50 & 59.55 & 88.25 & 98.65 \\ 
& \multicolumn{22}{l}{FMOLS} \\ \hline
\textbf{20} & -10.89 & -7.25 & -5.42 & -4.20 & \multicolumn{1}{r}{} & 
\multicolumn{1}{r}{} & 13.27 & 9.09 & 6.79 & 5.37 & \multicolumn{1}{r}{} & 
\multicolumn{1}{r}{} & 85.10 & 76.30 & 68.35 & 63.05 & \multicolumn{1}{r}{}
& \multicolumn{1}{r}{} & 54.40 & 54.40 & 73.30 & 88.35 \\ 
\textbf{30} & -10.32 & -7.26 & -5.29 & -4.06 & \multicolumn{1}{r}{} & 
\multicolumn{1}{r}{} & 12.28 & 8.70 & 6.45 & 5.04 & \multicolumn{1}{r}{} & 
\multicolumn{1}{r}{} & 88.60 & 82.70 & 76.00 & 71.35 & \multicolumn{1}{r}{}
& \multicolumn{1}{r}{} & 57.00 & 60.70 & 82.20 & 95.25 \\ 
\textbf{40} & -10.94 & -7.60 & -5.45 & -4.29 & \multicolumn{1}{r}{} & 
\multicolumn{1}{r}{} & 12.57 & 8.76 & 6.41 & 5.10 & \multicolumn{1}{r}{} & 
\multicolumn{1}{r}{} & 91.90 & 88.20 & 82.85 & 77.50 & \multicolumn{1}{r}{}
& \multicolumn{1}{r}{} & 60.10 & 61.75 & 84.95 & 96.65 \\ 
\textbf{50} & -10.72 & -7.42 & -5.38 & -4.30 & \multicolumn{1}{r}{} & 
\multicolumn{1}{r}{} & 12.26 & 8.52 & 6.28 & 5.00 & \multicolumn{1}{r}{} & 
\multicolumn{1}{r}{} & 93.50 & 89.15 & 85.40 & 82.35 & \multicolumn{1}{r}{}
& \multicolumn{1}{r}{} & 63.55 & 68.15 & 89.30 & 98.80 \\ \hline\hline
\end{tabular}%
\vspace{-0.2in}
\end{center}

\begin{flushleft}
\footnotesize%
\singlespacing%
Notes: DGP is given by $\Delta y_{it}=c_{i}-\alpha _{i}\left(
y_{i,t-1}-\beta x_{i,t-1}\right) +u_{y,it}$ and $\Delta x_{it}=u_{x,it}$,
for $i=1,2,...,n,$ $T=1,2,...,T$, with $\beta =1$ and $\alpha _{i}\sim IIDU%
\left[ 0.2,0.3\right] $. Errors $u_{y,it}$, $u_{x,it}$ are cross-sectionally
dependent, heteroskedastic over $i$, and also correlated over $y$ \& $x$
equations. See Section \ref{MCD} for complete description of the DGP. The
pooled Bewley estimator is given by (\ref{pmg}), with variance estimated
using (\ref{ve}). PMG is the Pooled Mean Group estimator proposed by 
\citeN{PesaranShinSmith1999}%
. PDOLS is panel dynamic OLS estimator by 
\citeN{MarkSul2003}%
. FMOLS is the group-mean fully modified OLS estimator by 
\citeANP{Pedroni1996} (\citeyearNP{Pedroni1996}, \citeyearNP{Pedroni2001ReStat})%
. The size and power findings are computed using 5\% nominal level and the
reported power is the rejection frequency for testing the hypothesis $\beta
=0.9$. 
\normalsize%
\pagebreak
\end{flushleft}

\begin{center}
\textbf{Table B2: }MC findings for the estimation of long-run coefficient $%
\beta $ in experiments with cross-sectionally dependent errors.

Bias corrected estimators and inference conducted using bootstrapped
critical values.

\bigskip

\footnotesize%

\renewcommand{\arraystretch}{1.0}\setlength{\tabcolsep}{3pt}%
\scriptsize%
\begin{tabular}{rrrrrccrrrrccrrrrccrrrr}
\hline\hline
& \multicolumn{4}{c}{\textbf{Bias (}$\mathbf{\times }$\textbf{\ 100)}} &  & 
& \multicolumn{4}{c}{\textbf{RMSE (}$\mathbf{\times }$\textbf{\ 100)}} &  & 
& \multicolumn{4}{c}{\textbf{Size (5\% level)}} &  &  & \multicolumn{4}{c}{%
\textbf{Power (5\% level)}} \\ 
\cline{2-5}\cline{8-11}\cline{14-17}\cline{20-23}
$n\backslash T$ & \textbf{20} & \textbf{30} & \textbf{40} & \textbf{50} & 
\multicolumn{1}{r}{} & \multicolumn{1}{r}{} & \textbf{20} & \textbf{30} & 
\textbf{40} & \textbf{50} & \multicolumn{1}{r}{} & \multicolumn{1}{r}{} & 
\textbf{20} & \textbf{30} & \textbf{40} & \textbf{50} & \multicolumn{1}{r}{}
& \multicolumn{1}{r}{} & \textbf{20} & \textbf{30} & \textbf{40} & \textbf{50%
} \\ \hline
& \multicolumn{22}{l}{\textbf{Jackknife bias-corrected estimators}} \\ \hline
& \multicolumn{22}{l}{PB} \\ \hline
\textbf{20} & -1.57 & -0.54 & -0.11 & -0.01 & \multicolumn{1}{r}{} & 
\multicolumn{1}{r}{} & 7.83 & 5.24 & 3.92 & 3.04 & \multicolumn{1}{r}{} & 
\multicolumn{1}{r}{} & 7.65 & 6.55 & 5.55 & 4.65 & \multicolumn{1}{r}{} & 
\multicolumn{1}{r}{} & 20.55 & 41.35 & 65.35 & 83.50 \\ 
\textbf{30} & -1.29 & -0.67 & -0.34 & -0.05 & \multicolumn{1}{r}{} & 
\multicolumn{1}{r}{} & 6.76 & 4.63 & 3.59 & 2.72 & \multicolumn{1}{r}{} & 
\multicolumn{1}{r}{} & 6.30 & 6.60 & 6.40 & 5.25 & \multicolumn{1}{r}{} & 
\multicolumn{1}{r}{} & 25.80 & 51.35 & 73.90 & 91.45 \\ 
\textbf{40} & -1.67 & -0.62 & -0.20 & -0.06 & \multicolumn{1}{r}{} & 
\multicolumn{1}{r}{} & 6.48 & 4.35 & 3.36 & 2.58 & \multicolumn{1}{r}{} & 
\multicolumn{1}{r}{} & 6.75 & 6.15 & 6.20 & 5.40 & \multicolumn{1}{r}{} & 
\multicolumn{1}{r}{} & 25.55 & 54.30 & 80.25 & 94.45 \\ 
\textbf{50} & -1.54 & -0.58 & -0.21 & -0.09 & \multicolumn{1}{r}{} & 
\multicolumn{1}{r}{} & 6.16 & 4.12 & 3.15 & 2.43 & \multicolumn{1}{r}{} & 
\multicolumn{1}{r}{} & 6.80 & 6.20 & 5.85 & 4.65 & \multicolumn{1}{r}{} & 
\multicolumn{1}{r}{} & 27.65 & 59.40 & 84.95 & 96.70 \\ \hline
& \multicolumn{22}{l}{PMG} \\ \hline
\textbf{20} & -0.54 & -0.15 & 0.02 & 0.06 & \multicolumn{1}{r}{} & 
\multicolumn{1}{r}{} & 10.66 & 6.26 & 4.45 & 3.44 & \multicolumn{1}{r}{} & 
\multicolumn{1}{r}{} & 16.05 & 10.80 & 8.55 & 7.90 & \multicolumn{1}{r}{} & 
\multicolumn{1}{r}{} & 29.10 & 45.05 & 69.15 & 84.75 \\ 
\textbf{30} & -0.25 & -0.39 & -0.19 & 0.03 & \multicolumn{1}{r}{} & 
\multicolumn{1}{r}{} & 9.10 & 5.50 & 4.02 & 2.94 & \multicolumn{1}{r}{} & 
\multicolumn{1}{r}{} & 14.80 & 10.40 & 9.05 & 6.90 & \multicolumn{1}{r}{} & 
\multicolumn{1}{r}{} & 36.30 & 53.85 & 76.40 & 93.40 \\ 
\textbf{40} & -0.57 & -0.27 & -0.08 & 0.03 & \multicolumn{1}{r}{} & 
\multicolumn{1}{r}{} & 8.25 & 5.13 & 3.78 & 2.78 & \multicolumn{1}{r}{} & 
\multicolumn{1}{r}{} & 14.10 & 10.50 & 9.15 & 7.80 & \multicolumn{1}{r}{} & 
\multicolumn{1}{r}{} & 35.10 & 59.50 & 81.25 & 95.55 \\ 
\textbf{50} & -0.41 & -0.10 & -0.06 & 0.03 & \multicolumn{1}{r}{} & 
\multicolumn{1}{r}{} & 7.81 & 4.81 & 3.44 & 2.63 & \multicolumn{1}{r}{} & 
\multicolumn{1}{r}{} & 14.85 & 11.05 & 8.65 & 7.65 & \multicolumn{1}{r}{} & 
\multicolumn{1}{r}{} & 39.55 & 64.65 & 85.90 & 97.15 \\ \cline{2-23}
& \multicolumn{22}{l}{PDOLS} \\ \hline
\textbf{20} & -4.61 & -2.84 & -2.06 & -1.54 & \multicolumn{1}{r}{} & 
\multicolumn{1}{r}{} & 9.89 & 6.23 & 4.51 & 3.51 & \multicolumn{1}{r}{} & 
\multicolumn{1}{r}{} & 8.75 & 7.50 & 6.45 & 4.95 & \multicolumn{1}{r}{} & 
\multicolumn{1}{r}{} & 10.05 & 21.55 & 39.00 & 63.80 \\ 
\textbf{30} & -4.09 & -2.96 & -2.15 & -1.55 & \multicolumn{1}{r}{} & 
\multicolumn{1}{r}{} & 8.86 & 5.72 & 4.21 & 3.20 & \multicolumn{1}{r}{} & 
\multicolumn{1}{r}{} & 8.70 & 7.30 & 6.55 & 5.40 & \multicolumn{1}{r}{} & 
\multicolumn{1}{r}{} & 12.20 & 24.05 & 46.10 & 73.10 \\ 
\textbf{40} & -4.66 & -2.94 & -2.13 & -1.57 & \multicolumn{1}{r}{} & 
\multicolumn{1}{r}{} & 8.62 & 5.53 & 4.00 & 3.06 & \multicolumn{1}{r}{} & 
\multicolumn{1}{r}{} & 8.80 & 6.55 & 6.45 & 4.55 & \multicolumn{1}{r}{} & 
\multicolumn{1}{r}{} & 10.00 & 23.90 & 47.70 & 77.80 \\ 
\textbf{50} & -4.41 & -2.91 & -2.11 & -1.63 & \multicolumn{1}{r}{} & 
\multicolumn{1}{r}{} & 8.30 & 5.29 & 3.84 & 2.95 & \multicolumn{1}{r}{} & 
\multicolumn{1}{r}{} & 9.85 & 6.50 & 6.25 & 5.10 & \multicolumn{1}{r}{} & 
\multicolumn{1}{r}{} & 11.95 & 24.40 & 52.10 & 80.80 \\ 
& \multicolumn{22}{l}{FMOLS} \\ \hline
\textbf{20} & -8.65 & -5.12 & -3.61 & -2.69 & \multicolumn{1}{r}{} & 
\multicolumn{1}{r}{} & 12.18 & 7.99 & 5.79 & 4.56 & \multicolumn{1}{r}{} & 
\multicolumn{1}{r}{} & 11.25 & 5.65 & 4.20 & 3.85 & \multicolumn{1}{r}{} & 
\multicolumn{1}{r}{} & 2.00 & 4.55 & 9.50 & 25.50 \\ 
\textbf{30} & -8.06 & -5.27 & -3.57 & -2.55 & \multicolumn{1}{r}{} & 
\multicolumn{1}{r}{} & 11.00 & 7.43 & 5.42 & 4.14 & \multicolumn{1}{r}{} & 
\multicolumn{1}{r}{} & 9.65 & 6.85 & 5.20 & 3.35 & \multicolumn{1}{r}{} & 
\multicolumn{1}{r}{} & 2.10 & 3.95 & 12.65 & 31.60 \\ 
\textbf{40} & -8.65 & -5.53 & -3.70 & -2.74 & \multicolumn{1}{r}{} & 
\multicolumn{1}{r}{} & 11.08 & 7.35 & 5.25 & 4.08 & \multicolumn{1}{r}{} & 
\multicolumn{1}{r}{} & 11.50 & 6.50 & 4.55 & 3.85 & \multicolumn{1}{r}{} & 
\multicolumn{1}{r}{} & 1.45 & 3.35 & 12.25 & 32.05 \\ 
\textbf{50} & -8.44 & -5.36 & -3.63 & -2.78 & \multicolumn{1}{r}{} & 
\multicolumn{1}{r}{} & 10.72 & 7.07 & 5.07 & 3.94 & \multicolumn{1}{r}{} & 
\multicolumn{1}{r}{} & 12.45 & 6.70 & 4.55 & 3.40 & \multicolumn{1}{r}{} & 
\multicolumn{1}{r}{} & 1.75 & 4.40 & 13.70 & 34.30 \\ \hline
& \multicolumn{22}{l}{\textbf{Bootstrap bias-corrected estimators}} \\ \hline
& \multicolumn{22}{l}{PB} \\ \hline
\textbf{20} & -1.26 & -0.43 & -0.05 & -0.01 & \multicolumn{1}{r}{} & 
\multicolumn{1}{r}{} & 7.34 & 4.80 & 3.56 & 2.78 & \multicolumn{1}{r}{} & 
\multicolumn{1}{r}{} & 10.10 & 7.95 & 6.45 & 5.85 & \multicolumn{1}{r}{} & 
\multicolumn{1}{r}{} & 29.35 & 52.15 & 76.10 & 90.15 \\ 
\textbf{30} & -1.04 & -0.51 & -0.23 & -0.05 & \multicolumn{1}{r}{} & 
\multicolumn{1}{r}{} & 6.29 & 4.21 & 3.24 & 2.42 & \multicolumn{1}{r}{} & 
\multicolumn{1}{r}{} & 10.20 & 7.35 & 7.20 & 5.40 & \multicolumn{1}{r}{} & 
\multicolumn{1}{r}{} & 37.20 & 64.30 & 85.00 & 95.90 \\ 
\textbf{40} & -1.31 & -0.51 & -0.12 & -0.08 & \multicolumn{1}{r}{} & 
\multicolumn{1}{r}{} & 5.98 & 3.93 & 3.01 & 2.32 & \multicolumn{1}{r}{} & 
\multicolumn{1}{r}{} & 10.15 & 7.50 & 6.75 & 5.80 & \multicolumn{1}{r}{} & 
\multicolumn{1}{r}{} & 39.00 & 69.80 & 89.80 & 98.15 \\ 
\textbf{50} & -1.26 & -0.45 & -0.12 & -0.09 & \multicolumn{1}{r}{} & 
\multicolumn{1}{r}{} & 5.78 & 3.72 & 2.78 & 2.17 & \multicolumn{1}{r}{} & 
\multicolumn{1}{r}{} & 10.20 & 7.85 & 7.05 & 6.00 & \multicolumn{1}{r}{} & 
\multicolumn{1}{r}{} & 42.90 & 75.25 & 93.50 & 99.20 \\ \hline
& \multicolumn{22}{l}{PMG} \\ \hline
\textbf{20} & -1.23 & -0.44 & -0.09 & -0.03 & \multicolumn{1}{r}{} & 
\multicolumn{1}{r}{} & 9.24 & 5.54 & 3.91 & 3.03 & \multicolumn{1}{r}{} & 
\multicolumn{1}{r}{} & 16.00 & 10.50 & 7.60 & 6.70 & \multicolumn{1}{r}{} & 
\multicolumn{1}{r}{} & 31.35 & 51.80 & 76.10 & 90.80 \\ 
\textbf{30} & -0.92 & -0.60 & -0.27 & -0.10 & \multicolumn{1}{r}{} & 
\multicolumn{1}{r}{} & 7.89 & 4.84 & 3.57 & 2.59 & \multicolumn{1}{r}{} & 
\multicolumn{1}{r}{} & 14.70 & 9.95 & 8.45 & 5.85 & \multicolumn{1}{r}{} & 
\multicolumn{1}{r}{} & 38.20 & 60.85 & 83.80 & 97.10 \\ 
\textbf{40} & -1.10 & -0.56 & -0.19 & -0.11 & \multicolumn{1}{r}{} & 
\multicolumn{1}{r}{} & 7.14 & 4.51 & 3.30 & 2.45 & \multicolumn{1}{r}{} & 
\multicolumn{1}{r}{} & 14.80 & 10.00 & 8.90 & 7.00 & \multicolumn{1}{r}{} & 
\multicolumn{1}{r}{} & 39.35 & 67.05 & 88.90 & 98.10 \\ 
\textbf{50} & -1.06 & -0.45 & -0.16 & -0.10 & \multicolumn{1}{r}{} & 
\multicolumn{1}{r}{} & 6.77 & 4.23 & 3.02 & 2.29 & \multicolumn{1}{r}{} & 
\multicolumn{1}{r}{} & 14.90 & 10.45 & 8.25 & 6.00 & \multicolumn{1}{r}{} & 
\multicolumn{1}{r}{} & 44.90 & 71.40 & 91.75 & 99.15 \\ \hline
& \multicolumn{22}{l}{PDOLS} \\ \hline
\textbf{20} & -2.34 & -1.13 & -0.66 & -0.37 & \multicolumn{1}{r}{} & 
\multicolumn{1}{r}{} & 8.73 & 5.46 & 3.92 & 3.05 & \multicolumn{1}{r}{} & 
\multicolumn{1}{r}{} & 12.05 & 10.05 & 7.75 & 6.90 & \multicolumn{1}{r}{} & 
\multicolumn{1}{r}{} & 22.40 & 45.25 & 71.30 & 87.75 \\ 
\textbf{30} & -2.03 & -1.29 & -0.74 & -0.41 & \multicolumn{1}{r}{} & 
\multicolumn{1}{r}{} & 7.73 & 4.85 & 3.53 & 2.68 & \multicolumn{1}{r}{} & 
\multicolumn{1}{r}{} & 12.15 & 8.60 & 8.45 & 7.65 & \multicolumn{1}{r}{} & 
\multicolumn{1}{r}{} & 26.95 & 51.50 & 79.30 & 94.15 \\ 
\textbf{40} & -2.46 & -1.28 & -0.71 & -0.41 & \multicolumn{1}{r}{} & 
\multicolumn{1}{r}{} & 7.40 & 4.68 & 3.31 & 2.54 & \multicolumn{1}{r}{} & 
\multicolumn{1}{r}{} & 12.65 & 10.20 & 9.05 & 7.75 & \multicolumn{1}{r}{} & 
\multicolumn{1}{r}{} & 26.15 & 55.55 & 84.05 & 96.90 \\ 
\textbf{50} & -2.32 & -1.22 & -0.68 & -0.44 & \multicolumn{1}{r}{} & 
\multicolumn{1}{r}{} & 7.28 & 4.41 & 3.15 & 2.39 & \multicolumn{1}{r}{} & 
\multicolumn{1}{r}{} & 15.00 & 9.60 & 8.40 & 7.40 & \multicolumn{1}{r}{} & 
\multicolumn{1}{r}{} & 30.15 & 59.85 & 88.40 & 98.55 \\ \hline
& \multicolumn{22}{l}{FMOLS} \\ \hline
\textbf{20} & -4.39 & -2.01 & -1.24 & -0.77 & \multicolumn{1}{r}{} & 
\multicolumn{1}{r}{} & 10.37 & 6.73 & 4.85 & 3.83 & \multicolumn{1}{r}{} & 
\multicolumn{1}{r}{} & 18.45 & 10.90 & 7.80 & 6.85 & \multicolumn{1}{r}{} & 
\multicolumn{1}{r}{} & 17.80 & 31.75 & 50.55 & 70.75 \\ 
\textbf{30} & -4.00 & -2.25 & -1.25 & -0.71 & \multicolumn{1}{r}{} & 
\multicolumn{1}{r}{} & 9.02 & 6.02 & 4.42 & 3.43 & \multicolumn{1}{r}{} & 
\multicolumn{1}{r}{} & 17.25 & 12.00 & 9.60 & 7.70 & \multicolumn{1}{r}{} & 
\multicolumn{1}{r}{} & 20.95 & 37.25 & 61.40 & 82.55 \\ 
\textbf{40} & -4.46 & -2.44 & -1.27 & -0.84 & \multicolumn{1}{r}{} & 
\multicolumn{1}{r}{} & 8.82 & 5.67 & 4.11 & 3.21 & \multicolumn{1}{r}{} & 
\multicolumn{1}{r}{} & 19.35 & 12.10 & 9.60 & 9.05 & \multicolumn{1}{r}{} & 
\multicolumn{1}{r}{} & 20.10 & 38.50 & 67.65 & 85.95 \\ 
\textbf{50} & -4.30 & -2.27 & -1.22 & -0.86 & \multicolumn{1}{r}{} & 
\multicolumn{1}{r}{} & 8.48 & 5.41 & 3.92 & 3.02 & \multicolumn{1}{r}{} & 
\multicolumn{1}{r}{} & 21.20 & 13.90 & 10.10 & 7.80 & \multicolumn{1}{r}{} & 
\multicolumn{1}{r}{} & 22.85 & 44.85 & 72.20 & 90.40 \\ \hline\hline
\end{tabular}%
\vspace{-0.2in}
\end{center}

\begin{flushleft}
\footnotesize%
\singlespacing%
Notes: See the notes to Table B1. Bias-corrected versions of the PB
estimator are described in Subsection \ref{BM}. Bias-corrected versions of
the PMG, PDOLS and FMOLS estimator are described in Appendix B. Inference is
conducted using bootstrapped critical values. 
\normalsize%
\end{flushleft}

\end{document}